\documentclass[11pt]{article}
\usepackage{amssymb, bm, bbm, subfigure}
\usepackage{amsmath}
\usepackage{amsfonts}
\usepackage{latexsym}
\usepackage{hyperref}
\usepackage{graphicx}
\usepackage{amsthm}
\usepackage{color}

\usepackage{tikz}
\usetikzlibrary{arrows,shapes,backgrounds,decorations.text}

\usepackage{bbm}
\usepackage{numprint}
\usepackage{footnote} 
\makesavenoteenv{tabular}
\makesavenoteenv{table}

\graphicspath{	{./figure/}   }

\newlength{\figtriml}
\newlength{\figtrimb}
\newlength{\figtrimr}
\newlength{\figtrimt}
\newlength{\figwidth}
\newlength{\figwidthcap}
\newlength{\figwidthduo}
\newlength{\figtrimla}
\newlength{\figtrimba}
\newlength{\figtrimra}
\newlength{\figtrimta}
\newtheorem{lemma}{Lemma}

\newtheorem{proposition}{Proposition}

\theoremstyle{definition}
\newtheorem{theorem}{Theorem}

\newtheorem{example}{Example}

\newcommand{\ds}{\displaystyle}

\newcommand{\ba}{\begin{array}}
\newcommand{\ea}{\end{array}}

\newcommand{\mb}{\boldsymbol}
\renewcommand{\l}{\left}\renewcommand{\r}{\right}
\newcommand{\be}{\begin{equation}}
\newcommand{\ee}{\end{equation}}
\newcommand{\eps}{\varepsilon}
\newcommand{\ups}{\upsilon}

\newcommand{\mc}{\mathcal}

\newcommand{\ov}{\overline}
\newcommand{\ul}{\underline}

\newcommand{\1}{\mathbbm{1}}
\newcommand{\E}{\mathbb{E}}

\newcommand{\R}{\mathbb{R}}

\newcommand{\pdkrx}{p_{d,k,r,s}}

\renewcommand{\P}{\mathbb{P}}

\newcommand{\de}{\mathrm{d}}
\newcommand{\ora}{\overrightarrow}

\newcommand{\se}{\text{ if }}

\def\E{\mathbb{E}}
\def\R{\mathbb{R}}

\def\P{\mathbb{P}}

\title{Threshold models of cascades in large-scale networks}
\author{Giacomo Como\footnote{Dept. of Automatic Control, Lund University, 22363 Lund, Sweden, \url{giacomo.como@control.lth.se}}, Wilbert Samuel Rossi\footnote{Dept. of Applied Mathematics, University of Twente, 7500 AE Enschede, The Netherlands, \newline\url{w.s.rossi@utwente.nl}}, and Fabio Fagnani\footnote{Dept. of Mathematical Sciences, Politecnico di Torino, 10129 Torino, Italy, \url{fabio.fagnani@polito.it}}}

\begin{document}\maketitle

\begin{abstract} 
The spread of new beliefs, behaviors, conventions, norms, and technologies in social and economic networks are often driven by cascading mechanisms, and so are contagion dynamics in financial networks. Global behaviors generally emerge from the interplay between the structure of the interconnection topology and the local agents' interactions. We focus on the Linear Threshold Model (LTM) of cascades first introduced by Granovetter (1978). This can be interpreted as the best response dynamics in a network game whereby agents choose strategically between two actions and their payoff is an increasing function of the number of their neighbors choosing the same action. Each agent is equipped with an individual threshold representing the number of her neighbors who must have adopted a certain action for that to become the agent's best response.  
We analyze the LTM dynamics on large-scale networks with heterogeneous agents.
Through a local mean-field approach, we obtain a nonlinear, one-dimensional, recursive equation that approximates the evolution of the LTM dynamics on most of the networks of a given size and distribution of degrees and thresholds. Specifically, we prove that, on all but a fraction of networks with given degree and threshold statistics that is vanishing as the network size grows large, the actual fraction of adopters of a given action in the LTM dynamics is arbitrarily close to the output of the aforementioned recursion. 
We then analyze the dynamic behavior of this recursion and its bifurcations from a dynamical systems viewpoint. 
Applications of our findings to some real network testbeds show good adherence of the theoretical predictions to numerical simulations.  
\end{abstract}

\section{Introduction}\label{sec:intro}
Cascading phenomena permeate the dynamics of social and economic networks. Notable examples are the adoption of new technologies and social norms, the spread of fads and behaviors,  participation to riots \cite{mG:1978,moJ:2008,EK:2010}. 
Such phenomena have been largely recognized to spread through networks of individual interactions \cite{tcS:1978,mG:1978,Blume:1993,Kandori.ea:1993,Ellison:1993,Rogers:1995,PeytonYoung:1998}.  However, in contrast to standard network epidemic models based on pairwise contact mechanisms \cite{Liggettbook2,rD:2007,mN:2003} ---whereby diffusion of a new state occurs independently on the links among the agents--- complex neighborhood effects ---whereby the propensity of an agent to adopt a new state grows nonlinearly with the fraction of adopters among her neighbors--- play a central role in the mechanisms underlying such cascading phenomena \cite{fVR:2007,MS:2010}.


One of the most studied models of cascading mechanisms capturing such complex neighborhood effects is Granovetter's Linear Threshold Model (LTM) \cite{mG:1978}. 
 Granovetter's original work \cite{mG:1978} is concerned with a fully mixed population of $n$ interacting agents, each holding a binary state $Z_i(t)=0,1$, for $i=1,\ldots,n$, and updating it at every discrete time instant $t=0,1,\ldots$ according to the following threshold rule:  $Z_i(t+1)=1$ if the current fraction of state-$1$ adopters in the population is not less than a certain value $\Theta_i$, i.e., if $\frac1n\sum_{j=1}^nZ_j(t)\ge\Theta_i$ and $Z_i(t+1)=0$ otherwise, i.e., if $\frac1n\sum_{j=1}^nZ_j(t)<\Theta_i$. 
Here $\Theta_i\in[0,1]$ is a normalized threshold value that measures the reluctance of agent $i$  in choosing state $1$, equivalently, her propensity to choose state $0$. 
In more realistic scenarios, the population is not fully mixed and agents interact on an interconnection network that can be represented as a, generally directed, graph $\mc G=(\mc V,\mc E)$ whose node set $\mc V=\{1,2,\ldots,n\}$ is identified with the set of agents themselves and where the presence of a link $(i,j)\in\mc E$ represents the fact that agent $i$ observes agent $j$ and gets directly influenced by her state. In this setting, the LTM dynamics reads as follows: 
\be\label{LTMdef}Z_i(t+1)=\left\{\ba{lcl}1&\se&\sum_{j:(i,j)\in\mc E}Z_j(t)\ge\Theta_ik_i\\[10pt]
0&\se&\sum_{j:(i,j)\in\mc E}Z_j(t)<\Theta_ik_i\,,\ea\right.\ee
 where $k_i$ stands for node $i$'s out-degree, see, e.g., \cite{sM:2000,ADO:2012}.
 This can be interpreted as the best response dynamics in a network  game whereby agents choose strategically between two actions, $0$ and $1$, and their payoff is an increasing function of the number of their neighbors choosing the same action. 
A variant of the LTM, that is referred to as Progressive Linear Threshold Model (PLTM) allows for state transitions from $0$ to $1$ only, but not from $1$ to $0$, so that when an agent adopts state $1$, she keeps it ever after \cite{KKT:2003,Centola.ea:2007,ADL:2009,hA:2010,ACM:2013}.

As illustrated in \cite{mG:1978}, there is a simple way to analyze the LTM in fully mixed populations. If one denotes by $z(t):=\frac1n\sum_{i}Z_i(t)$ the fraction of state-$1$ adopters at time $t$, and if 
$F(\theta):=\frac1n|\{i:\,\Theta_i\le\theta\}|$, for $0\le\theta\le 1$,
stands for the cumulative distribution function of the normalized thresholds,  
then \be\label{LTM-recursion-fullymixed}z(t+1)=F(z(t))\,, \qquad t\ge0\,.\ee 
Hence, the evolution of the fraction of state-$1$ adopters in the population can be determined by the above one-dimesional  non-linear recursion. 
This is a dramatic reduction of complexity with respect to the original LTM dynamics whose discrete state space has cardinality $2^n$ growing exponentially fast in the population size.
In fact, an analogous result can be verified to hold true for the PLTM, provided that agents with initial state $Z_i(0)=1$ are considered as if having threshold $0$, which is consistent with the fact they will always keep their state equal to $1$. More precisely, if one introduces the distribution function 
$\tilde F(\theta)=\frac1n|\{i:\,\Theta_i(1-Z_i(0))\le\theta\}|$
then the fraction $z(t)$ of state-$1$ adopters in the PLTM satisfies the recursion\footnote{Formally, the result follows from Lemma \ref{lemma:LTM=PLTM} in Section \ref{sec:model}.} 
$z(t+1)=\tilde F(z(t))$. 

In the more complex case where the population is not fully mixed but rather interacts along a given graph $\mc G=(\mc V,\mc E)$, the simple recursion \eqref{LTM-recursion-fullymixed} does not hold true any longer for the fraction of state-$1$ adopters $z(t)$ in the LTM \eqref{LTMdef}. In fact, for undirected (possibly infinite) graphs $\mc G$ and homogeneous normalized thresholds $\Theta_i=\theta$, Morris \cite{sM:2000} characterizes the fixed points of the LTM dynamics as those configurations in $\{0,1\}^{n}$ whose support $\mc U\subseteq\mc V$ is a $\theta$-cohesive subset of $\mc V$ with $(1-\theta)$-cohesive complement $\mc V\setminus\mc U$, meaning that all nodes in $\mc U$ have at least a fraction $\theta$ of neighbors in $\mc U$ and all nodes in $\mc V\setminus\mc U$ have less than a fraction $\theta$ of neighbors in $\mc U$. While such a characterization provides fundamental insight into the structure of the equilibria of the LTM, finding $\theta$-cohesive subsets of nodes with $(1-\theta)$-cohesive complement in an arbitrary graph $\mc G$ is a computationally hard problem. 
Computational complexity issues also arise in the PLTM dynamics, for which, e.g., Kempe, Kleinberg, and Tardos \cite{KKT:2003} prove NP-hardness of the selection problem of the $k$ `most influential' nodes, i.e., the choice of the cardinality-$k$ subset of nodes that, if initiated as state-$1$ adopters, lead to the largest set of final state-$1$ adopters. Building on submodularity properties of the number of final state-$1$ adopters as a function of the set of initial state-$1$ adopters, provable approximation guarantees are then provided in \cite{KKT:2003} for the $k$ `most influential' nodes  selection problem. Such `most influential' nodes  selection problem has attracted a large amount of attention recently, see, e.g., \cite{CYZ:2010,GLS:2011}. Asymptotic analysis of the LTM dynamics and associated complexity issues have also been addressed in \cite{ADO:2012}.


As the aforementioned results point out, analysis and optimization of the LTM and of the PLTM on general networks is typically a hard problem. On the other hand, in practical large-scale applications, complete information on the network structure and on the specific threshold configuration might not be available, while only aggregate statistics such as degree and threshold distributions might be known. With this motivation in mind, the present paper deals with the analysis of the LTM and of the PLTM dynamics on the {\it ensemble} of all graphs with a given joint degree/threshold distribution (formally we will consider the so-called \emph{configuration model} \cite{BollobasRG,rD:2007} of interconnections), rather than on a specific graph $\mc G$. Our main result shows that for all but a vanishingly small (as the network size $n$ grows large) fraction of networks from the configuration model  ensemble of given joint degree-threshold distribution, the fraction $z(t)$ of state-$1$ adopters in the LTM dynamics can be approximated, to an arbitrary small tolerance level, by the solution of the recursion 
\be\label{rec}x(t+1)=\phi(x(t))\,,\qquad y(t+1)=\psi(x(t))\,,\ee where $\phi(x)$ and $\psi(x)$ are suitably defined polynomial functions that map the interval $[0,1]$ in itself, whose form depends only on the joint degree-threshold distribution (see \eqref{psi-def} and \eqref{phi-def}). An analogous result for the PLTM is proved as well, provided that agents with initial state $Z_i(0)=1$ are treated as if having threshold $0$, equivalently, that the functions $\phi(x)$ and $\psi(x)$ are defined based on the joint distribution of node degrees and the product $(1-\Theta_i)Z_i(0)$. 

Our results should be compared to the literature on the analysis of the LTM or the PLTM on large-scale random networks with given degree distribution. The papers \cite{mL:2009,hA:2010,mL:2012}  all study the asymptotic behavior of the PLTM in random undirected networks. In particular, \cite{mL:2009} focuses on the asymptotic effect of two vaccination strategies equivalent to the {\it a priori} removal of nodes, whereas
\cite{hA:2010} and \cite{mL:2012} both  rigorously provide conditions, in the large-scale limit, for the PLTM contagion to eventually reach a sizeable fraction of nodes when started from a single node or a fraction of nodes that is sublinear in $n$. 
The paper \cite{ACM:2013} presents analogous results for a version of the PLTM on random weighted directed networks, proposed as a model for cascading failures in financial networks. In contrast with those results, ours are concerned with approximation of the dynamics rather than with the asymptotics of the fraction of state-$1$ adopters. The other major difference is that they are not limited to the PLTM but cover also the original LTM on the directed configuration model ensemble of networks. On the other hand, it should be stressed that our results do not extend to the analysis of the general LTM on the undirected configuration model ensemble. In fact, as pointed out in \cite{KM:2011}, the analysis of the LTM on undirected trees presents itself additional challenges beyond the scope of the approach proposed here. 

In summary, the main contributions of this paper consist in (a) providing a rigorous approximation result in terms of the output $y(t)$ of the recursion \eqref{rec} for the fraction $z(t)$ of state-$1$ adopters in the LTM and the PLTM dynamics on the ensemble of directed networks (Theorem \ref{theo:concentration}) and of the PLTM on the ensemble of undirected networks (Theorem \ref{theo:concentration2}); and (b) analyzing the asymptotic behavior of the recursion \eqref{rec} in both homogeneous (Section \ref{subsect:homogeneous}) and heterogeneous (Sections \ref{subsect:hetero1} and \ref{subsect:hetero2}) networks. Such theoretical results are then supported by numerical simulations on an actual large-scale network topology (see Section \ref{sec:simulations}). In the course of building up the tools for such analysis, we also prove that the PLTM can be regarded as a special case of the LTM (Lemma \ref{lemma:LTM=PLTM}), a result of potential independent interest. 

The rest of this paper is organized as follows. The final part of this section gathers some notational conventions used throughout the paper; Section \ref{sec:model} formally introduces the LTM and the PLTM dynamics, proves some fundamental monotonicity properties (Lemma \ref{lemma:monotonicity}), and builds on them to prove that PLTM can be regarded as a special case of the LTM when all agents with initial state $1$ have threshold $0$ (Lemma \ref{lemma:LTM=PLTM}); in Section \ref{sec:LTM-BR} we introduce the recursion \eqref{rec} by a heuristic argument and then analyze its asymptotic behavior first in homogeneous and then in heterogenous networks; in Section \ref{sec:density-evolution} we formally prove that the output $y(t)$ of the recursion \eqref{rec} provides a good approximation of the fraction of state-$1$ adopters in both the LTM and PLTM dynamics on the ensemble of directed networks (Theorem \ref{theo:concentration}) and in the PLTM dynamics on the ensemble of undirected networks (Theorem \ref{theo:concentration2}); in Section \ref{sec:simulations} we present numerical simulations on an actual large-scale network testbed.

\textbf{Notational conventions}  We denote the transpose of a matrix $M$ by $M'$ and the all-one vector by $\1$.
We model interconnection topologies as directed multi-graphs $\mc G=(\mc V,\mc E)$ where $\mc V=\{1,\ldots,n\}$ is a finite set of nodes representing the interacting agents and $\mc E\subseteq\mc V\times\mc V$ is a multi-set of directed links. Here, the use of the prefix \emph{multi} reflects the fact that links $(i,j)$ directed from the same tail node $i$ to the same head node $j$ may occur with multiplicity larger than $1$, i.e., we allow for the possible presence of parallel links. The adjacency matrix $A\in\R^{n\times n}$ of $\mc G$ has then nonnegative-integer entries $A_{ij}$ whose value represents the multiplicity with which link $(i,j)$ appears in $\mc E$.\footnote{In fact, one could easily relax the integer constraint on the entries of the adjacency matrix $A$ and consider weighted graphs, whereby each positive entry $A_{ij}$ stands for the weight of the link from node $i$ to node $j$. For the sake of simplicity in the exposition we will not consider this generalization explicitly in this paper.} Observe that we also allow for the possibility of selfloops, i.e., links of the form $(i,i)$ that correspond to nonzero diagonal entries $A_{ii}>0$ of the adjacency matrix. Of course, directed graphs with no self-loops can be recovered as a special case when $A$ has binary entries $A_{ij}\in\{0,1\}$ and zero diagonal, whereas undirected graphs can be recovered as a special case when the adjacency matrix $A'=A$ is symmetric. In particular, simple graphs (undirected and with no self-loops) correspond to the case when the adjacency matrix is symmetric and has zero diagonal and binary entries. The in-degree and out-degree vectors of a graph are then denoted by $\delta=A'\1$ and $\kappa=A\1$, respectively, so that $\delta_i=\sum_jA_{ji}$ and $\kappa_i=\sum_jA_{ij}$ are the in- and out-degree, respectively, of node $i$. 
Whenever the interconnection topology contains a link $(i,j)\in\mc E$ we refer to node $j$ as an out-neighbor of $i$ and to node $i$ as an in-neighbor of $j$. An $l$-tuple of nodes $i_0,i_1,\ldots i_l$ is referred to as a length-$l$ \textit{walk} from $i_0$ to $i_l$ if $(i_{h-1},i_h) \in \mc E$ for all $1\le h\le l$. 
Finally, the \textit{depth-$t$ neighborhood} $\mc N_t^i$ of a node $i$ is the subgraph of $\mc G$ containing all the nodes $j$ such that there exists a walk from $i$ to $j$ of length $l\le t$.

\section{The Linear Threshold Model and its progressive version} 
\label{sec:model}
In this section, we introduce the LTM dynamics on arbitrary interconnection networks. We then prove some basic monotonicity properties of the LTM and use them to show how the PLTM can be  recovered as a special case of the LTM with the proper choice of thresholds.

Let $\mc G=(\mc V,\mc E)$ be an interconnection topology. We follow the convention that the link direction is the opposite of the one of the influence, so that the presence of a link $(i,j)\in\mc E$ indicates that agent $i$ observes, and is influenced by, agent $j$. The behavior of each agent $i=1,\ldots,n$ in the LTM dynamics is characterized by a threshold value $\rho_i\in\{0,1,\ldots,\kappa_i\}$ that represents the minimum number of state-$1$ adopters that she needs to observe among her neighbors in order to adopt state $1$ at the next time instant. Such  threshold is related to the normalized threshold $\Theta_i\in[0,1]$ mentioned in Section \ref{sec:intro} by the identity $\rho_i=\lceil\Theta_i\kappa_i\rceil$. The vector of all agents' thresholds is then denoted by $\rho\in\R^n$. In order to introduce the LTM dynamics, we are left to specify an initial state $\sigma_i\in\{0,1\}$ for every agent $i$. Let the vector of all agents' initial states be denoted by $\sigma\in\{0,1\}^n$. We will refer to a network as the $4$-tuple $\mc N=(\mc V,\mc E,\rho,\sigma)$ of a set of agents $\mc V$, a multiset of links $\mc E$, a threshold vector $\rho$, and a vector of initial states $\sigma$.

The LTM on a network $\mc N=(\mc V,\mc E,\rho,\sigma)$ is then defined as the discrete-time dynamical system with state space $\{0,1\}^n$ and update rule
\be\label{LTM-def}
Z_i(0)=\sigma_i\,,\qquad Z_i(t+1)=\left\{\ba{lcl}1&\se&\sum_{j}A_{ij}Z_j(t)\ge\rho_i\\[10pt]
0&\se&\sum_{j}A_{ij}Z_j(t)<\rho_i\ea\r.\qquad i=1,\ldots,n\,,\qquad t\ge0\,.
\ee
In fact, the LTM can be interpreted as the best response dynamics in a network  game \cite{sM:2000, moJ:2008, EK:2010, mL:2012} whereby the agents $i\in\mc V$ choose their action  $Z_i\in\{0,1\}$ so as to maximize their payoff $u_i(Z)=(\kappa_i-\rho_i+\eps)Z_i\sum_jA_{ij}Z_j+\rho_i(1-Z_i)\sum_jA_{ij}(1-Z_j)$ 
where $0<\eps<1$ is introduced in order to break possible ties in favor of the $Z_i=1$ action.  Observe that Granovetter's recursion \eqref{LTMdef} for a fully mixed population can be recovered when the interaction topology is the complete graph with self-loops, i.e., the link set is $\mc E=\mc V\times\mc V$ so that the adjacency matrix is the all-one matrix  $A=\1\1'$, and the thresholds are chosen as $\rho_i=\lceil n\Theta_i\rceil$ for all $i\in\mc V$.

The following lemma captures some basic monotonicity properties of the LTM dynamics that prove particularly useful in its analysis. In stating and proving it we will adopt the notational convention that an inequality between vectors is meant to hold true entry-wise. 

\begin{lemma}\label{lemma:monotonicity}
Let $\mc N=(\mc V,\mc E,\rho,\sigma)$ and $\mc N^+=(\mc V,\mc E,\rho,\sigma^+)$ be two networks differing only (possibly) for the initial state vector. Let $Z(t)$ and $Z^+(t)$ be the state vectors of the LTM dynamics \eqref{LTM-def} on $\mc N$ and $\mc N^+$, respectively. Then, 
\begin{enumerate} 
\item[(i)] if $\sigma^+\ge \sigma$, then $Z^+(t)\ge Z(t)$ for all $t\ge0$;
\item[(ii)] if $\rho_i\le(1-\sigma_i)\kappa_i$ for all $i$, then  $Z(t)$ is non-decreasing  in $t$, hence, in particular, it is eventually constant. 
\end{enumerate}
\end{lemma}
\proof 
\begin{enumerate} 
\item[(i)] Let $A$ be the adjacency matrix of $\mc N$ and $\mc N^+$. Observe that, since $A$ is a nonnegative matrix, if $Z^+(t)\ge Z(t)$ for some $t\ge0$, then $AZ^+(t)\ge AZ(t)$, hence $Z^+(t+1)\ge Z(t+1)$ (because $Z^+_i(t+1)=0$ implies that $\sum_{j}A_{ij}Z_j(t)\le\sum_{j}A_{ij}Z^+_j(t)<\rho_i$ so that $Z_i(t=1)=0$). The claim now follows by induction on $t$.
\item[(ii)] Let $Z(0)=\sigma$ and $Z^+(0)=\sigma^+=Z(1)$. Observe that,  if $\rho_i\le(1-\sigma_i)\kappa_i$ for every $i$, then for all those $i$ such that $Z_i(0)=\sigma_i=1$ one has $\rho_i=0\le\sum_jA_{ij}Z_j(0)$ so that $\sigma^+_i=Z_i(1)=1$. Hence, necessarily $\sigma^+=Z(1)\ge \sigma$. It then follows from point (i) that $Z^+(t)=Z(t+1)\ge Z(t)$ for all $t\ge0$, i.e., $Z(t)$ is non-decreasing, hence eventually constant. 
\qed
\end{enumerate}\medskip

We now introduce a variation of the LTM known as \emph{Progressive} LTM (PLTM), whereby only state transitions from $0$ to $1$ are allowed, but not from $1$ to $0$. Formally, the PLTM on a network $\mc N=(\mc V,\mc E,\rho,\sigma)$ is defined by the following recursive relations 
\be\label{PLTM-def}
\begin{array}{l}Z_i(0)=\sigma_i\,,\qquad Z_i(t+1)=\left\{\ba{lcl}1&\se&\sum_{j}A_{ij}Z_j(t)\ge(1-Z_i(t))\rho_i\\[10pt]
0&\se&\sum_{j}A_{ij}Z_j(t)<(1-Z_i(t))\rho_i\ea\r.\qquad i=1,\ldots,n\,,\ t\ge0\,.
\end{array}
\ee
Observe that in the PLTM dynamics the state update rule of every agent $i$ depends on her own current state, regardless of the presence of self-loops in the network. This is in contrast with the LTM update rule, whereby the new state of every agent $i$ such that $A_{ii}=0$ depends on the current state of its out-neighbors only and not on itself. In spite of these differences,  the following result shows that the PLTM dynamics coincides with the LTM provided that agents with initial state $1$ are treated as if having effective threshold $0$. 

\begin{lemma}\label{lemma:LTM=PLTM}
The PLTM dynamics \eqref{PLTM-def} on a network $\mc N=(\mc V,\mc E,\rho,\sigma)$ coincide with the dynamics defined by
\be\label{PLTM-def1}
\begin{array}{l}
Z_i(0)=\sigma_i\,,\qquad Z_i(t+1)=\left\{\ba{lcl}1&\se&\sum_{j}A_{ij}Z_j(t)\ge(1-\sigma_i)\rho_i\\[10pt]
0&\se&\sum_{j}A_{ij}Z_j(t)<(1-\sigma_i)\rho_i\ea\r.\qquad i=1,\ldots,n\,,\ t\ge0\,.
\end{array}
\ee
In particular, if $\rho_i\le(1-\sigma_i)\kappa_i$ for every $i\in\mc V$, then the LTM dynamics \eqref{LTM-def} and the PLTM dynamics  \eqref{PLTM-def} coincide. 
\end{lemma}
\proof
Let us denote by $Z(t)$ and $\tilde Z(t)$ the state vectors generated by the recursions \eqref{PLTM-def} and \eqref{PLTM-def1}, respectively.
It follows from applying part (ii) of Lemma \ref{lemma:monotonicity} to the network $\tilde{\mc N}=(\mc V,\mc E,\tilde\rho,\sigma)$ where $\tilde\rho_i=\rho_i(1-\sigma_i)$ that $\tilde Z(t)$ is non-decreasing in $t$. On the other hand, $Z(t)$ is non-decreasing by construction, since only transitions from $0$ to $1$ are allowed by \eqref{PLTM-def} but not the other way around. Now, we shall proceed by an induction argument, assuming that $Z(s)=\tilde Z(s)$ for $s=0,1,\dots , t$ and showing that then $Z(t+1)=\tilde Z(t+1)$. For all those $i$ such that $Z_i(t)=\tilde Z_i(t)=0$ monotonicity of $\tilde Z(t)$ implies that $\sigma_i=\tilde Z_i(0)\le\tilde Z_i(t)=0$ and therefore the updates in \eqref{PLTM-def} and in \eqref{PLTM-def1} coincide, yielding $Z_i(t+1)=\tilde Z_i(t+1)$. On the other hand, for all those $i$ such that $Z_i(t)=\tilde Z_i(t)=1$, monotonicity implies that $Z_i(t+1)\ge Z_i(t)=1$ and $\tilde Z_i(t+1)\ge\tilde Z_i(t)=1$ so that $\tilde Z_i(t+1)=Z_i(t+1)$. This proves the first claim.

The second part of the Lemma simply follows from the fact that $\rho_i\le(1-\sigma_i)\kappa_i$ and $\sigma_i\in\{0,1\}$ imply $(1-\sigma_i)\rho_i=\rho_i$.
\qed\medskip

Lemma \ref{lemma:LTM=PLTM} is particularly significant in that it implies that the study of the PLTM dynamics \eqref{PLTM-def} can be reduced to that of a special case of the LTM dynamics \eqref{LTM-def}, where all agents with initial state $\sigma_i=1$ have threshold  $\rho_i=0$.

\section{Recursive equations for networks with given statistics} 
\label{sec:LTM-BR}
As mentioned in Section \ref{sec:intro}, the LTM on a complete network lends itself to a simple analysis enabled by the fact that the fraction of state-$1$ adopters 
 evolves according to the one-dimensional recursion $y(t+1)=F(y(t))$, where $F$ is the cumulative distribution of the normalized thresholds across the population \cite{mG:1978}. While such a one-dimensional recursion does not hold true for the LTM dynamics on general networks, the main contribution of this paper consists in showing that the fraction of state-$1$ adopters in the LTM and the PLTM dynamics on most directed networks can be approximated\footnote{Cf.~Figures \ref{fig:ex-homo-rand-dyn} and \ref{fig:ex-hetero-glob-rand-dyn}.  } ---in a quantitatively precise sense that will be formalized in Section \ref{sec:density-evolution}--- by the output $y(t)$ of another one-dimensional recursion of the form 
\be\label{CM-recursion}x(t+1)=\phi(x(t))\,,\qquad y(t+1)=\psi(x(t))\,,\ee 
where (cf.~\eqref{psi-def} and \eqref{phi-def}) $\phi(x)$ and $\psi(x)$ are polynomials with nonnegative coefficients that depend on the network's statistics $\mb p$ defined below. 
In this section, we introduce the specific form of the recursion \eqref{CM-recursion} and analyze its dynamical behavior, while postponing to Section \ref{sec:density-evolution} the formal proof that the output $y(t)$ of \eqref{CM-recursion} provides an effective approximation of the fraction of state-$1$ adopters in the LTM dynamics. 

Throughout, we will use the following notation. For a network $\mc N=(\mc V,\mc E,\rho,\sigma)$ of size $n$, 
\be\label{def:joint-distribution}p_{d,k,r,s}=\frac1n\left|\{i\in\mc V:\,\delta_i=d,\,\kappa_i=k,\,\rho_i=r,\,\sigma_i=s\}\right|\,,\qquad d\ge0\,,\ 0\le r\le k\,,\ s=0,1\,,\ee
stands for the fraction of agents having in-degree $d$, out-degree $k$, threshold $r$, and initial state $s$ and  
\be\label{def:average-degree}
l:=\sum_{i\in\mc V}\delta_i=\sum_{i\in\mc V}\kappa_i\,,\qquad
\ov d=\frac ln\ee 
stand for the network's total degree and average degree, respectively.
We refer to $\mb p=\{ \pdkrx \}$ as the \emph{network's statistics} and denote by
\be\label{def:degree-distribution}p_{k,r}:=\sum_{d\ge0}\sum_{s=0,1}p_{d,k,r,s}\,,\qquad 
q_{k,r}:=\frac{1}{\ov d}\sum_{d\ge0}\sum_{s=0,1}dq_{d,k,r,s}\,,\qquad k,r\ge0\,,\ee
the fractions of agents and, respectively, of links pointing to agents, of out-degree $k$ and threshold $r$, and by
\be\label{upsxi}
\ups:=\sum_{d\ge0}\sum_{k\ge0}\sum_{r\ge0}p_{d,k,r,1}\,,\qquad 
\xi:=\frac1{\ov d}\sum_{d\ge0}\sum_{k\ge0}\sum_{r\ge0}dp_{d,k,r,1}\ee 
the fractions of agents and, respectively, of links pointing towards agents, with initial state $\sigma_i=1$. 

\subsection{The recursion}\label{subsect:recursion}

In order to get a quick, unrigorous yet intuitive derivation of the recursion \eqref{CM-recursion}, consider the following random network dynamics with state vector $Y(t)\in\{0,1\}^n$ whose initial state is $Y(0)=\sigma$ and whereby, at each time $t\ge0$, agents $i\in\mc V$ select $\kappa_i$ agents $J^{i}_{1},\ldots,J^i_{\kappa_i}$ independently at random from the population with probability $\P(J^i_h=j)=\delta_j/l$ and update its state as $Y_i(t+1)=1$ if $\sum_{0\le h\le\kappa_i}Y_{J^i_h}\ge\rho_i$ and as $Y_i(t+1)=0$ if $\sum_{0\le h\le\kappa_i}Y_{J^i_h}<\rho_i$.
Let 
$y(t)=\frac1n\sum_{i}Y_i(t)$ and $x(t)=\frac1l\sum_{i}\delta_iY_i(t)$
be the fractions of state-$1$ adopters and links pointing towards state-$1$ adopters, respectively. 
It is immediate to verify that 
\be\label{initialx} 
x(0)=\xi\,,\qquad y(0)=\ups\,.
\ee
On the other hand, if $I$ is a random agent selected from $\mc V$ with uniform probability $\P(I=i)=1/n$, then 
$$\ba{rclcc}
\E[y(t+1)|Y(t)]\ds&=&\ds\P(Y_I(t+1)=1|Y(t))\\[5pt]
&=&\ds\sum_{k\ge0}\sum_{r\ge0}p_{k,r}\P\l(\sum_{h=1}^kY_{J^I_h}(t)\ge r\big|Y(t)\r)\\[15pt]
&=&\ds\sum_{k\ge0}\sum_{r\ge0}p_{k,r}\sum_{u=r}^k\P\l(\sum_{h=1}^kY_{J^I_h}(t)=u\big|Y(t)\r)\\[15pt]
&=&\ds\sum_{k\ge0}\sum_{r\ge0}p_{k,r}\varphi_{k,r}(x(t))\\[15pt]
&=&\ds\psi(x(t))\,,
\ea
$$
where the forth identity above follows with
\be\label{varphi-def}\varphi_{k,r}(x):=\sum_{u=r}^k\binom ku x^u(1-x)^{k-u}\,,\qquad 0\le r\le k\,,\ee
from the fact that, conditioned on $Y(t)$, the $Y_{J^i_h}(t)$ are independent Bernoulli random variables with $\P(Y_{J^i_h}(t)=1|Y(t))=x(t)$, while the last identity holds true upon defining 
\be\label{psi-def}\qquad \psi(x):=\sum_{k\ge0}\sum_{r\ge0} p_{k,r}\varphi_{k,r}(x)\,.\ee
An analogous computation shows that, if $J$ is a random agent selected with probability and $\P(J=j)=\delta_j/l$, then 
$$
\E[x(t+1)|Y(t)]=\P(Y_J(t+1)=1|Y(t))
=\phi(x(t))\,,
$$
where 
\be\label{phi-def}\phi(x):=\sum_{k\ge0}\sum_{r\ge0} q_{k,r}\varphi_{k,r}(x)\,.\ee

While the above computations are merely concerned with the conditional expected fractions of state-$1$ adopters, and links pointing towards state-$1$ adopters, in the random network dynamics $Y(t)$, in Section \ref{sec:density-evolution} we will prove that the output $y(t)$ of the recursion \eqref{CM-recursion} with initial condition \eqref{initialx} does in fact provide a good approximation of the evolution of the fraction of state-$1$ adopters for the actual LTM dynamics \eqref{LTMdef} on most of the networks with given statistics $\mb p$. It will then follow from Lemma \ref{lemma:LTM=PLTM} that such approximation result remains valid for the fraction of state-$1$ adopters in the PLTM dynamics as long as
\be\label{PLTMpcond} p_{d,k,r,1}=0\,,\qquad d\ge0\,,\quad 1\le r\le k\,,\ee 
i.e., when $\rho_i\le\kappa_i(1-\sigma_i)$ for all agents $i\in\mc V$. 
In the remainder of this section we study, both analytically and numerically, the behavior of the recursion \eqref{CM-recursion}.

Observe that every function $\varphi_{k,r}(x)$ defined as in \eqref{varphi-def} maps the unitary interval $[0,1]$ in itself in a continuous and monotonically non-decreasing way, and so do $\phi(x)$ and $\psi(x)$ which are convex combinations of the $\varphi_{k,r}(x)$. The dynamics of the recursion \eqref{CM-recursion} can then be understood by a graphical procedure, consisting in iteratively projecting points from the graph of the function $\phi(x)$ to the diagonal of the $[0,1]\times[0,1]$ square and {\it vice versa} (compare the left-most plot in Figure~\ref{fig:Phi-x}). Continuity and monotonicity of $\phi(x)$ and $\psi(x)$ imply that both the state $x(t)$ and output $y(t)$ of the recursion \eqref{CM-recursion} always converge, as $t$ grows large, to limit values 
that depend on the initial seed $\xi$ only, as formally stated in the following result. 
\begin{lemma}\label{lemma:simple}
Let $\phi(x)$ and $\psi(x)$ be defined as in \eqref{psi-def} and \eqref{phi-def}, respectively, for given network statistics $\mb p$. 
Then, as time $t$ grows large, the state $x(t)$ and output $y(t)$ of the recursion \eqref{CM-recursion} with initial condition \eqref{initialx} are convergent to limit values 
$$x^*(\xi):=\lim_{t\to\infty}x(t)\,,\qquad y^*(\xi):=\lim_{t\to\infty}y(t)$$ 
such that 
\be
x^*(\xi)=\left\{\ba{lcl}
\text{largest fixed point of }\phi(x)\text{ in }[0,\xi) &\se&\phi(\xi)<\xi\\
\xi&\se&\phi(\xi)=\xi\\
\text{smallest fixed point of }\phi(x)\text{ in }(\xi,1] &\se&\phi(\xi)>\xi\ea\right.\qquad\  y^*(\xi)=\psi(x^*(\xi))\,.
\ee
\end{lemma}
\proof
We consider the case $\phi(\xi)<\xi$ first. In this case, $x(1)=\phi(x(0))=\phi(\xi)<\xi=x(0)$ and monotonicity of $\phi(x)$ allows one to prove that  $x(t+1)=\phi(x(t))\le\phi(x(t-1))=x(t)$ by a simple induction argument. Then, $x(t)$ is monotonically non-increasing in $t$, hence converging to a limit $x^*(\xi)$. By continuity of $\phi(x)$, such a limit must be a fixed point $x^*(\xi)=\phi(x^*(\xi))$, and continuity of $\psi(x)$ implies that $y^*(\xi)=\lim_{t\to\infty}y(t)=\lim_{t\to\infty}\phi(x(t-1))=\phi(x^*(\xi))$. Observe that, since $x(t)$ is non-increasing in $t$, then $x^*(\xi)=\lim_tx(t)<x(0)=\xi$. Moreover, there cannot exist another fixed point $x^*=\phi(x^*)$ such that $x^*(\xi)<x^*<\xi$, since, if such $x^*$ existed, then monotonicity of $\phi(x)$ would imply that $x(t)=\phi(\phi(\ldots\phi(\xi)))>\phi(\phi(\ldots\phi(x^*)))=x^*$ leading to the contradiction $x^*(\xi)=\lim_{t\to\infty}x(t)\le x^*$. Hence, $x^*(\xi)=\phi(x^*(\xi))$ is necessarily the largest fixed point of $\phi(x)$ in the interval $[0,\xi)$. The other two cases can be proved analogously. 
\qed


\subsection{Out-regular networks with homogeneous thresholds}\label{subsect:homogeneous}
\begin{figure}
\begin{center}
\includegraphics[width=7cm]{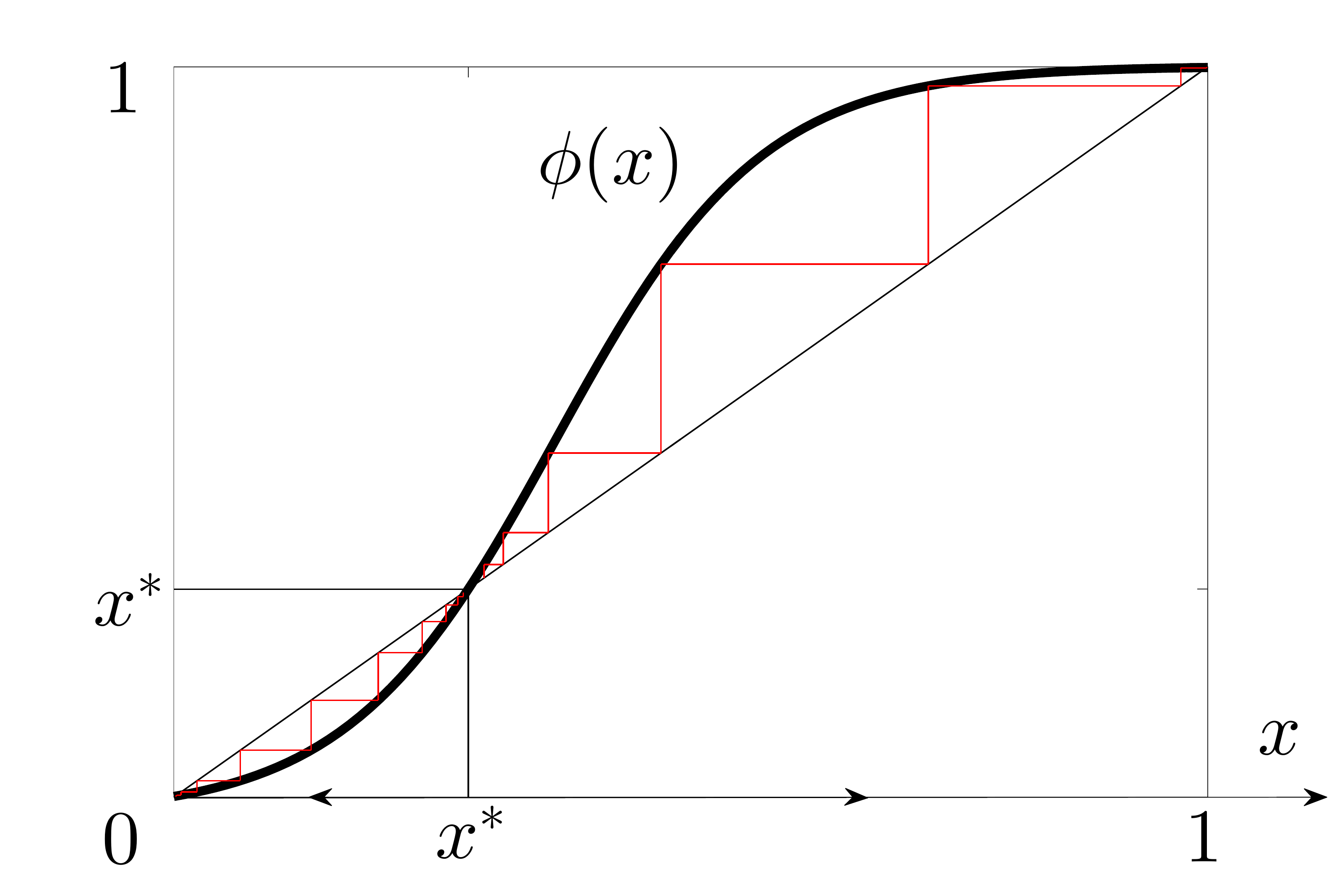}\hspace{1cm}
\includegraphics[width=7cm]{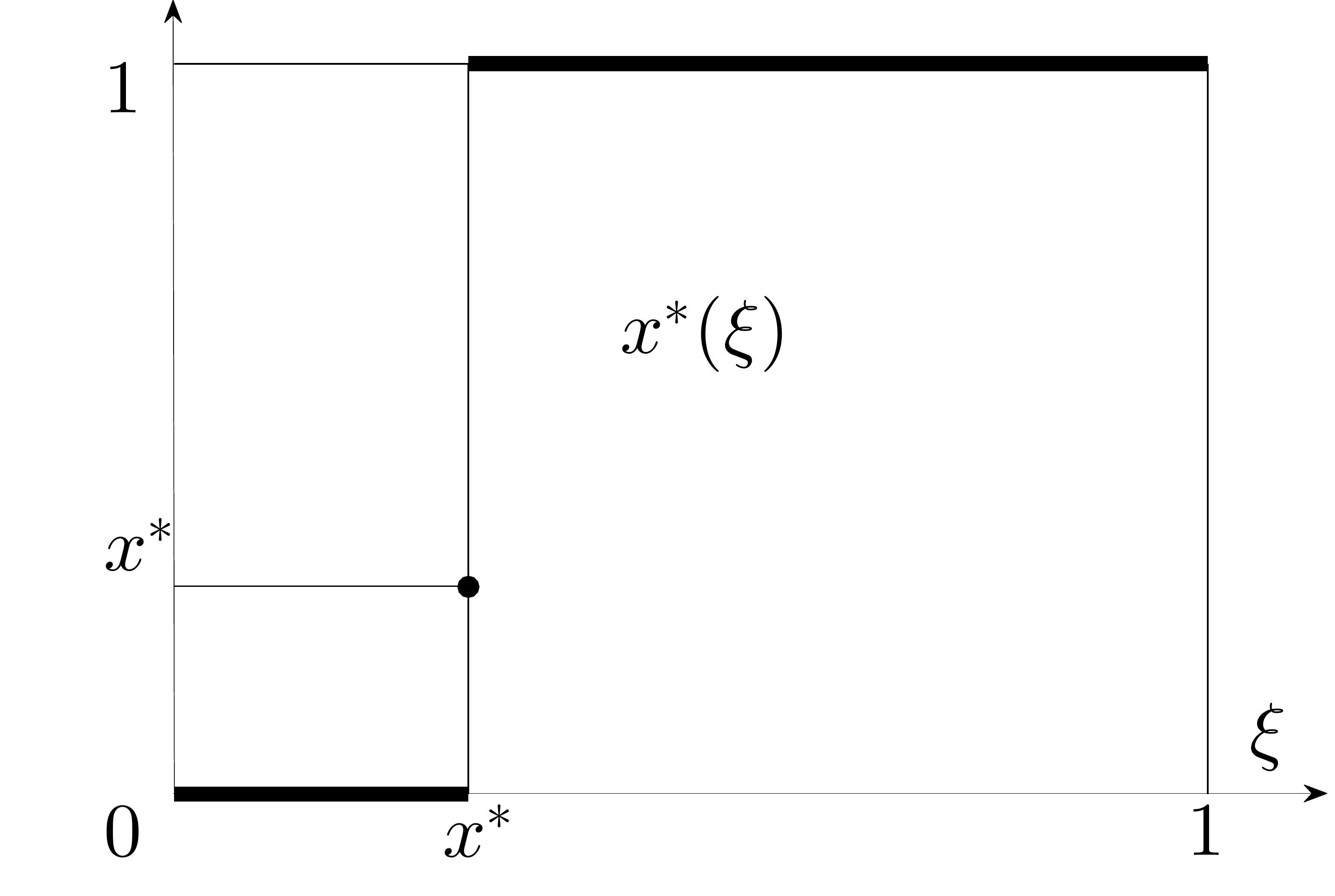}
\end{center}
\caption{\label{fig:Phi-x} The recursion \eqref{recursion-hom} for a typical out-regular network with out-degree $k\ge3$ and homogeneous thresholds $2\le r<k$. The function $\phi(x)=\varphi_{k,r}(x)$ displayed in the leftmost plot has three fixed points: $0$, $1$ and $x^*$.  
The state $x(t)$ of the recursion \eqref{recursion-hom} converges to $0$ for every initial condition $x(0)\in[0,x^*)$, to $1$ for every initial condition $x(0)\in(x^*,1]$, and stays put if $x(0)=x^*$. 
The figure on the right shows the limit of $x(t)$ as $t$ grows large as a function of the initial condition $x(0)$. 
}
\end{figure}
We now focus on the simplest case where all the agents have the same out-degree $\kappa_i=k$ and threshold $\rho_i=r$. In this case, one has that $p_{k,r}=q_{k,r}=1$, $\phi(x)=\psi(x)=\varphi_{k,r}(x)$ and $\ups = \xi$, so that the recursion \eqref{CM-recursion} reduces to 
\be\label{recursion-hom}
x(t+1)=y(t+1)= \varphi_{k,r}(x(t))\,
\ee 
with initial condition $x(0) = y(0) = \xi$.
In the following result, we gather some elementary properties of the functions $\varphi_{k,r}(x)$ whose proof relies merely on basic calculus and that will prove useful later on.
\begin{lemma}\label{lemma:phi-properties}
\begin{enumerate}
For $0\le r\le k$, let $\varphi_{k,r}(x)$ be the polynomial function defined in \eqref{varphi-def}. Then, for $x\in[0,1]$,
\item[(i)] $\varphi_{k,r}(x)$ is a non-decreasing function, strictly increasing if $0<r\le k$;
\item[(ii)] $\varphi'_{k,r}(x)=\binom kr rx^{r-1}(1-x)^{k-r}$;
\item[(iii)] $\varphi''_{k,r}(x)=\binom kr rx^{r-2}(1-x)^{k-r-1}(r-1-x(k-1))$;
\item[(iv)] $\varphi_{k,r}(1)=1$ for $0\le r\le k$; $\varphi_{k,r}(0)=0$ for $0<r\le k$; $\varphi_{k,0}(0)=1$; 
\item[(v)] $\varphi_{k,0}(x)=1$, $\varphi_{k,1}(x)=1-(1-x)^k$, $\varphi_{k,k}(x)=x^k$;
\item[(vi)] For $1\le r\le k$, with $k > 1$ the function $\varphi_{k,r}(x)$ has one inflection point in $\tilde x=(r-1)/(k-1)$. It is strictly convex for $0\leq x\leq \tilde x$ and strictly concave for $\tilde x\leq x\leq 1$;
\item[(vii)] For $2\le r<k$, the equation $\varphi_{k,r}(x)=x$ has exactly three solutions $\{0, x^*, 1\}$ with $0<x^*<1$ and  $\varphi_{k,r}'(x^*)>1$.
\end{enumerate}\end{lemma}

Lemma \ref{lemma:phi-properties} implies that, for $2\le r <k$, the function $\varphi_{k,r}(x)$ has a \emph{lazy-S}-shaped graph, i.e., it is increasing, with a unique inflection point $\tilde x=(r-1)/(k-1)$, it is convex on the left-hand side of $\tilde x$ and concave on the right-hand side of $\tilde x$ (compare Figure~\ref{fig:Phi-x}, left). Besides the trivial cases $r=0,1, k$, whose asymptotics are reported below
$$
y^*(\xi)=x^*(\xi)=\left\{\begin{array}{ll}1\quad &\text{ if }\, r=0\text{ or }r=1<k\\ \xi\quad &\text{ if }r=k=1\\ 0\quad &\text{ if } r=k>1\,,\end{array}\right.$$
point (vii) of Lemma \ref{lemma:phi-properties}  implies that the recursion \eqref{recursion-hom} exhibits a threshold behavior with respect to the initial fraction of state-$1$ adopters. In fact, for  
$2\le r<k$, it holds true that 
\be\label{xinfty}
y^*(\xi)=x^*(\xi)=\left\{\begin{array}{ll}0\quad &{\rm if}\, \xi<x^*\\ 
x^*\quad &{\rm if}\, \xi=x^*\\ 
1\quad &{\rm if}\, \xi>x^*\,,\end{array}\right.\ee
where $x^*=\varphi_{k,r}(x^*)$ is the unique fixed point of $\varphi_{k,r}(x)$ in the open interval $(0,1)$. 
Equation \eqref{xinfty} implies that, if the fraction $\ups=\xi$ of agents with initial state $\sigma_i=1$ is smaller than the threshold value $x^*$, then the fraction of state-$1$ adopters vanishes as time grows large whereas, if $\ups=\xi>x^*$ then the fraction of state-$1$ adopters approaches $1$ asymptotically (cf.~Figure~\ref{fig:Phi-x}, right).


\begin{figure}
	\centering	
	\includegraphics[trim={\figtrimla} {\figtrimba} {\figtrimra} {\figtrimta},clip, width={\figwidthduo},
		keepaspectratio=true]{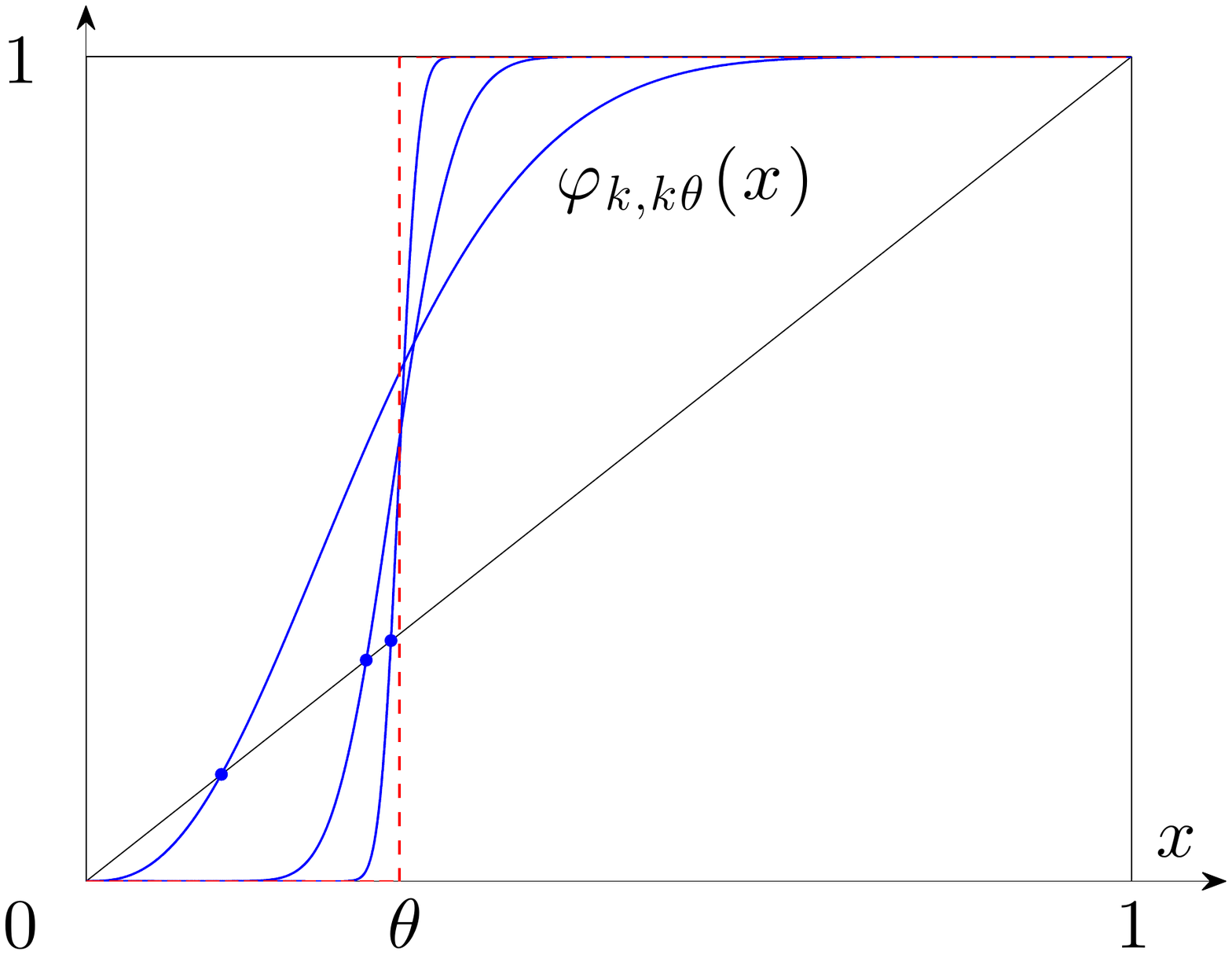}%
		

	\caption{\label{fig:varphi-conv} 
	Plots of the function $\varphi_{k,k\theta}(x)$ for $\theta = 0.3$ and $k = 10$, $100$, $1000$ (blue solid lines). The small full circles represent the internal fixed point $x^*$ of $\varphi_{k,k\theta}(x)$ for the three values of $k$. In the limit of large $k$, the step function $\varphi_{\theta}(x)=0$ for $x<\theta$ and $\varphi_{\theta}(x)=1$ for $x\ge\theta$ (red dashed plot) is achieved.}
\end{figure}


A simple estimation of the threshold value $x^*=\varphi_{k,r}(x^*)$ follows from the observation that $\varphi_{k,r}(x)$ can be interpreted as the probability that a random variable with binomial distribution of parameters $k$ and $x$ exceeds $r$. The fact that mean and median coincide for such binomial random variables when the mean $kx$ is an integer value implies that 
\be\label{median}\varphi_{k,r}(r/k)\geq 1/2\,,\qquad\varphi_{k,r}((r-1)/k)\leq 1/2\,,\ee so that, in particular,
$$\begin{array}{rclll}x^*&\leq& r/k &  \text{ if }& r/k\leq 1/2\\
x^*&\geq& (r-1)/k &\text{ if }& (r-1)/k\geq 1/2\,.\end{array}$$
In fact, if we fix a value $\theta\in [0,1]$ and let $r=\lfloor\theta k\rfloor$, then the law of large numbers implies that
$$\lim_{k\to\infty}\varphi_{k,\lfloor k\theta\rfloor}(x)=\left\{\ba{lcl}0&\se&0\le x<\theta\\1&\se&\theta\le x\le1\,,\ea\right.$$ 
i.e., $\varphi_{k,r}(x)$ approaches a step function in the limit as $k$ grows large (cf.~Figure~\ref{fig:varphi-conv}). 
This shows that the ratio $\theta=r/k$ is a good approximation of the threshold value $x^*=\varphi_{k,r}(x^*)$ when $r$ and $k$ are large enough.

\begin{figure}
	\centering
	\includegraphics[trim={\figtrimla} {\figtrimba} {\figtrimra} {\figtrimta},clip, width={\figwidthduo},
		keepaspectratio=true]{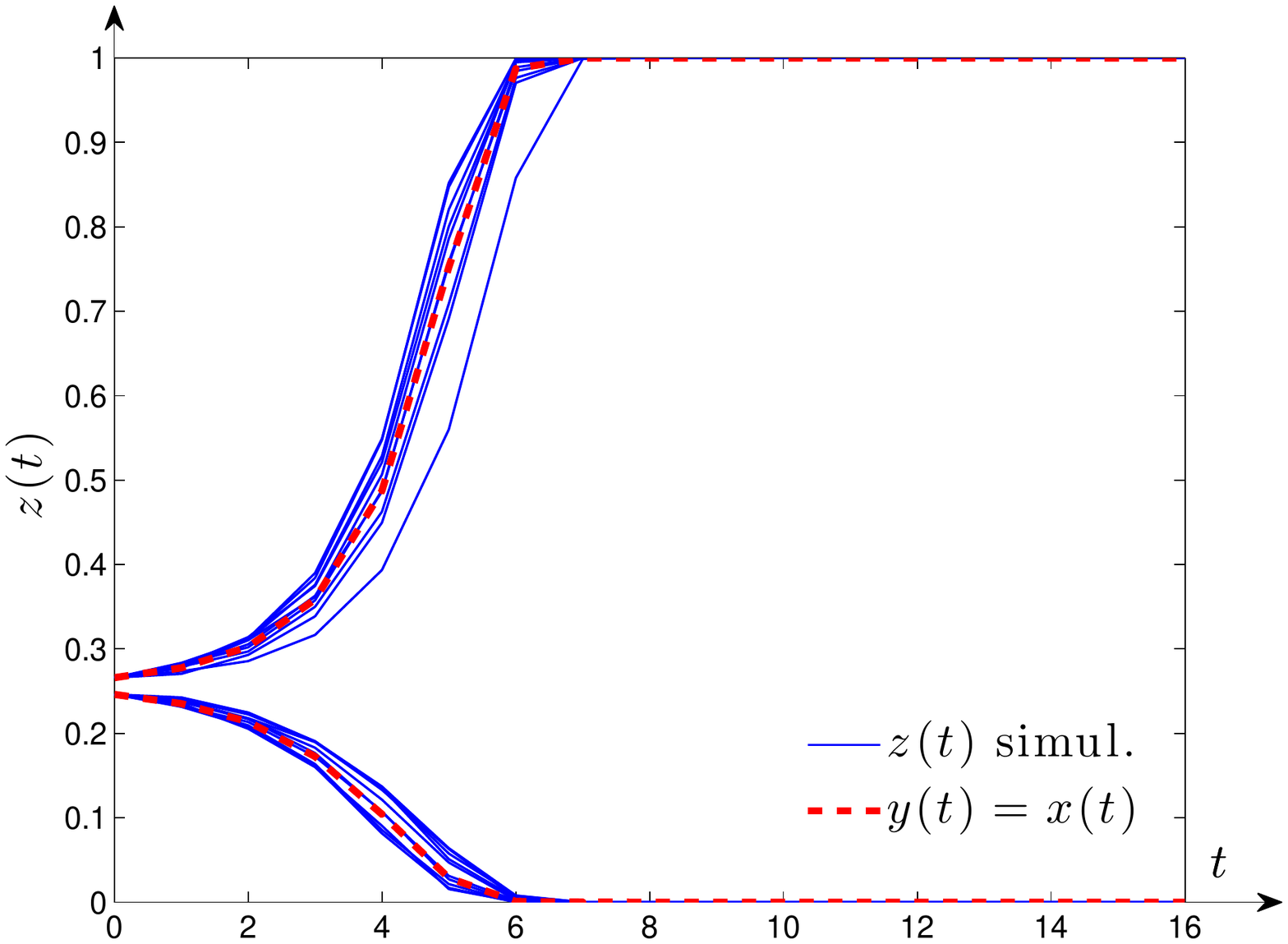}%
	\hspace{1cm}%
	\includegraphics[trim={\figtrimla} {\figtrimba} {\figtrimra} {\figtrimta},clip, width={\figwidthduo},
		keepaspectratio=true]{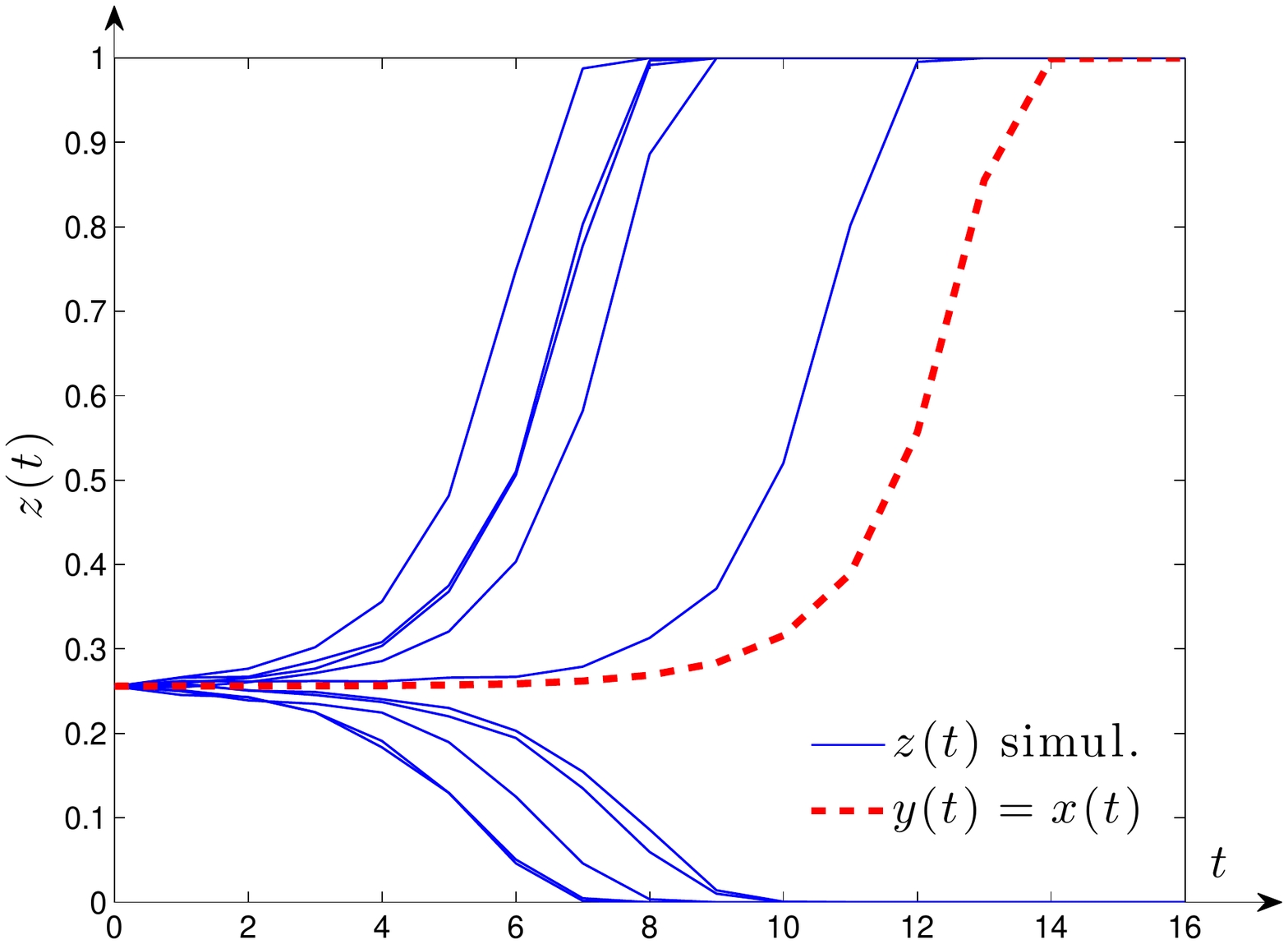}%
	\caption{\label{fig:ex-homo-rand-dyn} The dynamics of the fraction $z(t)$ of state-$1$ adopters in the LTM \eqref{LTM-def} (blue solid lines) compared with  the recursion \eqref{recursion-hom} (dashed red lines). The random networks used in the simulations have $n = 2000$ nodes, $k = 7$, $r = 3$. The initial condition are $\ups = 0.246$ and $\ups = 0.266$ (left plot) and $\ups = 0.256$ (right plot). Note that $\varphi_{7,3}(x)$ has  $x^* \approx 0.256$. 
The simulations with $\ups = 0.246$ converge to zero, those for $ \ups = 0.266$ converge to one; in both case the recursion captures the behavior and the timing of the simulated dynamic. 
For $\ups = 0.256$, which  happens slightly larger than $x^*$, after a slow start the recursion converges to one while the simulations are evenly spread, half converge to one and half to zero, with different timing too. 
The simulations $z(t)$ become closer to the recursion $y(t)$ if the network sizes $n$ is increased, or (as in the left plot) if $\ups$ is chosen a bit away from $x^*$.
The limit $y^*(\xi)$ for various seed $\xi$ can also be compared with simulations on the same random network used above. If the seed $\xi$ is not very close to $x^*$, the simulation always match the predicted limit, like the right plot of Figure~\ref{fig:Phi-x}.	}
\end{figure}

\subsection{Heterogeneous networks: local analysis}\label{subsect:hetero1}
\begin{figure}
\begin{center}
\includegraphics[width=7cm]{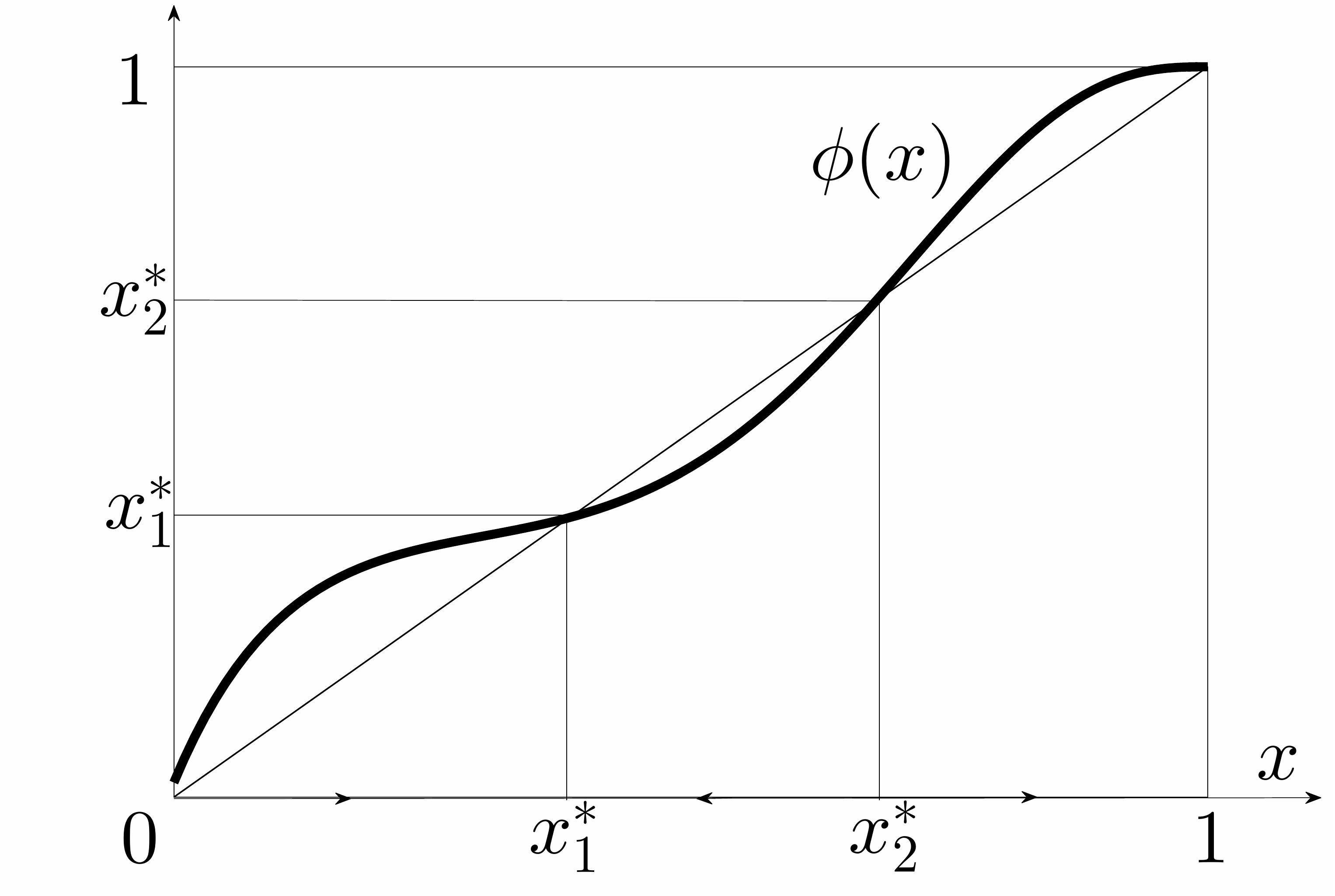}\hspace{1cm}
\includegraphics[width=7cm]{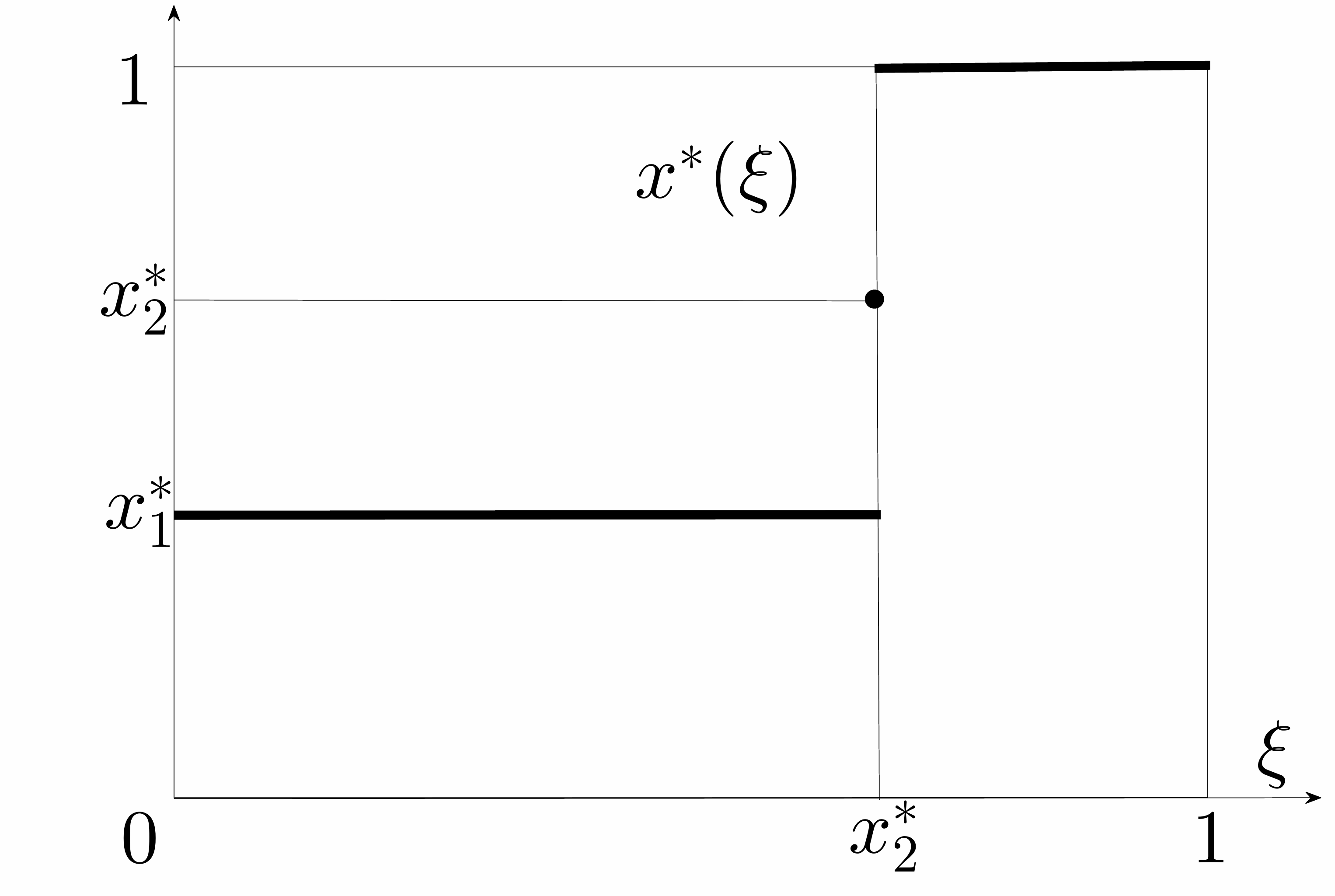}
\end{center}
\caption{\label{fig:phi-multi-conv}
On the left, the function $\phi(x)$ for a heterogeneous network with $q_{10,0}=0.02$ $q_{8,6}=0.64$ and $q_{10,1}=0.34$. The function has two inflection points and three fixed points: $x^*_1$, $x^*_2$, and $x_3^*=1$. Note that $\phi(0)=0.02>0$ and $\phi'(0)=3.4>1$, while $\phi'(1)=0$. On the right, the limit value $x^*(\xi)$ as a function of the initial seed $\xi$. }
\end{figure}

In heterogeneous networks, containing a mixture of agents with different out-degrees and thresholds, the functions $\phi(x)$ and $\psi(x)$ remain non-decreasing ---as they are polynomials with nonnegative coefficients--- while the shape of their graph can be more complex than the simple lazy-$S$ of the homogeneous case. In particular, their convexity may change several times in the interval $[0,1]$ (see, e.g., Figures~\ref{fig:phi-multi-conv} and \ref{fig:phi-xi}). 
Observe that $\phi(x)=\psi(x)$ as in the homogeneous case whenever \be\label{indep}\sum_{s=0,1}p_{d,k,r,s}=p_{k,r}\sum_{k'\ge0}\sum_{r'\ge0}\sum_{s=0,1}p_{d,k',r',s}\,,\qquad 0\le r\le k\,,\ee i.e., when the statistics of the in-degrees across the population are independent from the statistics of the out-degrees and thresholds (since in this case $q_{k,r}=p_{k,r}$).
Instead, one has that $\phi(x)\ne\psi(x)$ for general networks that do not enjoy property \eqref{indep}. 
In this latter case, the fraction of state-$1$ adopters in the LTM dynamics, estimated by the output $y(t)$ of the recursion \eqref{CM-recursion}, does not necessarily coincide with the fraction of links pointing towards state-$1$ adopters in the LTM dynamics, approximated by the state $x(t)$ of the recursion \eqref{CM-recursion}.

In this subsection, we analyze the dynamical behavior of the recursion \eqref{CM-recursion} for values of the initial seed that are either close to $0$ or to $1$. 
To start with, notice that 
point (iv) of Lemma~\ref{lemma:phi-properties} implies that 
\be\label{phi1}
\phi(1)=\psi(1)=1\,.
\ee
On the other hand, 
\be\label{phi0}
\phi(0)=\sum_{k\ge0}q_{k,0}\,,\qquad \psi(0)=\sum_{k\ge0}p_{k,0}
\ee
coincide with the fractions of links pointing towards agents, and, respectively, of agents, with threshold $0$. Analogously, it follows from point (ii) of Lemma \ref{lemma:phi-properties} that 
\be\label{phi'}\phi'(0)=\sum_{k\ge1}kq_{k,1}\,,\qquad \phi'(1)=\sum_{k\ge0}kq_{k,k}\,.\ee
The rightmost identity in \eqref{phi'} and  \eqref{phi1} imply that the asymptotic behavior of the recursion \eqref{CM-recursion} for the standard LTM when the initial seed $\xi$ is close to $1$ is determined by the sign of   $\vartheta-1$ where  $$\vartheta:=\sum_{k\ge0}kq_{k,k}\,.$$ 
Since $\phi(1) = 1$, if $\vartheta<1$ then $\phi(x)>x$ in a left  neighborhood of $1$, whereas if $\vartheta>1$ then $\phi(x)<x$ in a left  neighborhood of $1$.
In the first case, for a seed $\xi$ close enough to $1$, the fraction of state-$1$ adopters approaches  $y^*(\xi)=1$ as $t$ grows large, whereas in the second case it converges to some $y^*(\xi)<1$ even for values of the seed $\xi$ arbitrarily close to $1$ (while, clearly, the recursion stays put in $x(t)=y(t)=1$ if $\xi=\ups=1$).

On the other hand, the leftmost identity in \eqref{phi'} implies that the asymptotic behavior of the recursion \eqref{CM-recursion} when the initial seed $\xi$ is close to $0$ is determined by $\phi(0)$ and by the sign of $\gamma-1$, where 
\be\label{gamma}\gamma:=\sum_{k\ge0}kq_{k,1}\,.\ee 
This can be appreciated in two different settings. 
First, we focus on the standard LTM on networks containing no stubborn agents, i.e., where  $\sum_{k\ge0}p_{k,0}=\sum_{k\ge0}q_{k,0}=0$. Then, $\phi(0)=\psi(0)=0$ by \eqref{phi0} and the leftmost identity in \eqref{phi'} implies that, if  $\gamma<1$, then $\phi(x)<x$ in a right neighborhood of $0$, 
whereas, if $\gamma>1$, then $\phi(x)>x$ in a right neighborhood of $0$. 
In the first case, for small enough seed $\xi>0$, the fraction of state-$1$ adopters approaches  $y^*(\xi)=0$ as $t$ grows large, whereas in the second case it converges to some $y^*(\xi)>0$ even for arbitrarily small positive values of the seed $\xi$ (the recursion stays put in $x(t)=y(t)=0$ if $\xi=\ups=0$).

Alternatively, one can focus on the analysis of the PLTM on networks where the statistics of the initial states are independent from the ones of the degrees and thresholds. Specifically, consider networks with joint degree, threshold, and initial state distributions of the form 
$$p_{d,k,0,1}=\xi \sum_{0\le r\le k}\ov p_{d,k,r}\,,\qquad  p_{d,k,r,0}=(1-\xi)\ov p_{d,k,r}\,,\qquad d,k\ge0,\ 1\le r\le k\,,$$ 
where $\xi\in[0,1]$ stands for the fraction of initial state-$1$ adopters, and $p_{d,k,0,0}=p_{d,k,r,1}=0$ for all $k,d\ge0$ and $1\le r\le k$.  
Here, $\ov p_{d,k,r}$ stands for the fraction of agents with initial state $0$ that have in-degree $d$, out-degree $k$, and threshold $r$. 
Observe that, in this setting, condition \eqref{PLTMpcond} is satisfied, 
while the functions in \eqref{psi-def} and \eqref{phi-def} satisfy
$$\label{phipsieps}\phi_{\xi}(x)=\xi+(1-\xi)\phi_0(x)\,,\qquad \psi_{\xi}(x)=\xi+(1-\xi)\psi_0(x)\,,$$
i.e., they are obtained by rescaling the ones with seed $\xi=0$, i.e., where all agents have initial state $0$.  (See, e.g.,  Figure~\ref{fig:phi-xi}.)
\begin{figure}
	\centering
	\includegraphics[trim={5mm} {\figtrimb} {\figtrimr} {\figtrimt},clip, width={\figwidth},
		keepaspectratio=true]{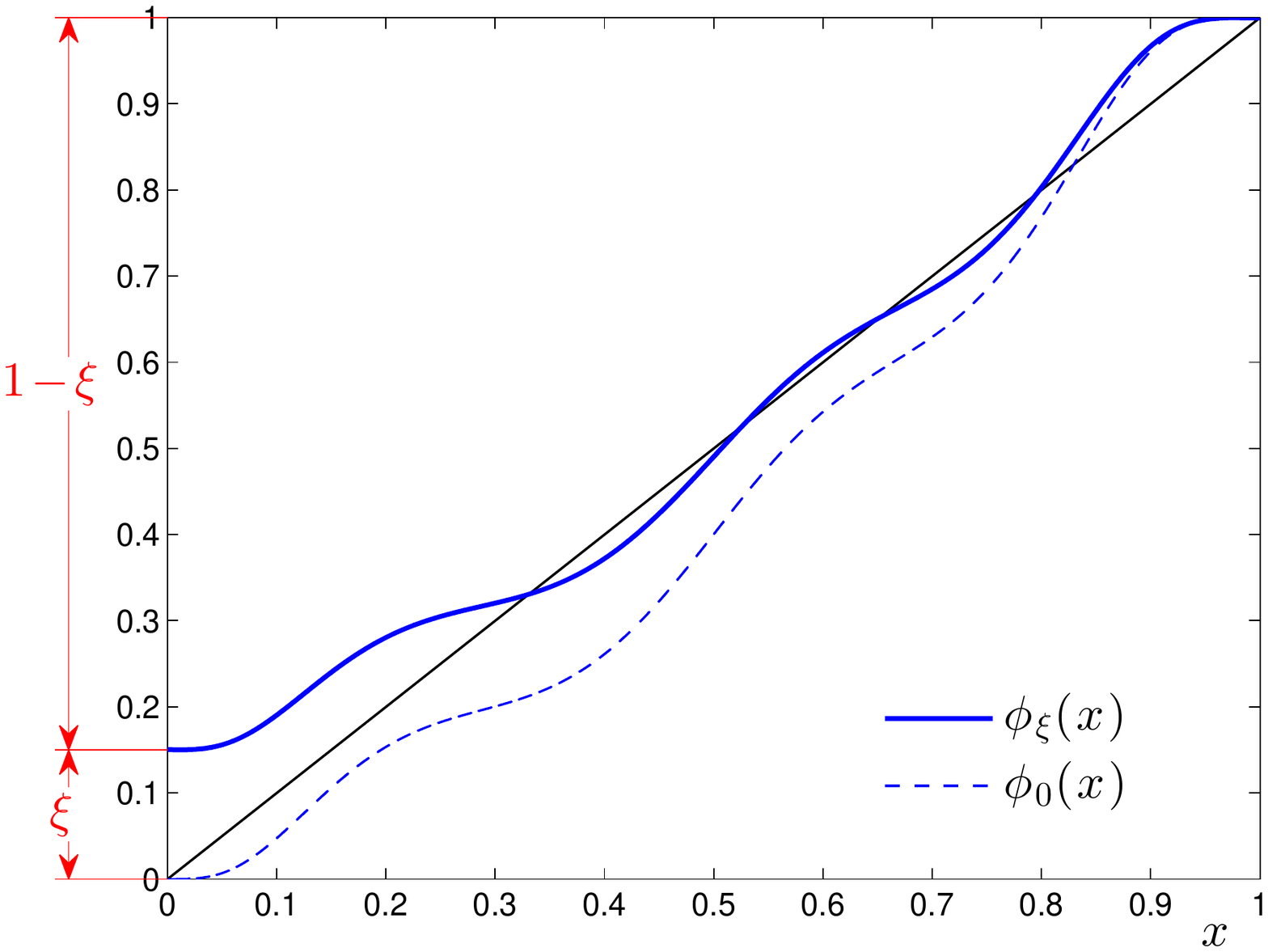}%
	\caption{\label{fig:phi-xi} The function drawn with solid line is $\phi_{\xi}(x) = \xi + (1-\xi) \phi_0(x)$ with $\xi = 0.15$, obtained from $\phi_0(x) = 0.2\,\varphi_{25,4}(x) + 0.4\,\varphi_{25,13}(x) + 0.4\,\varphi_{25,21}(x) $ (dashed line). }
\end{figure}
%
%
In fact, we have that 
\be\label{phi0eps}\phi_{\xi}(0)=\psi_{\xi}(0)=\xi\,,\qquad \phi'_{\xi}(0)=(1-\xi)\gamma\,,\ee
where $\gamma$
is as in \eqref{gamma}. 
It then follows from  \eqref{phi0eps} that 
\be\label{epsto0} \left.\ba{rclcrcl}\ds\lim_{\xi\to0}y^*(\xi)=0&\se&\gamma<1&\qquad&\ds\lim_{\xi\to0}y^*(\xi)>0&\se&\gamma>1\,.\ea\right.\ee
It is worth pointing out that equation \eqref{epsto0} is consistent with the result stated as Theorem 11 in \cite{mL:2012}. In fact, while reference \cite{mL:2012} deals with the PLTM on random undirected graphs with given degree distribution, our results can be extended to  the configuration model ensemble of undirected graphs, as opposed to directed ones, as illustrated in Section \ref{sect:undirectedCM}. The main difference between our approach and the one in \cite{mL:2012} is then that we deal with approximations of the (P)LTM dynamics for large-scale networks and equation \eqref{epsto0} concerns the asymptotic behavior of this approximation as the time $t$ grows large, whereas the results in \cite{mL:2012} deal with the large-scale limit of the asymptotic behavior (as $t$ grows large) of the PLTM dynamics, thus considering the double limit ---in time $t$ and network size $n$--- in the opposite order as we do in this paper.

\subsection{Heterogeneous networks: global analysis} \label{subsect:hetero2}
As mentioned in the previous subsection, the function $\phi(x)$ may have a complex shape for heterogeneous networks and in general it is hard to predict analytically, in terms of the network statistics $\mb p$, the number and value of the fixed points $x^*=\phi(x^*)$ that ---as stated in Lemma \ref{lemma:simple}--- determine the asymptotic behavior of the recursion \eqref{CM-recursion}  as a function of  the initial seed $\xi$.
We present below two special cases when such analytical conditions on the network statistics $\mb p$ can be found explicitely.  

\begin{example}
Let $h>0$ be an integer value, and assume that $q_{k,r}=p_{k,r}=0$ for all  pairs $(k,r)$ except for a subset of those such that $k=jh+1$ and $r=j+1$ for some  $j>0$. Since, by Lemma \ref{lemma:phi-properties}(vi), the functions $\varphi_{jh+1,j+1}(x)$, for $j>0$, all take value $0$ for $x=0$ and $1$ for $x=1$, have a unique inflection point in $\tilde x=1/h$, are convex in $[0,\tilde x]$ and concave in $[\tilde x,1]$, the same does the function 
$$\phi(x)=\sum_{k\ge0}\sum_{r\ge0}q_{k,r}\varphi_{k,r}(x)=\sum_{j>0}q_{jh+1,j+1}\varphi_{jh+1,j+1}(x)\,.$$
Hence, in this very special heterogenous case, the qualitative asymptotic behavior of the recursion \eqref{CM-recursion} is provably the same as in the homogeneous case, as discussed in Section \ref{subsect:homogeneous}: there exists a unique fixed point $x^*=\phi(x^*)$ in $(0,1)$ such that 
$$
x^*(\xi)=\left\{\begin{array}{ll}0\quad &{\rm if}\, \xi<x^*\\ 
x^*\quad &{\rm if}\, \xi=x^*\\ 
1\quad &{\rm if}\, \xi>x^*\,,\end{array}\right.\qquad y^*(\xi)=\left\{\begin{array}{ll}0\quad &{\rm if}\, \xi<x^*\\ 
\psi(x^*)\quad &{\rm if}\, \xi=x^*\\ 
1\quad &{\rm if}\, \xi>x^*\,.\end{array}\right.$$
\end{example}

\begin{example}\label{example:multi-threshold}  
For given $0<\epsilon <1/4$ and $\tau$ such that $2\epsilon<\tau<1-2\epsilon$, consider a network comprising two types of agents,  $h=1,2$, each with out-degree $k_h$ and threshold $r_h$, respectively. 
Assume that $1< r_1<\eps k_1$, that $(1-\eps)(k_2-1) + 1 < r_2 < k_2$ and that the fraction of links pointing towards agents of type $1$ is $q_{r_1,k_1}=\tau=1-q_{r_2,k_2}$. 
Notice that, because of (\ref{median}), $\phi(x)=\tau \varphi_{k_1,r_1}(x)+(1-\tau)\varphi_{k_2,r_2}(x)$ satisfies
$$\phi(r_1/k_1)\geq \frac{1}{2}\tau>\epsilon>\frac{r_1}{k_1}$$
while
$$\phi((r_2-1)/k_2)\leq \tau +\frac{1}{2}(1-\tau)=\frac{1}{2}(1+\tau)<1-\epsilon<\frac{r_2-1}{k_2}\,.$$
Since $\phi(0)=0$, $\phi(1)=1$, and $\phi'(0)=\phi'(1)=0$, this implies that there must be at least five fixed points $x_j^*=\phi(x_j^*)$, $j=0,\ldots,4$, such that 
$$0=x^*_0<x^*_1<\frac{r_1}{k_1}<x^*_2<\frac{r_2-1}{k_2}<x^*_3<x^*_4=1$$
and 
$$
x^*(\xi)=\left\{\begin{array}{ll}
0\quad &{\rm if}\, \xi<x^*_1\\ 
x^*_1\quad &{\rm if}\, \xi=x^*_1\\ 
x^*_2\quad &{\rm if}\, \overline x^*_1<\xi<\underline x^*_3\\
x^*_3\quad &{\rm if}\, \xi=x^*_3\\ 
1\quad &{\rm if}\, \xi>x^*_3
\end{array}\right.\qquad 
y^*(\xi)=\left\{\begin{array}{ll}
0\quad &{\rm if}\, \xi<x^*_1\\ 
\psi(x^*_1)\quad &{\rm if}\, \xi=x^*_1\\ 
\psi(x^*_2)\quad &{\rm if}\, \overline x^*_1<\xi<\underline x^*_3\\
\psi(x^*_3)\quad &{\rm if}\, \xi=x^*_3\\ 
1\quad &{\rm if}\, \xi>x^*_3\,,
\end{array}\right.$$
where $\ov x_1^*=\phi(\ov x_1^*)\in[x_1^*,x_2^*)$ and $\ul x_3^*=\phi(\ul x_3^*)\in(x_2^*,x_3^*]$ are possibly additional fixed points (the largest below and, respectively, the lowest above $x_2^*$). 
This instantiates a multiple threshold phenomenon that is a specific feature of heterogeneous networks, as it cannot occur in homogeneous ones.

The following simulation illustrates the multiple threshold phenomenon just described. 
We consider a random network with $n = 2000$ agents, 45\% of whom has out-degree $k_1 = 14$ and threshold $r_1 = 3$, the remaining 55\% has $k_2 = 11$ and $r_2 = 9$. The initial state of a fraction $\ups$ of the agents is one. The agents have in-degree chosen in $\{11, 14\}$ independently from the out-degree, threshold and initial condition. 
Therefore, the random network satisfies the assumptions of Example~\ref{example:multi-threshold} with $\eps = 0.225$ and $\tau = 0.45$; moreover $\psi(x) = \phi(x)$ and $\xi = \upsilon$.
%
The left plot of Figure~\ref{fig:ex-hetero-glob-rand-dyn} represents the function  
$\psi(x) = 0.450 \, \varphi_{14,3}(x) + 0.550\, \varphi_{11,9}(x)$, 
which has exactly five fixed points: $x^*_0 = 0$, $x^*_1 \approx 0.140$, $x^*_2 \approx 0.451$, $x^*_3 \approx 0.813 $ and  $x^*_4 = 1$. 
%
The right plot of Figure~\ref{fig:ex-hetero-glob-rand-dyn} contains (in solid red) the predicted limit of the recursion, $y^*(\xi)$, that is a ``staircase'' function with two discontinuities. 
The blue crosses represent thew simulations, namely the fraction $z(T)$ of state-$1$ adopters in the random networks at $T = 200$, for various initial conditions. 

\begin{figure}
	\centering
	\includegraphics[trim={\figtrimla} {\figtrimba} {\figtrimra} {\figtrimta},clip, width={\figwidthduo},
		keepaspectratio=true]{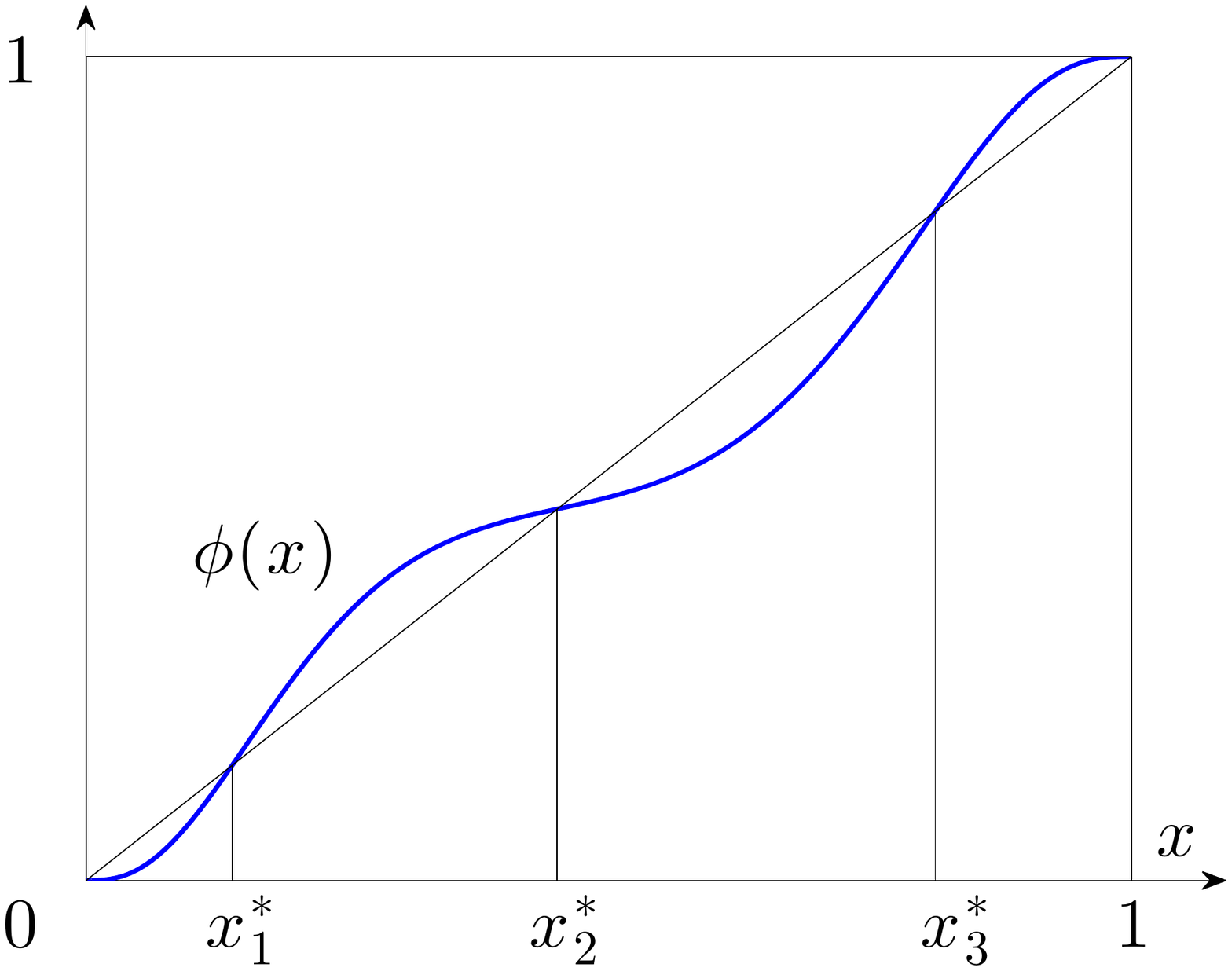}%
	\hspace{1cm}%
	\includegraphics[trim={\figtrimla} {\figtrimba} {\figtrimra} {\figtrimta},clip, width={\figwidthduo},
		keepaspectratio=true]{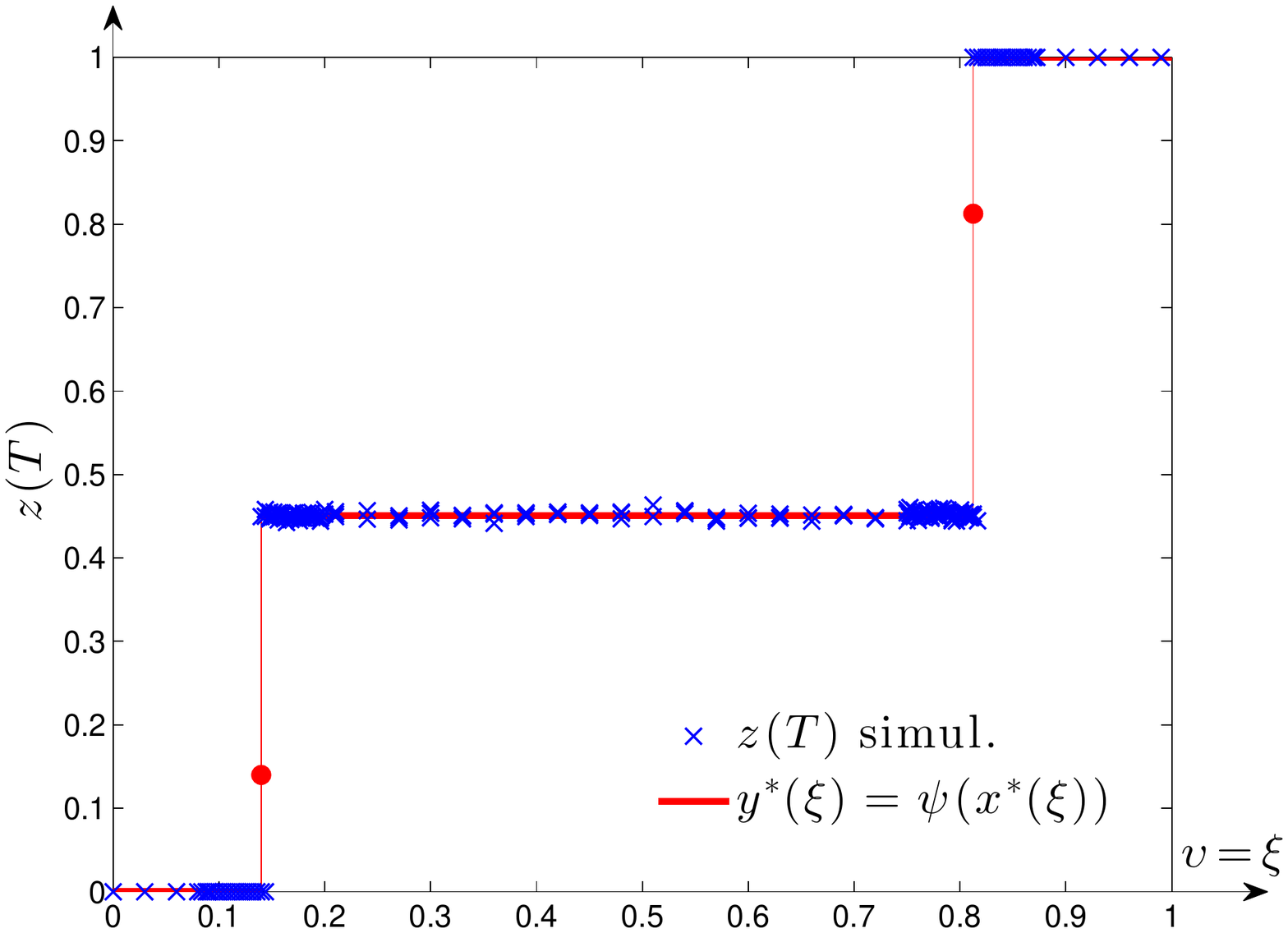}%
	\caption{\label{fig:ex-hetero-glob-rand-dyn} 
	Left plot: the function $\psi(x) = 0.450 \, \varphi_{14,3}(x) + 0.550\, \varphi_{11,9}(x)$ with its fixed points.  
	Right plot: the ``staircase'' function (solid red) is the limit $y^*(\eps)$ of the recursion. 
	The blue crosses represent the values $z(T)$ at time $T = 200$ of the simulations on the random network with $n=2000$ agents and  different seed values $\xi$. 	}
\end{figure}
\end{example}


While the investigation of the exact number and positions of the various fixed points of $\phi(x)$ for general heterogeneous network is analytically unfeasible, fundamental insight can be obtained by taking a large-degree limit as follows. Let $F(\theta)$ be a normalized threshold cumulative distribution function, i.e., $F(\theta)$ is non-decresing, right-continuous, with $F(\theta)=0$ for $\theta<0$ and $F(\theta)=1$ for $\theta\ge1$. Assume that the network statistics $\mb p$ satisfy 
\be\label{Frk}p_{k,r}=p_k(F(r/k)-F((r-1)/k))\,,\qquad q_{k,r}=q_k(F(r/k)-F((r-1)/k))\,,\qquad 0\le r\le k\,.\ee
where $p_k$ and $q_k$ stand for the fractions of agents and, respectively, links pointing towards agents of degree $k$, and the minimum out-degree $k_{\min}\ge2$ is such that $p_k=q_k=0$ for all $0\le k\le k_{\min}$. 
In this case, the function $\phi(x)$ and $\psi(x)$ take the form 
$$
\ba{rcl}
\ds\phi(x)&=&\ds\sum_{k\ge 0}\sum_{0\le r\le k}q_{k,r}\varphi_{k,r}(x)\\[10pt]
&=&
\ds\sum_{k\ge k_{\min}}q_k\sum_{0\le r\le k}\big(F(r/k)-F((r-1)/k) \big)\sum_{r\le j\le k}\binom kjx^j(1-x)^{k-j}\\[10pt]
&=&
\ds\sum_{k\ge k_{\min}}q_k\sum_{0\le j\le k}\sum_{0\le r\le j}\big(F(r/k)-F((r-1)/k)\big)\binom kjx^j(1-x)^{k-j}\\[10pt]
&=&
\ds\sum_{k\ge k_{\min}}q_k\sum_{0\le j\le k}F(j/k)\binom kjx^j(1-x)^{k-j}\,,
\ea
$$
$$\psi(x)=\sum_{k\ge 0}\sum_{0\le r\le k}p_{k,r}\varphi_{k,r}(x)=\sum_{k\ge k_{\min}}p_k\sum_{0\le j\le k}F(j/k)\binom kjx^j(1-x)^{k-j}.$$

Then, if a sequence of network statistics with increasing minimum out-degree $k_{\min}$ is considered satisfying \eqref{Frk} with the same normalized threshold cumulative distribution function $F(\theta)$, then 
\be\label{phi->F}
\lim_{k_{\min}\to+\infty}\phi(x)=F(x)\,,\qquad\lim_{k_{\min}\to+\infty}\psi(x)=F(x)\,.
\ee
The result above establishes that, as the minimum degree grows large, the recursion \eqref{CM-recursion} reduces to the Granovetter one \eqref{LTM-recursion-fullymixed}. In fact, by applying Lemma \ref{lemma:simple} with $\phi(x)=\psi(x)=F(x)$ one gets that the (approximate) fraction of state-$1$ adopters $y(t)$ converges to the largest (respectively, lowest) fixed point $x^*(\xi)=F(x^*(\xi))$  
that is not higher (not lower) than the initial seed $\xi$. 
That together with \eqref{phi->F} highlights a selected activation phenomenon for networks satisfying \eqref{Frk}: for large enough $k_{\min}$ the eventual state-$1$ adopters tend to be those agents $i$ whose normalized threshold $\theta_i=\rho_i/\kappa_i$ is below the fixed point $x^*(\xi)=F(x^*(\xi))$. 

\section{Approximation results for the configuration model ensemble} 
\label{sec:density-evolution}
In this section, we show that the output $y(t)$ of the recursion \eqref{CM-recursion} introduced in Section \ref{sec:LTM-BR} does in fact provide an accurate approximation of the fraction of state-$1$ adopters in the LTM dynamics \eqref{LTM-def} on most of the directed networks $\mc N$ with the same statistics $\mb p$. 
Specifically, we introduce the so-called \emph{configuration model} ensemble $\mc C_{n,\mb p}$ of all directed networks of given size $n$ and statistics $\mb p$ and prove that the fraction of state-$1$ adopters 
$$z(t)=\frac1n\sum_{i\in\mc V}Z_i(t)$$ 
after a finite number of iterations of the LTM dynamics \eqref{LTM-def} is arbitrarily close to the output $y(t)$ of the recursion \eqref{CM-recursion} on all but a fraction of networks in $\mc C_{n,\mb p}$ that vanishes as $n$ grows large. 

Our result is proved in three main steps. 
First, we introduce a different random graph model with rooted tree structure, the \emph{two-stage branching process} $\mc T_{\mb p}$, and show that the output $y(t)$ of the recursion \eqref{CM-recursion} gives the \emph{exact} expression of the \emph{expected value} of the root node's state in the LTM dynamics \eqref{LTM-def} on $\mc T_{\mb p}$. 
Second, we consider the configuration model $\mc C_{n,\mb p}$ and prove that, after $t$ iterations of the LTM dynamics \eqref{LTM-def} on the configuration model ensemble, the \emph{average} fraction $\ov z(t)$ of state-$1$ adopters is arbitrarily close to $y(t)$, i.e., the expected value of the root node's state on $\mc T_{\mb p}$.
Finally, a concentration result is obtained, showing that on most of the networks in $\mc C_{n,\mb p}$, the fraction $z(t)$ of state-$1$ adopters after $t$ iterations of the LTM dynamics is arbitrarily close to its average $\ov z(t)$, hence to the output $y(t)$ of the recursion \eqref{CM-recursion}.

\subsection{The LTM on the two-stage branching process}
\label{subsect:LTM-BP}
\begin{figure}
\begin{center}
\includegraphics[width=7cm]{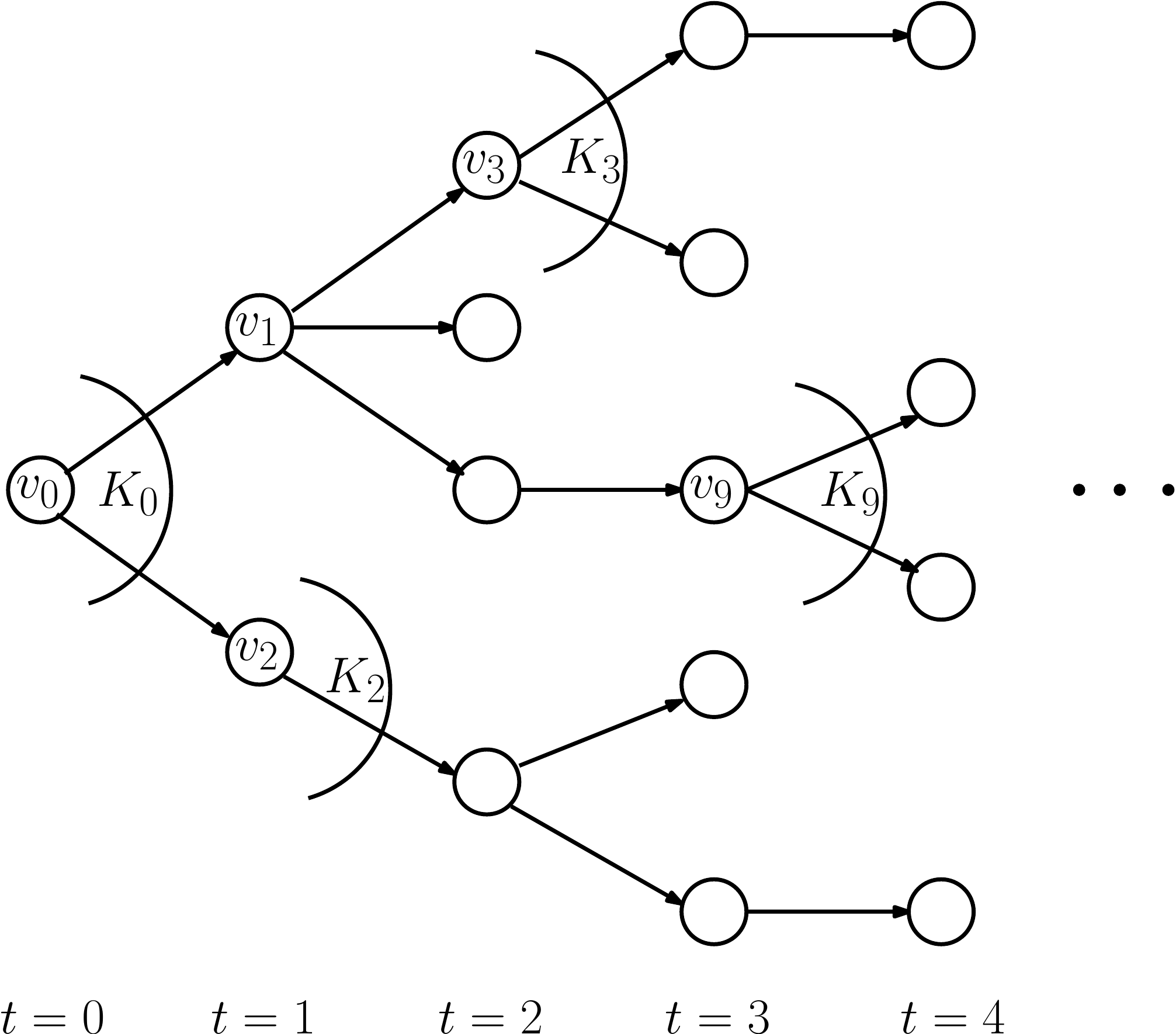}
\end{center}
\caption{\label{fig:BP} A directed two-stage branching process $\mc T$ with root node $v_0$. The triples $(K_h,R_h,S_h)$, for $h\ge0$, of the agents' outdegrees, thresholds, and initial states are mutually independent and have distribution $\P(K_0=k,R_0=r,S_0=s)=p_{k,r,s}$ and $\P(K_h=k,R_h=r,S_h=s)=q_{k,r,s}$ for $h\ge1$, $0\le r\le k$, and $s=0,1$. The state $X_{v_0}(t)$ of the root node at time $t\ge0$ is a deterministic function of the initial states $S_j$ of the agents $j$ in generation $t$.} 
\end{figure}

In this subsection we first introduce a random graph model with rooted directed tree structure, to be referred to as the two-stage branching process $\mc T_{\mb p}$. 
Then, we provide a complete theoretical analysis of the LTM dynamics on $\mc T_{\mb p}$ that will be the basis for then considering, in the next subsection, the configuration model ensemble $\mc C_{n,\mb p}$ which exhibits a local tree-like structure.

Let $\mb p$ be the network statistics with average degree $\ov d=\sum_{d,k,r,s}dp_{d,k,r,s}=\sum_{d,k,r,s}kp_{d,k,r,s}$ and
$$p_{k,r,s}=\sum_{d\ge0}p_{d,k,r,s}\,,\qquad q_{k,r,s}=\frac1{\ov d}\sum_{d\ge0}dp_{d,k,r,s}\,,\qquad 0\le r\le k,\ s=0,1\,,$$ 
be the fractions of agents and, respectively, of links pointing to agents, of out-degree $k$, threshold $r$, and initial state $s$. 
In order to define the associated two-stage branching process $\mc T_{\mb p}$, we start from a root node $v_0$ and randomly generate a directed tree graph according to the following rule (compare Figure~\ref{fig:BP}). First, we assign to the root node $v_0$ a random out-degree $\kappa_{v_0}=K_0$, threshold $\rho_{v_0}=R_0$ and initial state $\sigma_{v_0}=S_0$ such that the triple $(K_0,R_0,S_0)$ has joint probability distribution $\P(K_0=k,R_0=r,S_0=s)=p_{k,r,s}$ for $0\le r\le k$ and $s=0,1$. Then, we connect the root node $v_0$ with $K_0$ directed links pointing to new nodes $v_1,\ldots,v_{K_0}$, and assign to each such generation-$1$ node $v_h$, $1\le h\le K_0$, out-degree $\kappa_{v_h}=K_h$, threshold $\rho_{v_h}=R_{h}$, and initial state $\sigma_{v_h}=S_{h}$ such that the triples $(K_h,R_h,S_h)$ are mutually independent, independent from $(K_0,R_0,S_0)$, and identically distributed with $\P(K_h=k,R_h=r,S_h=s)=q_{k,r,s}$ for $0\le r\le k$ and $s=0,1$.
We then connect each of the generation-$1$ nodes $v_h$ with $K_h$ directed links pointing to distinct new nodes, and assign to such generation-$2$ nodes $v_{J_1+1},\ldots,v_{J_2}$, where $J_1=K_0$ and $J_2=\sum_{0\le j\le J_1}K_j$, out-degree $\kappa_{v_h}=K_h$, threshold $\rho_{v_h}=R_{h}$, and initial state $\sigma_{v_h}=S_{h}$ such that the triples $(K_h,R_h,S_h)$, for $J_1+1\le h\le J_2$, are mutually independent, independent from $(K_0,R_0,S_0),\ldots,(K_{J_1},R_{J_1},S_{J_1})$, and identically distributed with $\P(K_h=k,R_h=r,S_h=s)=q_{k,r,s}$ for $k\ge0$, $0\le r\le k$, and $s=0,1$. We then keep on repeating the same procedure over and over, thus generating, in a breadth-first manner, a possibly infinite random tree network $\mc T_{\mb p}$ with node set $\mc V=\{v_0,v_1,\ldots\}$, thresholds $\rho_{v_0},\rho_{v_1},\ldots$, and initial states $\sigma_{v_0},\sigma_{v_1},\ldots$. For $t\ge0$, we let $\mc T_{\mb p,t}$ be the finite random tree network obtained by truncating $\mc T_{\mb p}$ at the $t$-th generation. Observe that the specific realization of the two-stage branching process is uniquely determined by the  sequence of mutually independent triples $(K_0,R_0,S_0),(K_1,R_1,S_1),(K_2,R_2,S_2)\ldots\,,$ which are distributed according to $\P(K_0=k,R_0=r,S_0=s)=p_{k,r,s}$ and $\P(K_h=k,R_h=r,S_h=s)=q_{k,r,s}$ for $h\ge1$.

The following result shows that the state $x(t)$ and output $y(t)$ of the recursion \eqref{CM-recursion} coincide with the exact expected states of the LTM dynamics on $\mc T_{\mb p}$. Observe that the LTM dynamics \eqref{LTM-def} is a deterministic process, hence the only randomness concernes the generation of $\mc T_{\mb p}$. 

\begin{proposition}\label{proposition recursive}
Let $\mb p$ be the network statistics and $\mc T_{\mb p}=(\mc V,\mc E,\rho,\sigma)$ be the associated two-stage branching process with node set $\mc V=\{v_0,v_1,\ldots\}$, where $v_0$ is the root node. Let $Z(t)$, for $t\ge0$, be the state vector of the LTM dynamics on $\mc T_{\mb p}$, and let $x(t)$ and $y(t)$ be respectively the state and output of the recursion \eqref{CM-recursion}. Then, for every fixed time $t\geq 0$, the following holds:
\begin{enumerate} 
\item[(i)] For every  $i\in\mc V$, the states $\{Z_{j}(t)\}_{j:\,(i,j)\in\mc E}$ of the offsprings $v_j$ of $v_i$ in $\mc T_{\mb p}$ are independent and identically distributed Bernoulli random variables with expected value $x(t)$;
\item[(ii)] The state $Z_{v_0}(t)$ of the root node $v_0$ is a Bernoulli random variable with expected value $y(t)$.
\end{enumerate}
\end{proposition}
\begin{proof}
(i) First notice that the state $Z_{i}(t)$ of any node $i\in\mc V$ is a deterministic function of the threshold and of the initial states of the descendants of node $i$ in $\mc T_{\mb p}$ up to generation $t$.
It follows that, given any two non-root nodes $j,l\in\mc V\setminus\{v_0\}$, $Z_{j}(t)$ and $Z_{l}(t)$ are Bernoulli random variables with identical distribution, since the two subnetworks of their descendants are branching processes with the same statistics. 
Moreover, for every node $i\in\mc V$, let $\mc N_i$ be the set of its out-neighbors in  $\mc T_{\mb p}$ and observe that the variables $Z_j(t)$, for $j\in\mc  N_i$, are mutually independent since each pair of the subnetworks of their descendants have empty intersection.
Let $\zeta(t)=\E[Z_{j}(t)]$, $j\in\mc V\setminus\{v_0\}$, be the expected value of all these r.v.'s. 
Fix now any $i\in\mc V$ and $j\in\mc N_{i}$. From (\ref{LTM-def}), we obtain
$$\zeta(t+1)=\P\left(\sum\nolimits_{h\in\mc V}A_{jh}Z_h(t)\ge\rho_j\!\right)=
\sum\limits_{k\ge0}\sum_{0\le r\le k}\!\P\left(\sum\nolimits_{h\in\mc N_j}Z_h(t)\ge r\,\Big|\, k_j=k,\, \rho_j=r\right)q_{k,r}\,.$$
Now, observe that the conditional probability in the rightmost summation above is simply the probability that a sum of $k$ independent and identically distributed Bernoulli random variables having mean $\zeta(t)$ is not below the threshold $r$. Therefore, such conditional probability is equal to $\varphi_{k,r}(\zeta(t))$. Substituting we get
$$\zeta(t+1)=\sum\limits_{k\ge0}\sum_{0\le r\le k}\varphi_{k,r}(\zeta(t))q_{k,r}=\phi(\zeta(t))\,.$$
Since $\zeta(0)=\P(Z_j(0)=1)=\P(\sigma_j=1)=x(0)$, it follows that $\zeta(t)=x(t)$ for every $t\ge0$.

(ii) Put $\nu(t)=\E[Z_{v_0}(t)]$. From (\ref{LTM-def}) and point (i), it follows that
\begin{align*}
\nu(t+1)&=\ds\P\left(\sum\nolimits_{h\in\mc V}A_{v_0h}Z_h(t)\ge\rho_{v_0}\right)\\
	&=\ds\sum\limits_{k\ge0}\sum_{0\le r\le k}\P\left(\sum\nolimits_{h\in\mc  N_{v_0}}Z_h(t)\ge\rho_{v_0}\,\Big|\, k_{v_0}=k,\, \rho_{v_0}=r\right)p_{k,r} \\ 
	&=\ds\sum\limits_{k\ge0}\sum_{0\le r\le k}\varphi_{k,r}(\zeta(t))p_{k,r}=\psi(\zeta(t))\,,
\end{align*}
thus completing the proof.
\end{proof}

\subsection{The LTM on the configuration model ensemble}\label{subsect:conc}

We introduce now the configuration model ensemble $\mc C_{n,\mb p}$ of all networks with given size $n$ and statistics $\mb p$. We refer to $\mb p$ and $n$ as compatible if $np_{d,k,r,s}$ is an integer for all non-negative values of $d$, $k$, $0\le r\le k$, and $s\in\{0,1\}$, and $\ov d=\sum_{d,k,r,s}dp_{d,k,r,s}=\sum_{d,k,r,s}kp_{d,k,r,s}$ and
 construct a random network $\mc N=(\mc V,\mc E,\rho,\sigma)$ of compatible size $n$ and statistics $\mb p$ as follows. 
Let $\mc V=\{1,\ldots,n\}$ be a node set and let $\delta$, $\kappa$, $\rho$, and $\sigma$ be a designed vectors of in-degrees, out-degrees, thresholds, and initial states, such that \eqref{def:joint-distribution} holds true, i.e., there is exactly a fraction $p_{d,k,r,s}$ of agents $i\in\mc V$ with $(\delta_i,\kappa_i,\rho_i,\sigma_i)=(d,k,r,s)$. Let $l=\ov d n$ be the number of directed links, put $\mc L=\{1,2,\ldots,l\}$, and let $\nu,\lambda:\mc L\to\mc V$ be  two maps such that $|\nu^{-1}(i)|=\delta_i$ and $|\lambda^{-1}(i)=\kappa_i|$. Then, let $\pi$ be a uniform random permutation of $\mc L$ and let the network $\mc N=(\mc V,\mc E,\rho,\sigma)$ have node set $\mc V$, link multiset $\mc E=\{(\lambda(h),\nu(\pi(h))):\,1\le h\le l\}$, threshold vector $\rho$, and initial state vector $\sigma$. Figure~\ref{fig:CM} illustrates the above construction. We refer to such network $\mc N$ as being sampled from the configuration model ensemble $\mc C_{n,\mb p}$.


\begin{figure}
\centering
\begin{tikzpicture}
	\tikzstyle{main node} = [shape = circle, draw, thick, inner sep = 0cm, minimum size = 0.6cm]

	\def\Dy{1cm};
	\def\DyE{0.3cm};
	\def\edL{1.5cm}
	\def\edh{4mm}
	\def\edhh{2mm}
	\def\Dx{6cm} ;
	\def\arcR{1cm};
	\def\arcA{25};

	\node[main node ] (L1) at (0,0)  {};
	\node[main node ] (R1) at (\Dx,0)  {};
	\draw[->, thick ] (L1) -- ++ (\edL,\edh) node  [right] {$1$};
	\draw[->, thick ] (L1) -- ++ (\edL,0)  node [right] {};
	\draw[->, thick ] (L1) -- ++ (\edL,-\edh) node [right] {};
	\draw[<-, thick ] (R1) -- ++ (-\edL,\edhh) node [left] {$1$};
	\draw[<-, thick ] (R1) -- ++ (-\edL,-\edhh) node [ inner sep = 0] (ee3) [left] {};
	
	\node[main node ] (L2) at (0,-\Dy)  {$i$};
	\node[main node ] (R2) at (\Dx,-\Dy)  {$i$};
	\draw[->, thick ] (L2) -- ++ (\edL,\edhh) node  [ inner sep = 0] (e3) [right] {};
	\draw[->, thick ] (L2) -- ++ (\edL,-\edhh) node [ inner sep = 0] (e2)  [right] {};
	\draw[<-, thick ] (R2) -- ++ (-\edL,0) node [left] {};
	
	\draw [color = red] (L2) ++ (+\arcA:\arcR) arc (+\arcA:-\arcA:\arcR) node [left ] {$\kappa_i$};
	\draw [color = red] (R2) ++ (180-\arcA:\arcR) arc (180-\arcA:180+\arcA:\arcR) node [right] {$\delta_i$};
	
	\node[main node ] (L3) at (0,-2*\Dy)  {$\lambda(h)$};
	\node[main node ] (R3) at (\Dx,-2*\Dy)  {};
	\draw[->,thick  ] (L3) -- ++ (\edL,\edhh) node (h) [right] {$h$};
	\draw[->, thick ] (L3) -- ++ (\edL,-\edhh) node  [right] {};
	\draw[<-, thick ] (R3) -- ++ (-\edL,\edh) node  [ inner sep = 0] (ee2) [left] {};
	\draw[<-, thick ] (R3) -- ++ (-\edL,0) node [left] {};
	\draw[<-, thick ] (R3) -- ++ (-\edL,-\edh) node [ inner sep = 0] (ee1) [left] {};

	\node[main node ] (L4) at (0,-3*\Dy)  {}; 
	\node[main node ] (R4) at (\Dx,-3*\Dy)  {$\nu(j)$};
	\draw[->,thick  ] (L4) -- ++ (\edL,\edhh) node [right] {};
	\draw[->, thick ] (L4) -- ++ (\edL,-\edhh) node [ inner sep = 0] (e1) [right] {};
	\draw[<-, thick ] (R4) -- ++ (-\edL,\edhh) node [left] {};
	\draw[<-, thick ] (R4) -- ++ (-\edL,-\edhh) node (j) [left] {$j$};
	
	\node (dot) at (\Dx/2, -3.8*\Dy - .0*\DyE) {\Large $\ldots$};
	
	\node[main node ] (L5) at (0,-4*\Dy - \DyE)  {}; 
	\node[main node ] (R5) at (\Dx,-4*\Dy - \DyE)  {}; 
	\draw[->, thick ] (L5) -- ++ (\edL,0) node [right] {$l$};
	\draw[<-, thick ] (R5) -- ++ (-\edL,\edhh) node [left] {};
	\draw[<-, thick ] (R5) -- ++ (-\edL,-\edhh) node [left] {$l$};

	\draw[->, thick, color = red] (h) edge [out=5, in = 195] node [red, pos=0.4, sloped, above] {$\pi(h) = j$} (j);

	\draw[dashed, ->, color = red] (e1) to [out=0, in = 200] (ee1) ;
	\draw[dashed, ->, color = red] (e2) to [out=-10, in = 160] (ee2) ;
	\draw[dashed, ->, color = red] (e3) to [out=10, in = 190] (ee3) ;
	
\end{tikzpicture}
\caption{\label{fig:CM} The Configuration Model, with each node represented twice, on the left and on the right side of the picture. The picture contains the edge $(\lambda(h), \nu(\pi(h)) )$ and a few other dashed edges. }
\end{figure}
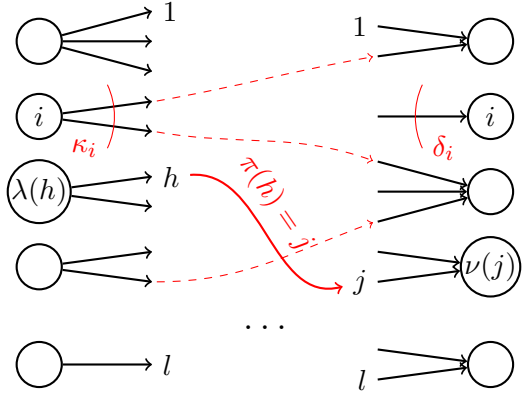

\begin{lemma}\label{lemma:TV}
Let $\mc N$ be a network sampled from the configuration model ensemble $\mc C_{n,\mb p}$ of compatible size $n$ and statistics $\mb p$.
For $t\ge0$, let $\mc N_t$ be the depth-$t$ neighborhood of a node in $\mc N$ chosen uniformly at random from the node set $\mc V$, and let $\mu_{\mc N_t}$ its probability distribution. Let $\mc T_{\mb p,t}$ be a two-stage branching process truncated at depth $t$, and let $\mu_{\mc T_{\mb p,t}}$ be its distribution. Then, the total variation distance $||\mu_{\mc N_t}-\mu_{\mc T_{\mb p,t}}||_{TV}$ between $\mu_{\mc N_t}$ and $\mu_{\mc T_{\mb p,t}}$ satisfies
$$||\mu_{\mc N_t}-\mu_{\mc T_{\mb p,t}}||_{TV}\le\frac{\gamma_t}{2n}\,,\qquad \gamma_t=\frac{d_{\max}k_{\max}^{2t+3}}{\ov d}\,,$$
where $d_{\max}=\max\{d\ge0:\,\sum_{k,r,s}p_{d,k,r,s}>0\}$ is the maximum in-degree and $k_{\max}=\max\{k\ge0:\,\sum_{d,r,s}p_{d,k,r,s}>0\}$ is the maximum out-degree.
\end{lemma}
\proof 

We will construct a coupling of the configuration model $\mc C_{n,\mb p}$ and the two-stage branching process $\mc T$ such that the depth-$t$ neighborhood $\mc N_t$ of a uniform random node in $\mc N$ and the depth-$t$ truncated branching process $\mc T_{\mb p,t}$ satisfy $\P(\mc N_t\ne\mc T_{\mb p,t})\le\gamma_t/n$. The claim will then follow from the well-known bound $||\mu_{\mc N_t}-\mu_{\mc T_{\mb p,t}}||_{TV}\le\P(\mc N_t\ne\mc T_{\mb p,t})$ valid for every coupling of $\mc N_t$ and $\mc T_{\mb p,t}$ (cf., e.g., \cite[Proposition 4.7]{LPW:2009}). 

In order to sample a network $\mc N$ from $\mc C_{n,\mb p}$ and define the coupling altogether, let us assign in-degree $\delta_i$, out-degree $\kappa_i$, threshold $\rho_i$, and initial state $\sigma_i$ to each of the $n$ nodes $i\in\mc V$ in such a way that there are exactly $np_{d,k,r,s}$ nodes of in-degree $d$, out-degree $k$, threshold $r$, and initial state $s$. Let $l=n\ov d=\1'\delta$, $\mc L=\{1,2,\ldots,l\}$, and let $\nu:\mc L\to\mc V$ be a map such that $|\nu^{-1}(i)|=\delta_i$. 
Let $w_0$ be a random node chosen uniformly from $\mc V$, and let $K_0=\kappa_{w_0}$, $R_0=\rho_{w_0}$, and $S_0=\sigma_{w_0}$ be its out-degree, threshold, and initial state, respectively. 
Let $(L_h)_{h=1,2,\ldots}$ be a sequence of mutually independent random variables with identical uniform distribution on the set $\mc L$ and independent from $w_0$. Let $(M_h)_{h=1,2,\ldots,l}$ be a finite sequence of $\mc L$-valued random variables such that, conditioned on $w_0$, $L_1,\ldots,L_h$ and $M_1,\ldots,M_{h-1}$, one has $M_h=L_h$ if $L_h\notin\{M_1,\ldots,M_{h-1}\}$, while, if  $L_h\in\{M_1,\ldots,M_{h-1}\}$, $M_h$ is conditionally uniformly distributed on the set $\mc L\setminus\{M_1,\ldots,M_{h-1}\}$. 
Notice that the marginal probability distributions of the two sequences $(L_h)_{h=1,2,\ldots}$ and $(M_h)_{h=1,2,\ldots,l}$ correspond to sampling with replacement and, respectively, sampling without replacement, from the same set $\mc L$ (note that $(M_h)_{h=1,2,\ldots,l}$ represents a permutation on $\mc L$). 
Moreover,  observe that 
\be\label{P0}\P\left(L_{h+1}\ne M_{h+1} |(L_1,\ldots,L_{h})=(M_1,\ldots,M_{h})\right)\le\frac hl\,,\qquad 1\le h< l\,.\ee

Let $\mc T_{\mb p,t}$ be the random directed tree whose root $v_0$ has out-degree $K_0$, threshold $R_0$ and initial state $S_0$, and that is then generated starting from $v_0$ in a breadth-first fashion, by assigning to each node $v_h$, $h\ge1$ at generation $1\le u\le t$ out-degree $K_h=\kappa_{\nu(L_h)}$, threshold $R_h=\rho_{\nu(L_h)}$ and initial state $S_h=\sigma_{\nu(L_h)}$. Observe that the triples $(K_h,R_h,S_h)$ for $h\ge0$ are mutually independent and have distribution
$\P(K_0=k,R_0=r,S_0=s)=p_{k,r,s}$ and  
$\P(K_h=k,R_h=r,S_h=s)=\frac1{\ov d}\sum_ddp_{d,k,r,s}=q_{k,r,s}$ for $h\ge1$. Hence, $\mc T_{\mb p,t}$ generated in this way has indeed the desired distribution $\mu_{\mc T_{\mb p,t}}$. 

On the other hand, let the network $\mc N$, and hence $\mc N_t$, be generated starting from $w_0$ and exploring its neighborhood in a breadth-first fashion. First let the $J_0=K_0$ outgoing links of $v_0$ point to the nodes $v_1=\nu(M_1),\ldots,v_{J_0}=\nu(M_{J_0})$; then let the $J_1$ links outgoing from the set $\{v_1,\ldots,v_{J_0}\}\setminus\{v_0\}$ of new out-neighbors of $v_0$ point to the nodes $\nu(M_{J_0+1}),\ldots,\nu(M_{J_0+J_1})$; then let the $J_2$ links outgoing from the set $\{v_{{J_0+1}},\ldots,v_{J_0+J_1}\}\setminus\{v_0,v_1,\ldots v_{J_0}\}$ point to the nodes $\nu(M_{J_0+J_1+1}),\ldots,\nu(M_{J_0+J_1+J_2})$, and so on, possibly restarting from one of the unreached nodes in $\mc V$ if the process has arrived to a point where $J_u=0$ and $\sum_{h\le u}J_h<l$ (so that not all nodes have been reached from $v_0$). 
Now, let $H_t=\sum_{0\le u\le t-1}J_u$ and $N_t=|\{v_0,v_1,\ldots,v_{H_t}\}|$ be the total number of links and, respectively, nodes in $\mc N_t$. Observe that $\mc N_t$ is a directed tree if and only if $N_t=H_t+1$, which is in turn equivalent to $ \nu(M_h)\ne\nu(M_{h'})\ne w_0$ for all $1\le h < h' \le N_t$. 

Notice that
\begin{align*}
\P\left(\nu(M_{h+1})\in\{w_0,\nu(M_1),\ldots,\nu(M_h)\}\big|(L_1,\ldots,L_{h+1})=(M_1,\ldots,M_{h+1})\right)\le \frac{(h\!+\!1)(d_{\max}\!-\!1)+1}l\,,
\end{align*}
for $0\le h< l $, which, together with \eqref{P0} gives 
\begin{align*}
\varsigma_h &:=  \P\left(L_{h+1}\ne M_{h+1} \text{ or } \nu(M_{h+1})\in\{w_0,\nu(M_1),\ldots,\nu(M_h)\}\big|(L_1,\ldots,L_{h})=(M_1,\ldots,M_{h})\right)\\
	&\phantom{:} =  \P\left(L_{h+1}\ne M_{h+1}\big|(L_1,\ldots,L_{h})=(M_1,\ldots,M_{h})\right)\\
	&\phantom{:} + \P\left(L_{h+1}= M_{h+1} \text{ and } \nu(M_{h+1})\in\{w_0,\nu(M_1),\ldots,\nu(M_h)\}\big|(L_1,\ldots,L_{h})=(M_1,\ldots,M_{h})\right) \\
	&\phantom{:}\le  \P\left(L_{h+1}\ne M_{h+1}\big|(L_1,\ldots,L_{h})=(M_1,\ldots,M_{h})\right)\\
	&\phantom{:} + \P\left(\nu(M_{h+1})\in\{w_0,\nu(M_1),\ldots,\nu(M_h)\}\big|(L_1,\ldots,L_{h+1})=(M_1,\ldots,M_{h+1})\right)\\
	&\phantom{:}\le  \ds\frac{h}l+\frac{(h+1)(d_{\max}-1)+1}l 
	\le  \ds\frac{(h+1)d_{\max}}l \,.
\end{align*}

The key observation is that, upon identifying node $v_h\in\mc N$ with node $w_h\in\mc T_{\mb p,t}$ for all $0\le h< N_t$, in order for $\mc N_t\ne\mc T_{\mb p,t}$ it is necessary that either $N_t\ne H_t+1$ (in which case $\mc N_t$ is not a tree) or $(L_1,\ldots,L_{H_{t}})\ne(M_1,\ldots,M_{H_{t}})$ (in which case the nodes $v_h$ and $w_h$ might have different outdegree, threshold, or initial state). In order to estimate the probability that any of this occurs, first observe that a standard induction argument shows that $J_u\le k_{\max}^{u+1}$ for all $u\ge0$, so that $H_t\le\sum_{1\le u\le t}k_{\max}^u\le k_{\max}^{t+1}$. 
Then,

$$\ba{rcl}
\P(\mc N_t\ne\mc T_{\mb p,t})&\le&
\P\l((L_1,\ldots,L_{H_{t}})\ne(M_1,\ldots,M_{H_{t}})\text{ or }\bigcup_{1\le h< h'\le H_t}\{\nu(h)=\nu(h')\} \r. \\[10pt]
 & & \hspace{7cm} \l. \text{ or }\bigcup_{1\le h\le H_t}\{\nu(h)= w_0\}\r)\\[10pt]
&\le&
\ds\sum_{0\le h\le H_t-1}\varsigma_h\le
\ds\sum_{0\le h\le k_{\max}^{t+1}-1}\frac{d_{\max}(h+1)}{l}=
\frac{d_{\max}k_{\max}^{t+1}(k_{\max}^{t+1}+1)}{2 n\ov d}\\[10pt]
&\le&\ds\frac{d_{\max}}{2 n\ov d}k_{\max}^{2t+3}\,. 
\ea$$
Hence, the claim follows from the above and the aforementioned bound on the total variation distance between $\mu_{\mc N_t}$ and $\mu_{\mc T_{\mb p,t}}$. 
\qed\medskip

As a consequence of Lemma \ref{lemma:TV}, we get the following result. 

\begin{proposition}\label{prop:mean}
Let $\mc N$ be a network sampled from the configuration model ensemble $\mc C_{n,\mb p}$ of compatible size $n$ and statistics $\mb p$. Let $Z(t)$, for $t\ge0$, be the state vector of the LTM dynamics \eqref{LTM-def} on $\mc N$, $z(t)=\frac1n\sum_{i}Z_i(t)$ be the fraction of state-$1$ adopters at time $t$, and $\ov z(t)=\E[z(t)]$ be its expectation. 
Then, 
$$|\ov z(t)-y(t)|\le \frac{\gamma_t}{2n}\,,$$
where $y(t)$ is the output of the recursion \eqref{CM-recursion} and $\gamma_t={d_{\max}k_{\max}^{2t+3}}/{\ov d}$ as in Lemma \ref{lemma:TV}. 
\end{proposition}
\proof
Observe that, in the LTM dynamics, the state $Z_i(t)$ of an agent $i$ in a network $\mc N=(\mc V,\mc E,\rho,\sigma)$ is a deterministic function of the initial states $Z_j(0)=\sigma_j$ of the agents $j$ reachable from $i$ with $t$ hops in $\mc N$ 
and of the thresholds $\rho_k$ of the agents $k$ reachable from $i$ with less than $t$ hops 
in $\mc N$. In particular, if $\mc N_t^i$ is the depth-$t$ neighborhood of node $i$ in $\mc N$, then $Z_i(t)=\chi(\mc N_t^i)$, where $\chi$ is a certain deterministic $\{0,1\}$-valued function. It follows that, if $\mc N$ is a network sampled from the configuration model ensemble $\mc C_{n,\mb p}$, $\mc N_t$ is the depth-$t$ neighborhood of uniform random node in $\mc N$, and $\mu_{\mc N_t}$ is its distribution, then 
$$\ov z(t)=\E\left[\frac1n\sum_{i\in\mc V}Z_i(t)\right]=\int\chi(\omega)\de \mu_{\mc N_t}(\omega)\,.$$
On the other hand, it follows from Proposition \ref{proposition recursive} that, if $\mc T_{\mb p,t}$ is a two-stage directed branching process with offspring distribution $p_{k,r,s}=\sum_dp_{d,k,r,s}$ for the first generation and $q_{k,r,s}=\frac1{\ov d}\sum_{d\ge0}dp_{d,k,r,s}$ for the following generations, truncated at depth $t$, and $\mu_{\mc T_{\mb p,t}}$ is its distribution, then the output $y(t)$ of the recursion \eqref{CM-recursion} satisfies
$$y(t)=\int\chi(\omega)\de \mu_{\mc T_{\mb p,t}}(\omega)\,.$$
It then follows from the fact that $\chi$ is a $\{0,1\}$-valued random variable and Lemma \ref{lemma:TV} that 
$$|\ov z(t)-y(t)| = \left|\int \l(\chi(\omega)-\frac12 \r)\de \mu_{\mc N_t}(\omega)-\int \l(\chi(\omega)-\frac12\r)\de\mu_{\mc T_{\mb p,t}}(\omega)\right|
	\le ||\mu_{\mc N_t}-\mu_{\mc T_{\mb p,t}}||_{TV}\le \frac{\gamma_t}{2 n}\,,$$
thus completing the proof.
\qed\medskip

The following result establishes concentration of the fraction of state-$1$ adopters in the LTM dynamics on a random network drawn from the configuration model ensemble and its expectation.   
\begin{proposition}\label{prop:concentration}
Let $n$ and $\mb p$ be compatible network size and statistics. Then, for all $\eps>0$, for at least a fraction 
$$1-2e^{-\eps^2\beta n}\qquad \text{with} \qquad \beta=(32 \ov dd_{\max}^{2t})^{-1}$$ 
of networks $\mc N$ from the configuration model ensemble $\mc C_{n,\mb p}$, the fraction of $z(t)=\frac1n\sum_{i\in\mc V}Z_i(t)$ of state-$1$ adopters in the LTM dynamics \eqref{LTM-def} on $\mc N$ satisfies 
$$|z(t)-\ov z(t)|\le\eps /2\,, $$
where $\ov z(t)$ is the average of $z(t)$ over the choice of $\mc N$ from $\mc C_{n,\mb p}$.
\end{proposition}
\proof
Let $a(t)=nz(t)=\sum_{i\in\mc V}Z_i(t)$ be the total number of agents in state $1$ at time $t$ in the network $\mc N$ drawn uniformly from the configuration model ensemble, and let $\ov a(t)=n\ov z(t)$ be its average over the ensemble. In order to prove the result we will construct a martingale $A_0,A_1,\ldots,A_l$, where $l=n\ov d$ is the total number of links, such that 
$A_0=\ov a(t)$, $A_l=a(t)$, and 
\be\label{|A-A|} |A_h-A_{h-1}|\le\alpha\,,\qquad \alpha:=\frac{2d_{\max}^t}{d_{\max}-1}\,,\qquad h=1,2,\ldots,l\,.
\ee
The result will then follow from the Hoeffding-Azuma inequality \cite[Theorem 7.2.1]{AS:2008} which implies that the fraction of networks from the configuration model ensemble for which $|A_0-A_l|\ge \eta= n \eps /2$ is upper bounded by
$$ 2\exp \l(-\frac{\eta^2}{2l\alpha^2} \r)= 2\exp \l(-\frac{n \eps^2}{8 \ov d \alpha^2} \r) 
	= 2\exp \l(-\frac{n \eps^2 (d_{\max}-1)^2 }{32 \ov d d_{\max}^2} \r)\leq 2\exp(-\eps^2\beta n)\,,
$$
where $\beta=(32\ov d d_{\max}^{2t})^{-1}$.

In order to define the aforementioned martingale,  let $\mc L=\{1,2,\ldots,l\}$ and recall that the configuration model ensemble is defined starting from in-degree, out-degree, threshold, and initial state vectors $\delta,\kappa,\rho,\sigma\in\R^n$ with empirical frequency coinciding with the prescribed distribution $\{ \pdkrx \} $ and two maps $\nu,\lambda:\mc L\to\mc V$ such that $|\nu^{-1}(i)|=\delta_i$ and $|\lambda^{-1}(i)|=\kappa_i$ for all $i\in\mc V$. The ensemble is then defined by taking a uniform permutation $\pi$ of the set $\mc L$ and wiring the $h$-th link 
from node $\lambda(h)$ to node $\nu(\pi(h))$ for $h=1,\ldots,l$.
 Let $\pi_{[h]}=(\pi(1),\pi(2),\ldots,\pi(h))$ be the vector obtained by unveiling the first $h$ values of $\pi$.
Then, define $A_h=\E[a(t)|\pi_{[h]}]$, for $h=0,1,\ldots,l$ and observe that $A_0,A_1,\ldots,A_l$ is indeed a (Doob) martingale, generally referred to as the link-exposure martingale. It is easily verified that $A_0=\E[a(t)]=\ov a(t)$ and $A_l=\E[a(t)|\pi]=a(t)$.

What remains to be proven is the bound \eqref{|A-A|}. For a given $h=1,\ldots,l$, let $\tilde\pi$ be a random permutation of $\mc L$ which is obtained from $\pi$ by choosing some $j$ uniformly at random from the set $\mc L\setminus\{\pi(1),\ldots,\pi(h-1)\}$ and putting $\tilde\pi(h)=j$ and $\tilde\pi(\pi^{-1}(j))=\pi(h)$, and $\tilde\pi(k)=\pi(k)$ for all $k\in\mc L\setminus\{h,\pi^{-1}(j)\}$. Notice that $\tilde\pi$ and $\pi$ differ in at most two positions, $h$ and $\pi^{-1}(j)\ge h$, the latter inequality following from the fact that $j\in\mc L\setminus\{\pi(1),\ldots,\pi(h-1)\}$.  Hence, in particular, $\tilde\pi_{[h-1]}=\pi_{[h-1]}$. Moreover, $\tilde\pi$ and $\pi$ have the same conditional distribution given $\pi_{[h-1]}$ (since they both correspond to choosing a bijection of $\{h,h+1,\ldots,l\}$ to $\mc L\setminus\{\pi(1),\ldots,\pi(h-1)\}$ uniformly) and $\tilde\pi$ is conditionally independent from $\pi_{[h]}$ given $\pi_{[h-1]}$. Therefore, 
\be\label{AhAh-1}A_h-A_{h-1}=\E[A(t)|\pi_{[h]}]-\E\left[A(t)|\pi_{[h-1]}\right]=\E[A(t)|\pi_{[h]}]-\E[\tilde A(t)|\pi_{[h-1]}]=\E[A(t)-\tilde A(t)|\pi_{[h]}]\,,\ee
for all $h=1,\ldots,l$.

Now, observe that the value of $\pi(h)$ affects the  depth-$t$ neighborhoods of the node $\lambda(h)$, 
of its in-neighbors, the in-neighbors of its in-neighbors and so on, until those nodes from which $\lambda(h)$ can be reached in less than $t$ hops, for a total of at most 
$$\sum_{s=0}^{t-1}d_{\max}^s=\frac{d_{\max}^t-1}{d_{\max}-1}<\frac{d_{\max}^t}{d_{\max}-1}= c$$ 
nodes in $\mc N$. 
Analogously, the value of $j$ affects the depth-$t$ neighborhoods of the node $\lambda(\pi^{-1}(j))$ 
as well as its in-neighbors, the in-neighbors of its in-neighbors and so on, for a total of less than $c$ 
nodes in $\mc N$. It follows that, if $\tilde A(t)=\sum_i\tilde Z_i(t)$ where $\tilde Z(t)$ is the state vector of the LTM dynamics on the network $\tilde N$ associated to the permutation $\tilde\pi$ in the configuration model, then 
$$|A(t)-\tilde A(t)|\le 2c \,. $$ 
It then follows from \eqref{AhAh-1} and the above that 
$$|A_h-A_{h-1}| \le\l|\E\left[A(t)-\tilde A(t)|\pi_{[h]}\r]\r| 
	\le\E\l[|A(t)-\tilde A(t)||\pi_{[h]}\r]\le 2c\,.$$ 
which proves \eqref{|A-A|}.
The claim then follows from the Hoeffding-Azuma inequality as outlined earlier.
\qed\medskip

By combining Propositions \ref{prop:mean} and \ref{prop:concentration} we get the following result. 

\begin{theorem}\label{theo:concentration}
Let $\mc N$ be a network sampled from configuration model ensemble $\mc C_{n,\mb p}$ of size $n$ and statistics $\mb p$. Let $Z(t)$, for $t\ge0$ be the state vector of the LTM dynamics \eqref{LTM-def} on $\mc N$, let $z(t)=\frac1n\sum_{i}Z_i(t)$, and let $y(t)$ be the output of the recursion \eqref{CM-recursion}. 


Then, for $\eps >0$ and $n \ge \gamma_t /\eps$ where $\gamma_t={d_{\max}k_{\max}^{2t+3}}/{\ov d}$, it holds true

$$|z(t)-y(t)|\le\eps $$
for all but at most a fraction $2e^{-\eps^2\beta n}$ of networks $\mc N$ from $\mc C_{n,\mb p}$, where $\beta=(32\ov dk_{\max}^{2t})^{-1}$. 
\end{theorem}
\proof
Proposition \ref{prop:concentration} implies that $|z(t)-\ov z(t)|\le\eps/2$ for all but at most a fraction $2e^{-\eps^2\beta n}$ of networks from $\mc C_{n,\mb p}$. On the other hand, Proposition \ref{prop:mean} implies that $|\ov z(t)-y(t)|\le\eps/2$ for $\gamma_t/n \le \eps$. 
\qed\medskip


\subsection{Extentions}
We conclude this section by discussing how Theorem \ref{theo:concentration} can be extended to including two variants of the model:  undirected configuration model and time-varying thresholds. 

\subsubsection{PLTM on the undirected configuration model ensemble} \label{sect:undirectedCM}
While Theorem \ref{theo:concentration} concerns the approximation of the average fraction of state-$1$ adopters in the LTM dynamics for most networks in the directed configuration model ensemble $\mc C_{n,\mb p}$, for the PLTM only the result can be extended to the undirected configuration model ensemble as defined below. 

Let $u_{k,r,s}=p_{k,k,r,s}$ for $k\ge0$, $0\le r\le k$, and $s\in\{0,1\}$, denote the fraction of agents of degree $k$, threshold $r$ and initial state $s$ in an undirected network. We shall refer to $\mb u=\{u_{k,r,s}\}$ as undirected network statistics. 
A network size $n$ and undirected network statistics $\mb u$ are said to be compatible if $nu_{k,r,s}$ is an integer for all $0\le r\le k$ and $s=0,1$, and $l=\sum_{k\ge0}\sum_{0\le r\le k}\sum_{s=0,1}nku_{k,r,s}$ is even. 
For compatible undirected network statistics $\mb u$ and size $n$, let $\mc V=\{1,\ldots,n\}$ be a node set and let $\kappa$, $\rho$, and $\sigma$ be  designed vectors of degrees, thresholds, and initial states, such that there is exactly a fraction $u_{k,r,s}$ of agents $i\in\mc V$ with $(\kappa_i,\rho_i,\sigma_i)=(k,r,s)$. Put $\mc L=\{1,2,\ldots,l\}$, and let $\lambda:\mc L\to\mc V$ be  a map such that $|\lambda^{-1}(i)|=\kappa_i$ for all agents $i\in\mc V$. Let $\pi$ be a uniform random permutation of $\mc L$ and let the network $\mc N=(\mc V,\mc E,\rho,\sigma)$ have node set $\mc V$, link multiset $\mc E=\{(\lambda(\pi(2h-1)),\lambda(\pi(2h))),(\lambda(\pi(2h)),\lambda(\pi(2h-1)))\,:\,1\le h\le l/2\}$, threshold vector $\rho$, and initial state vector $\sigma$. Observe that, for every realization of the permutation $\pi$, the resulting network $\mc N$ is undirected, has size $n$ and statistics $\mb u$. We refer to such network $\mc N$ as being sampled from the undirected configuration model ensemble $\mc M_{n,\mb u}$.

The key step for extending Theorem \ref{theo:concentration} to the PLTM dynamics on undirected configuration model ensemble $\mc M_{n,\mb u}$ is the following result showing that the PLTM dynamics on a rooted undirected tree coincides with PLTM dynamics on the directed version of the tree. 

\begin{lemma}\label{lemma:undirected-tree}
For every network $\mc T=(\mc V,\mc E,\rho,\sigma)$ with undirected tree topology and  
every node $i\in\mc V$,  the state vector $Z(t)$ of the PLTM dynamics \eqref{PLTM-def} on $\mc T$ satisfies 
$$\label{Tt} Z_i(t)= Z^{(i)}_i(t)\,,\qquad t\ge0\,,$$
where 
$Z^{(i)}(t)$ is the state vector of the PLTM dynamics on the network 
$\ora{\mc T_{(i)}}=(\mc V,\ora{\mc E_{(i)}},\rho,\sigma)$ with directed tree topology rooted in $i$, obtained from $\mc T$ by making all its links directed from nodes at lower distance from $i$ to nodes at higher distance from it. 
\end{lemma}
\proof
We proceed by induction on $t$. The case $t=0$ is trivial as the initial condition is the same $Z_i(0)=\sigma_i=Z_i^{(i)}(0)$ for all $i\in\mc V$. 
Now, assuming that, for some given $t\ge0$, the PLTM dynamics on every network with undirected tree topology satisfies 
$$Z_i(t)=Z^{(i)}_i(t)\,,\qquad \forall i\in\mc V$$ 
we will prove that 
$$Z_i(t+1)=Z^{(i)}_i(t+1)\,,\qquad \forall i\in\mc V$$ 
for all networks with undirected tree topology $\mc T=(\mc V,\mc E,\rho,\sigma)$. 
We separately deal with the two cases: (a) $Z_i(t)=Z^{(i)}_i(t)=1$; and (b) $Z_i(t)=Z^{(i)}_i(t)=0$.  
Since we are considering the PLTM dynamics, case (a) is easily dealt with, as $Z_i(t)=1=Z^{(i)}_i(t)$ implies 
$Z_i(t+1)=1=Z^{(i)}_i(t+1)$. On the other hand, in order to address case (b), let $\mc J$ be the set of 
neighbors of $i$ in $\mc T$, which coincides with the set of offsprings of node $i$ in $\ora{\mc T_{(i)}}$. 
For every $j\in\mc J$, let $\ora{\mc T_{(i,j)}}=(\mc V_{(i,j)},\ora{\mc E_{(i,j)}},\sigma,\rho)$ be the network obtained by restricting $\ora{\mc T_{(i)}}$ to node $j$ and all its offsprings, let ${\mc T}_{(i,j)}=(\mc V_{(i,j)},\mc E_{(i,j)},\sigma,\rho)$ be the undirected version of $\ora{\mc T_{(i,j)}}$, and let $W(t)$ and $W^{(j)}(t)$ be the vector states of the PLTM dynamics on ${\mc T}_{(i,j)}$ and $\ora{\mc T_{(i,j)}}$, respectively. Now, note that $Z^{(i)}_j(t)=W^{(j)}_j(t)$, since $j$ has the same $t$-depth neighborhood in the two networks.
On the other hand, note that, if the state of the PLTM dynamics on $\mc T$ is such that $Z_i(t)=0$, then $Z_i(s)=0$ for all $0\le s\le t$, so that the state of node $j$ in the PLTM dynamics on $\mc T$ depends only on the thresholds $\rho_h$ and the initial states $\sigma_h$ of agents $h\in\mc V_{(i,j)}$, and is the same as the state of node $j$ in PLTM dynamics on the original network $\mc T_{(i,j)}$, i.e., $Z_j(t)=W_j(t)$. Finally, observe that the inductive assumption applied to the restricted network $\mc T_{(i,j)}$ implies that $W_j(t)=W^{(j)}_j(t)$.
It then follows that, if $Z_i(t)=Z^{(i)}_i(t)=0$, then 
$$Z_j(t)=W_j(t)=W^{(j)}_j(t)=Z^{(i)}_j(t)\,,\qquad \forall j\in\mc J\,.$$
This implies, by the structure of the recursive equation \eqref{PLTM-def} that $Z_i(t+1)=Z_i^{(i)}(t+1)$. This completes the proof.
\qed

Using Lemma \ref{lemma:undirected-tree} it is straightforward to extend Proposition  \ref{proposition recursive} to the undirected two-stage branching process. Then, the results in Section \ref{subsect:conc} carry over to the undirected configuration model ensemble without signficant changes, leading the following result. 

\begin{theorem}\label{theo:concentration2}
Let $\mc N$ be a network sampled from the undirected configuration model ensemble $\mc M_{n,\mb u}$ of size $n$ and statistics $\mb u$. Let $Z(t)$, for $t\ge0$ be the state vector of the PLTM dynamics \eqref{PLTM-def} on $\mc N$, let $z(t)=\frac1n\sum_{i}Z_i(t)$, and let $y(t)$ be the output of the recursion \eqref{CM-recursion}. 
Then, for $\eps >0$ and $n \ge \gamma_t /\eps$ where $\gamma_t={k_{\max}^{2t+4}}/{\ov k}$, it holds true 
$$|z(t)-y(t)|\le\eps $$
for all but at most a fraction $2e^{-\eps^2\beta n}$ of networks $\mc N$ from the $\mc M_{n,\mb u}$, where $\beta=(32\ov kk_{\max}^{2t})^{-1}$. 
\end{theorem}

We stress the fact that the proposed extension of the approximation results for the undirected configuration model ensemble is strictly limited to the PLTM and does not apply to the general LTM. The key step where the structure of the PLTM model is used is in the proof of Lemma \ref{lemma:undirected-tree} which allows one to reduce the study of the PLTM on undirected trees to the one of PLTM on directed trees. An analogous results does not hold true for the LTM without permanent activation and indeed the analysis on undirected trees is known to face relevant additional challenges as illustrated in \cite{KM:2011} for the majority dynamics (that can be considered a special case of the LTM).

\subsubsection{Time-varying thresholds}
We first observe that, while we have not made it explicit yet, all the results discussed in this section carry over, along with their proofs, also for networks with time-varying thresholds $\rho_i(t)$. In this case, the network statistics
$$p_{d,k,r,s}(t)=\frac1n\left|\{i\in\mc V:\,\delta_i=d,\,\kappa_i=k,\,\rho_i(t)=r,\,\sigma_i=s\}\right|\,,\qquad d\ge0\,,\ 0\le r\le k\,,\ s=0,1\,,$$
become time-varying, and so do their marginals 
\be\label{def:degree-distribution-t}p_{k,r}(t):=\sum_{d\ge0}\sum_{s=0,1}p_{d,k,r,s}(t)\,,\qquad q_{k,r}(t):=\frac{1}{\ov d}\sum_{d\ge0}\sum_{s=0,1}dq_{d,k,r,s}(t)\,,\qquad k,r\ge0\,.\ee
In contrast, the marginal $p_{d,k,s}=\sum_{0\le r\le k}p_{d,k,r,s}(t)$ remain constant in time since so do the degrees $\delta_i$ and $\kappa_i$ and the initial states $\sigma_i$ of all agents $i$. 
For networks with such time-varying thresholds, Theorem \ref{theo:concentration} continues to hold true provided that $y(t)$ is interpreted as the  output of the modified recursion 
\be\label{TV-recursion}x(t+1)=\phi(x(t),t)\,,\qquad y(t+1)=\psi(x(t),t)\,,\qquad t\ge0\,,\ee
where
$$\phi(x,t):=\sum_{k\ge0}\sum_{r\ge0} q_{k,r}(t)\varphi_{k,r}(x)\,,\qquad \psi(x,t):=\sum_{k\ge0}\sum_{r\ge0} p_{k,r}(t)\varphi_{k,r}(x)\,.$$

A note of caution concerns extensions of Lemma \ref{lemma:LTM=PLTM} to networks with time-varying thresholds. This result, allowing one to identify the LTM dynamics with the \textit{progressive} LTM (PLTM) dynamics whenever the condition $\rho_i\le \delta_i(1-\sigma_i)$ is met for all agents $i$, continues to hold true for time-varying networks only with the additional assumption that the thresholds are monotonically non-increasing in time, i.e., $\rho_i(t+1)\le\rho_i(t)$ for every node $i$ and time instant $t\ge0$. It is also worth stressing that the analysis of Section \ref{sec:LTM-BR} for the asymptotic behavior of the recursion \eqref{CM-recursion} does not carry over as such to the time-varying case \eqref{TV-recursion}.

\section{Numerical simulations on a real large-scale social network} 
\label{sec:simulations}

We test the prediction capability of our theoretical results for the Linear Threshold Model (LTM) on a real large-scale online social network. We consider the directed interconnection topology of the online social network \url{Epinions.com}, we endow each node with a threshold and  assign an initial states, then run the LTM \eqref{LTM-def} and compare the results with the predictions obtained using the recursion \eqref{CM-recursion}. 

The online social network \url{Epinions.com} was a general consumer review website with a community of users, operating from 1999 until 2014. 
The members of the community were encouraged to submit product reviews for any of over one hundred thousand products, to rate other reviews and to list the reviewers they trusted. 
The directed graph of trust relationships between users, called the ``Web of Trust'', was used in combination with the review's ratings to determine which reviews were shown to the user.  
Being highly connected and containing cycles, the graph remains an interesting source for experiments on social networks and viral marketing \cite{ RAD:2003:snap-source, RD:2002:miningKS}.

The entire ``Web of Trust'' directed graph $\mc G = (\mc V,\mc E)$ of the \url{Epinions.com} social network was obtained by crawling the website \cite{RD:2002:miningKS} and is available from the online collection \cite{SNAP:2014}. 
The dataset\footnote{Retrieved from \url{http://snap.stanford.edu/data/soc-Epinions1.html}. The page contains further informations and statistics about the dataset and mentions \cite{RAD:2003:snap-source} as original source. Further statistics can be retrieved from \url{http://konect.uni-koblenz.de/networks/soc-Epinions1}.} is a list of directed edges expressed as pairs $(i,j)$, representing the \textit{who-trust-whom} relations between users: the list contains \numprint{508 837} directed edges corresponding to $n = \numprint{75 879}$ different users. There are no other information for the LTM (e.g. thresholds or initial states).

From the dataset topology, we computed the empirical joint degree statistic  
$$p_{d,k} = n^{-1} \l| \{i : \delta_i = d, \kappa_i = k \} \r| \,,$$
i.e., the fractions of nodes with in-degree $d$ and out-degree $k$.
Figure~\ref{fig:epinions-empirical-stat} represents the in-degree statistics $p_d = \sum_k p_{d,k}$ and the out-degree statistic $p_k =\sum_d p_{d,k}$; both  follow an approximate power law distribution with exponent $\approx 1.6$.
A few nodes have no in-neighbors or out-neighbors, while the maximum in-degree is \numprint{3035} and the maximum out-degree is \numprint{1801}. The average in/out-degree is {6.705}.

We also computed the fraction of links pointing to nodes with given in-degree $d$  and out-degree $k$, i.e. in-degree weighted, joint degree statistic $q_{d,k} = d p_{d,k} / \ov d$. The values of the joint degree statistics $p_{d,k}$ and $q_{d,k}$, in the interval $d, k \in \{0,1,\ldots,150\}$, are represented by a logarithmic grayscale  in Figure~\ref{fig:epinions-joint-pdk-qdk},  showing a mild correlation between in-degree and out-degree.

\begin{figure}
	\centering
	\includegraphics[trim={\figtriml} {60mm} {\figtrimr} {\figtrimt},clip, width={\figwidthduo},
		keepaspectratio=true]{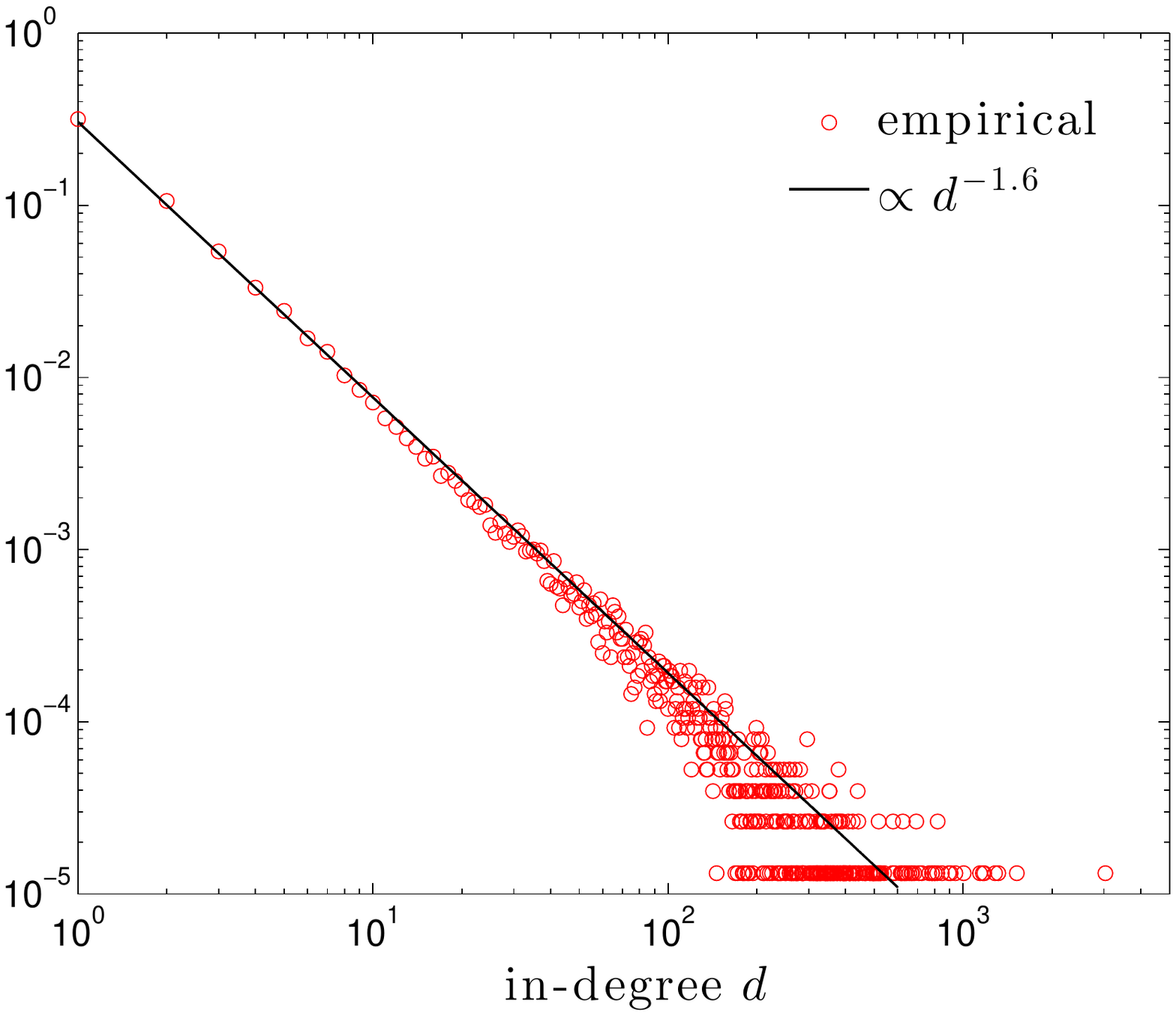}%
	\hspace{1cm}%
	\includegraphics[trim={\figtriml} {60mm} {\figtrimr} {\figtrimt},clip, width={\figwidthduo},
		keepaspectratio=true]{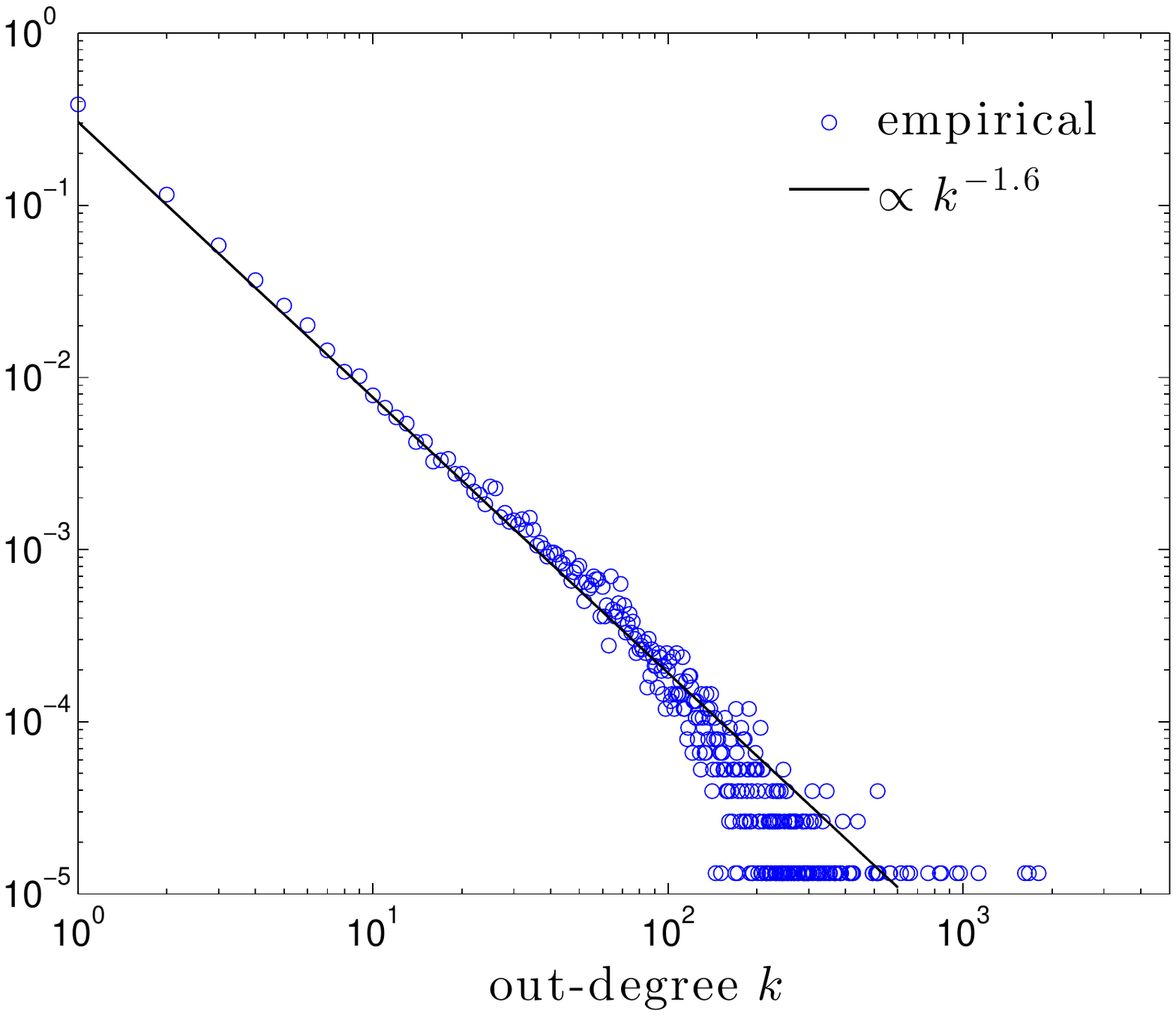}%
	\caption{\label{fig:epinions-empirical-stat} 
	The degree statistics of the \url{Epinions.com} social network. 
	The left plot represents the in-degree statistic $p_d = \sum_k p_{d,k}$ with red circles; the black solid line is proportional to $d^{-1.6}$. Not represented in the logarithmic plot, a fraction $p_0 = 0.315 $ of nodes has no in-neighbors. 
	The right plot represents the out-degree statistic $p_k = \sum_d p_{d,k}$ with blue circles; the black solid line is proportional to $d^{-1.6}$. A fraction $p_0 = 0.205$ of nodes (not represented)  has no out-neighbors.  }
\end{figure}

\begin{figure}
	\centering
	\includegraphics[trim={\figtriml} {60mm} {\figtrimr} {\figtrimt},clip, width={\figwidthduo},
		keepaspectratio=true]{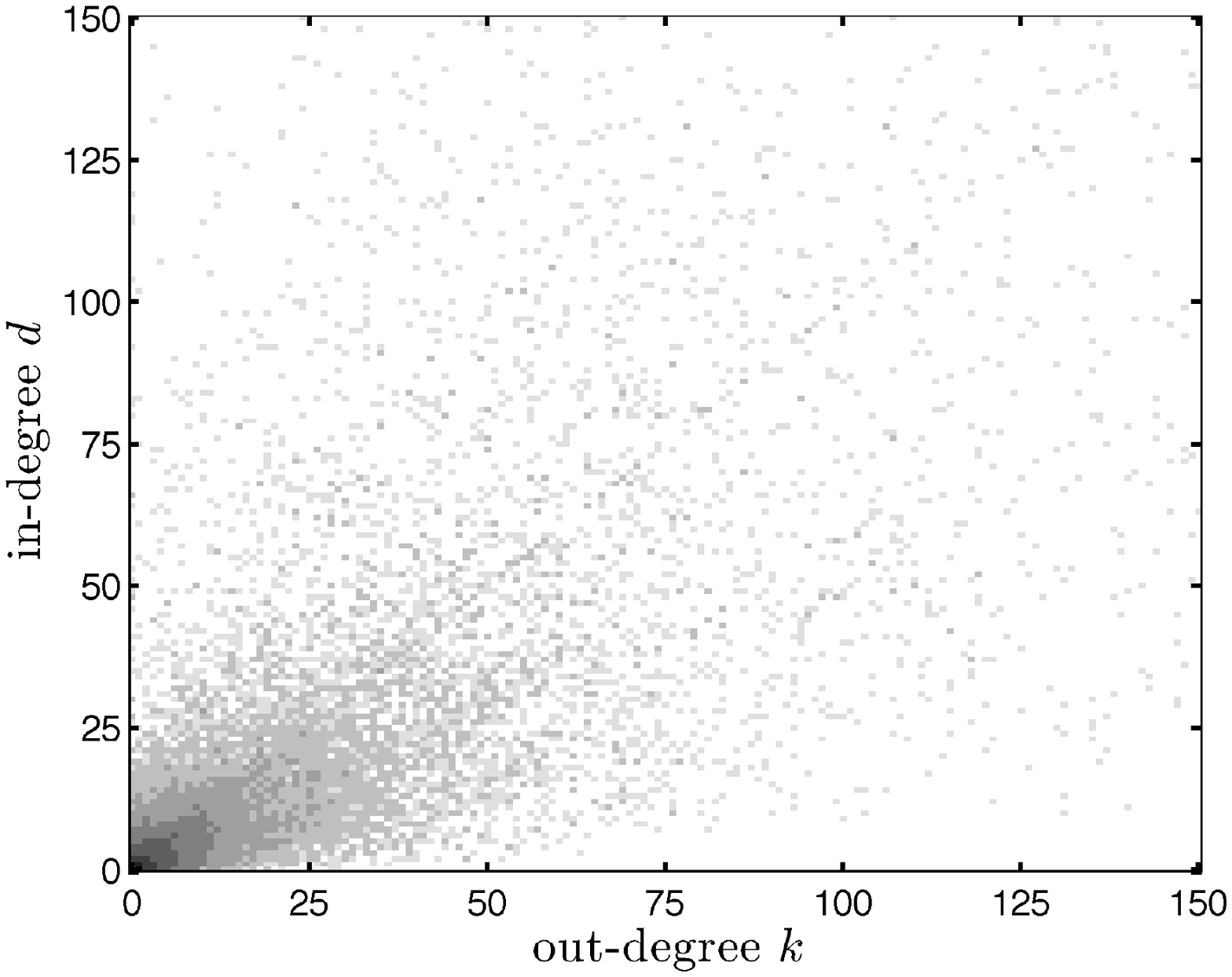}%
	\hspace{1cm}%
	\includegraphics[trim={\figtriml} {60mm} {\figtrimr} {\figtrimt},clip, width={\figwidthduo},
		keepaspectratio=true]{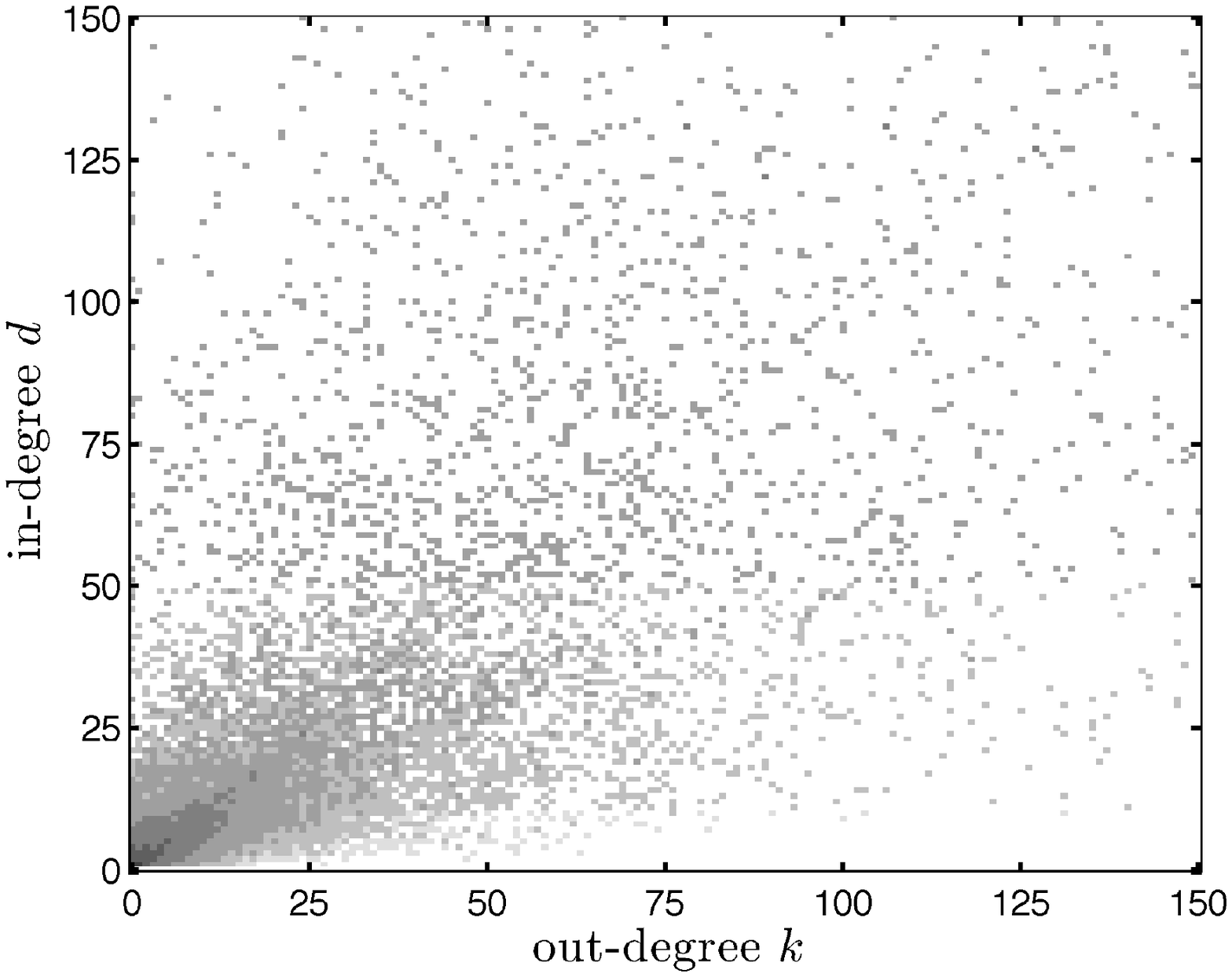}%
	\caption{\label{fig:epinions-joint-pdk-qdk} 
	A pictorial representation, with a logarithmic grayscale, of the joint degree statistics' values $p_{d,k}$ (left plot) and of the in-degree weighted, joint degree statistics' values $q_{d,k} = d p_{d,k} / \ov d$ (right plot) in the interval $0 \leq d, k \leq 150$, for the \url{Epinions.com} social network.
	Note that, the values of $p_{d,k}$ in the interval $0 \leq d, k \leq 150$ (i.e. those represented in the left plot) add to 99.0\% of the full statistics; for $q_{d,k}$ they add to 64.7\% only (right plot).}
\end{figure}

\subsection{The simulations and comparison with the recursion}

The dataset contains the interconnection topology $\mc G=(\mc V,\mc E)$ of the \url{Epinions.com} social network, but no information about thresholds and initial condition for an hypothetical LTM process. Hence, to simulate the LTM we have to combine the topology with a vector of thresholds and initial states. 
In this subsection, we describe how we choose the missing information and present three simulations.

First, we consider a vector $\Theta \in [0,1]^n$, with $n = |\mc V|$, of normalized thresholds with given cumulative distribution function $F(\theta):=\frac1n|\{i:\,\Theta_i\le\theta\}|$, such that $F(\theta)$ is non-decresing, right-continuous, with $F(\theta)=0$ for $\theta<0$ and $F(\theta)=1$ for $\theta\ge1$.
Given the fraction $\ups\in [0,1]$, we also consider the binary vector $\Sigma \in \{0,1\}^{n}$ such that $\ups = \frac1n \sum_i \Sigma_i$, i.e. a fraction $\ups$ of entries is equal to one. 

Then we define the network $\mc N=(\mc V,\mc E,\rho,\sigma)$ as follows. The set of agents $\mc V$ and the set of links $\mc E$ are those of the \url{Epinions.com} dataset. 
Let $\pi'$ and $\pi''$ be two independent and uniformly chosen permutation on the set $\mc V=\{1,2,\ldots,n \}$. 
The threshold vector $\rho$ has entries $\rho_i = \lceil \Theta_{\pi'(i)}\kappa_i \rceil$ where $\kappa_i$ is the out-degree of node $i$, i.e. the threshold vector corresponds to a permutation of the normalized threshold vector.  
The vector of initial states $\sigma$ has components $\sigma_i = \Sigma_{\pi''(i)}$, i.e. is a permutation of $\Sigma$. 
Given the network $\mc N=(\mc V,\mc E,\rho,\sigma)$, the LTM \eqref{LTM-def} is a deterministic process: we compute the evolution of the configuration $Z(t) \in \{0,1\}^n$ (which may not converge) until a fixed time horizon $T$. 
%
From the configuration $Z(t)$ we compute the fraction of state-$1$ adopters at time $t$, i.e. $z(t):=\frac1n\sum_{i}Z_i(t)$. 
To discuss the simulations, we also compute the fraction of links pointing to state-$1$ adopters at time $t$, i.e. $$a(t):=\frac{1}{|\mc E|}\sum_{i}\delta_i Z_i(t)\,,$$
where $\delta_i$ is the in-degree of node $i$.
For a given cumulative distribution function $F(\theta)$ of the normalized threshold and fraction $\ups$ of initially active nodes, we repeat a few times the extraction of the permutations $\pi'$ and $\pi''$ (that establish the specific thresholds and initial states assignment) and the computation of the LTM evolution.

We will compare the simulations with the prediction obtained by the recursion \eqref{CM-recursion}. The recursion requires the \emph{network's statistics}  $\mb p=\{ \pdkrx \}$, as defined in \eqref{def:joint-distribution}, and the initial condition $\xi$ defined in \eqref{upsxi}.
We stress that in the simulations we assign the normalized thresholds and the initial condition using two permutations $\pi'$ and $\pi''$ chosen independently and uniformly at random among those over the set $\{1,2,\ldots, n\}$. 
Hence, \textit{a priori}, the  elements of network's statistics $\mb p$ take the form
\begin{align} \label{eq:sim-joint-distrib}
\l\{\begin{array}{ll}
 p_{d,k,r,s}  = p_{d,k} 
	\l( F( \textstyle{\frac{r}{k}} )-F( \textstyle{\frac{r-1}{k}} )\r) 
	\l( \ups\mathbbm{1}_{s=1}  + (1-\ups)\mathbbm{1}_{s=0}\r) 
	& d\ge0,\, k>0, \, 0\le r\le k,\, s=0,1\,\\
 p_{d,0,0,s}  = p_{d,0}  
	\l( \ups\mathbbm{1}_{s=1}  + (1-\ups)\mathbbm{1}_{s=0}\r) 
	& d\ge0,\, s=0,1 
\end{array} \r. \end{align}
where $p_{d,k}$ is the joint degree statistic corresponding the  \url{Epinions.com} graph $\mc G = (\mc V, \mc E)$. 
Consequently, \textit{a priori}, we obtain the values of the fractions $p_{k,r}$ and $q_{k,r}$, that enter in the definition of the recursion's functions $\phi(x)$ and $\psi(x)$, by plugging in their definitions \eqref{def:degree-distribution}
the above \textit{a priori} network's statistics $\mb p$ \eqref{eq:sim-joint-distrib}.
Finally, the seed $\xi$, initial condition of the recursion, \textit{a priori} coincides with $\ups$, i.e. $\xi = \ups$, because the permutation $\pi''$ is independent from the in-degree of the nodes.

In the following we describe three group of simulations. We will denote with $h(x)$ the right-continuous unit step function 
$$ h(x) :=\l\{ \begin{array}{ll} 0, & x<0 \\ 1, & x\geq 0\,. \end{array} \r. $$

\begin{example} 

\begin{figure}
	\centering
	\includegraphics[trim={\figtrimla} {\figtrimba} {\figtrimra} {\figtrimta},clip, width={\figwidthduo},
		keepaspectratio=true]{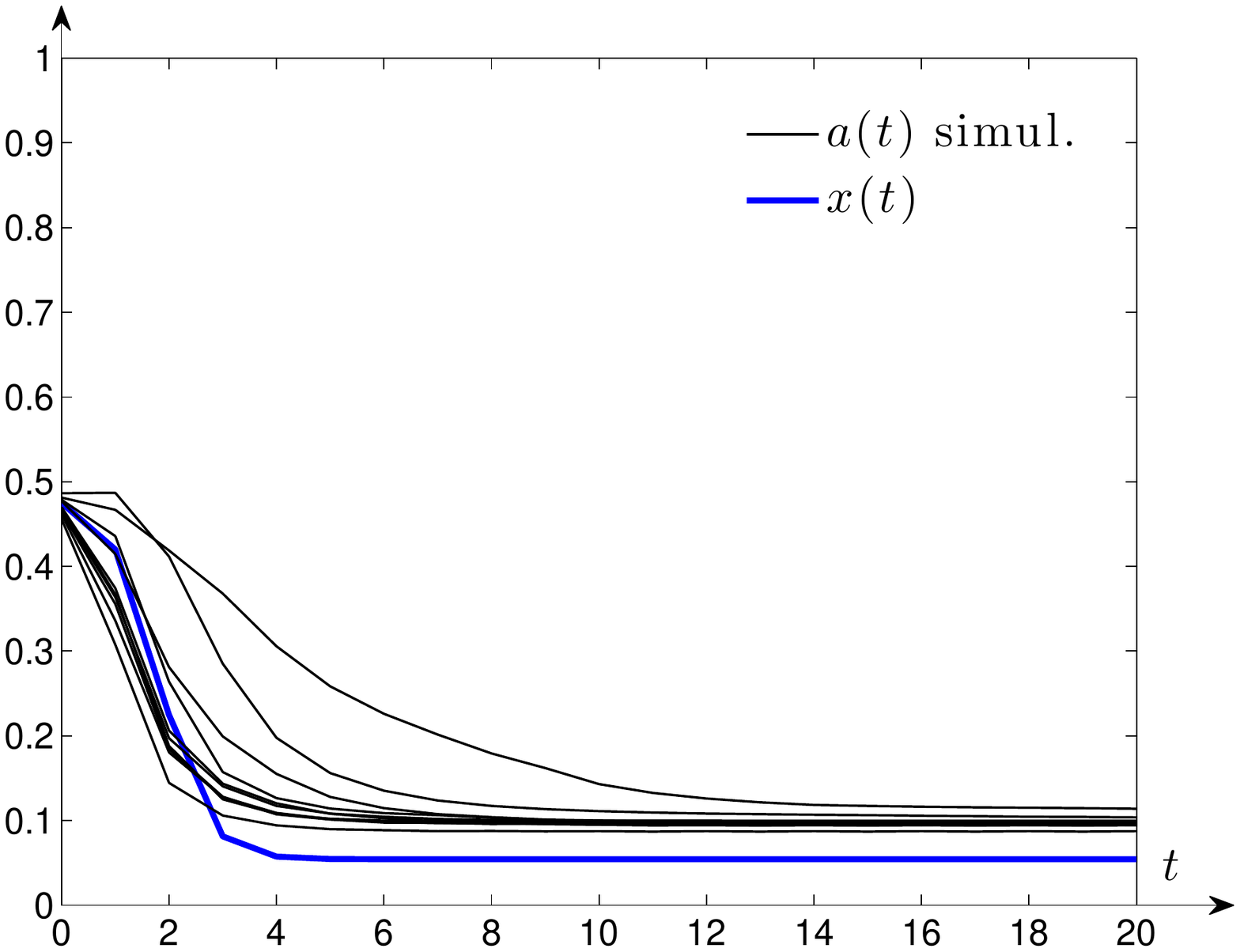}%
	\hspace{1cm}%
	\includegraphics[trim={\figtrimla} {\figtrimba} {\figtrimra} {\figtrimta},clip, width={\figwidthduo},
		keepaspectratio=true]{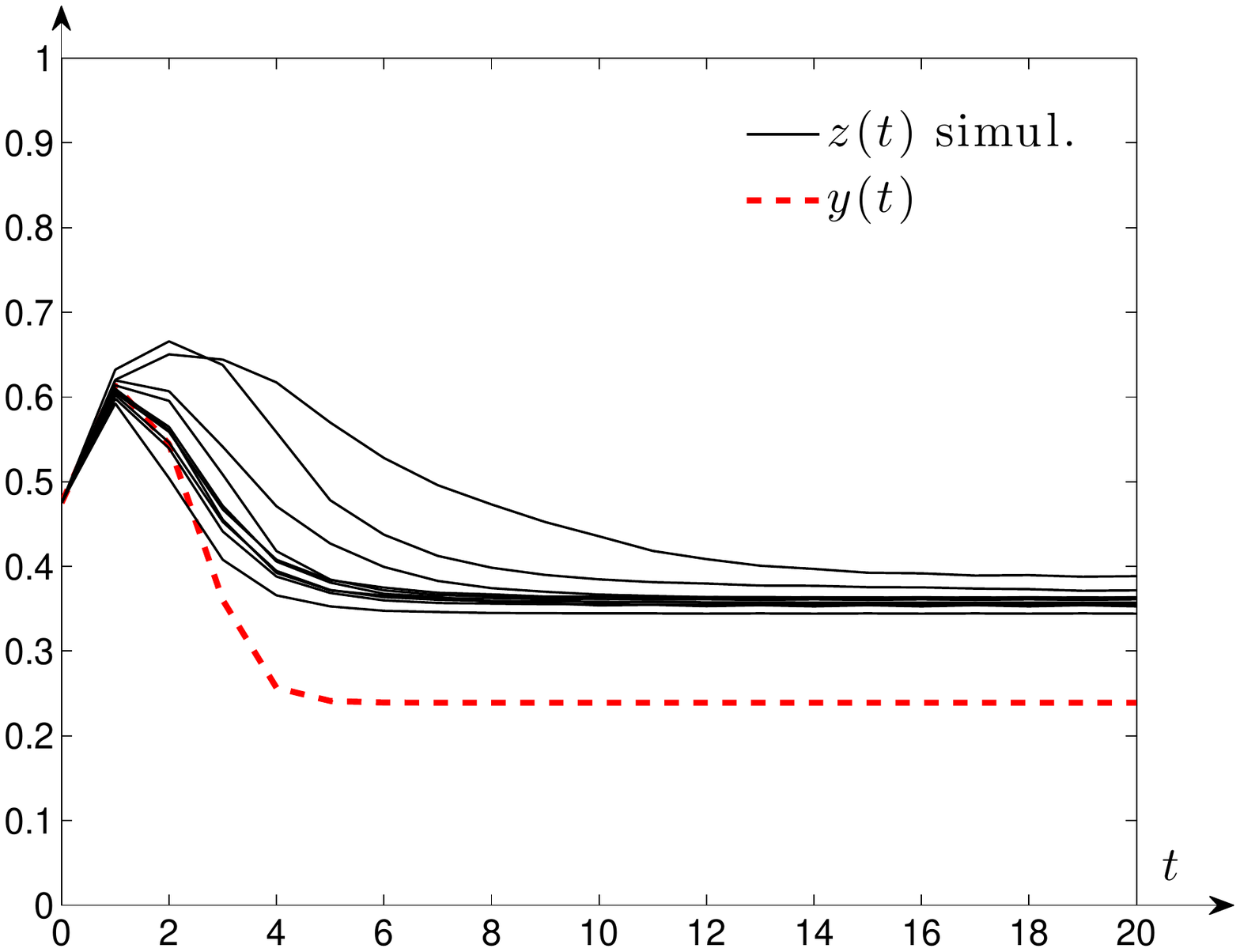}%
	\caption{\label{fig:epinions-sigle-dyn} Simulations of the LTM dynamics on the \url{Epinions.com} topology, with agents endowed by the thresholds $\rho_i = \lceil \frac12 \kappa_i \rceil$. The initial states are randomly selected, conditioned on a fraction $\ups = 0.475$ of nodes having $\sigma_i = 1$. The left plot compares the simulations of the fraction of links pointing to state-$1$ adopters $a(t)$ (thin black lines) with the recursion's state dynamic $x(t)$ (thick blue line), initialized with seed $\xi = \ups$. Simulations and recursion agree fairly well.
	The right plot reports the simulated fraction of state-$1$ adopters $z(t)$ (thin black lines) to be compared with the  recursion's output dynamic $y(t)$ (dashed red line).
	The recursion captures the qualitative behavior of the simulation, with a mismatch of about 15\% in the settling value. A close look reveals that several simulations show a little ripple with period two. }
\end{figure}

\begin{figure}
	\centering
	\includegraphics[trim={\figtrimla} {\figtrimba} {\figtrimra} {\figtrimta},clip, width={\figwidthduo},
		keepaspectratio=true]{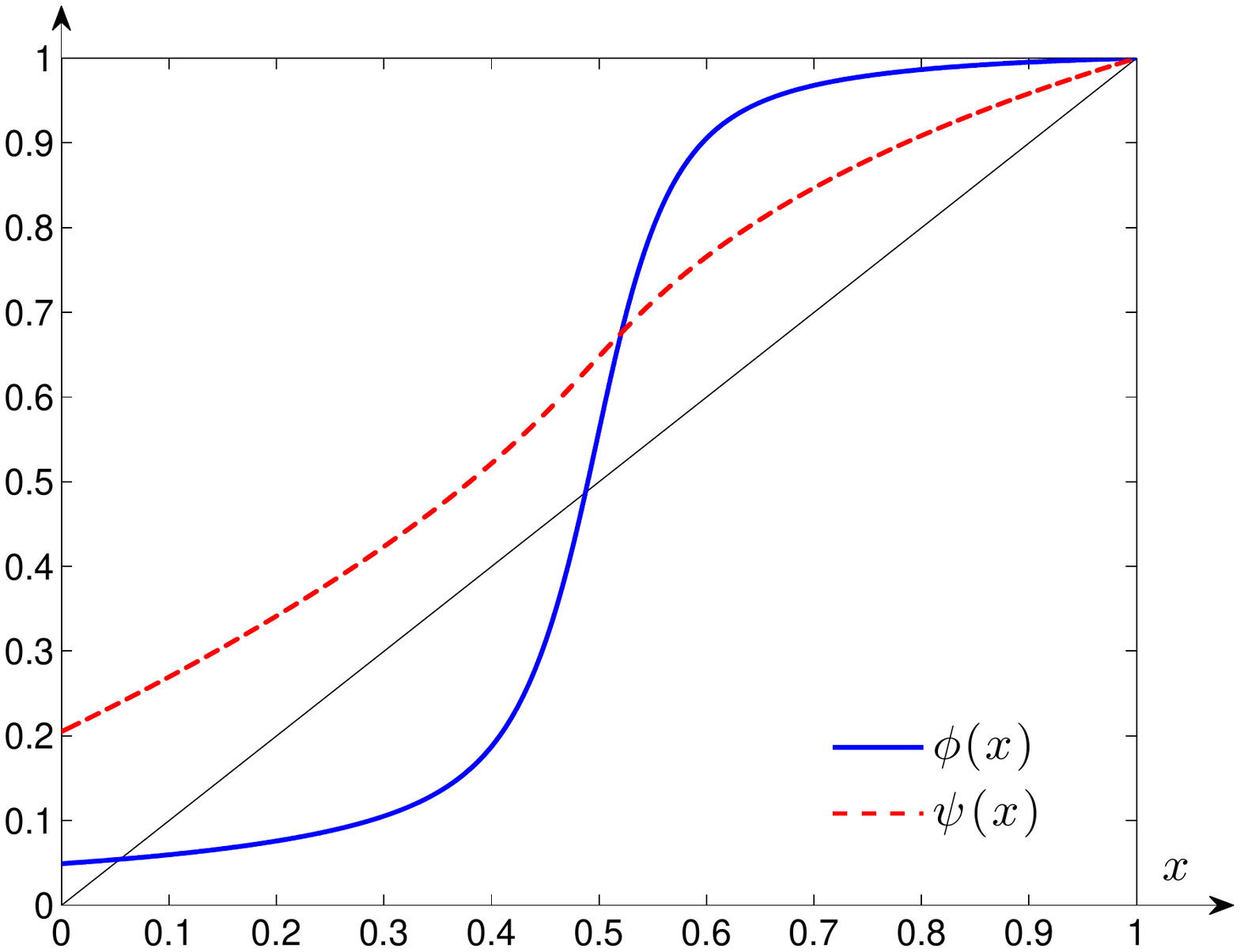}%
	\hspace{1cm}%
	\includegraphics[trim={\figtrimla} {\figtrimba} {\figtrimra} {\figtrimta},clip, width={\figwidthduo},
		keepaspectratio=true]{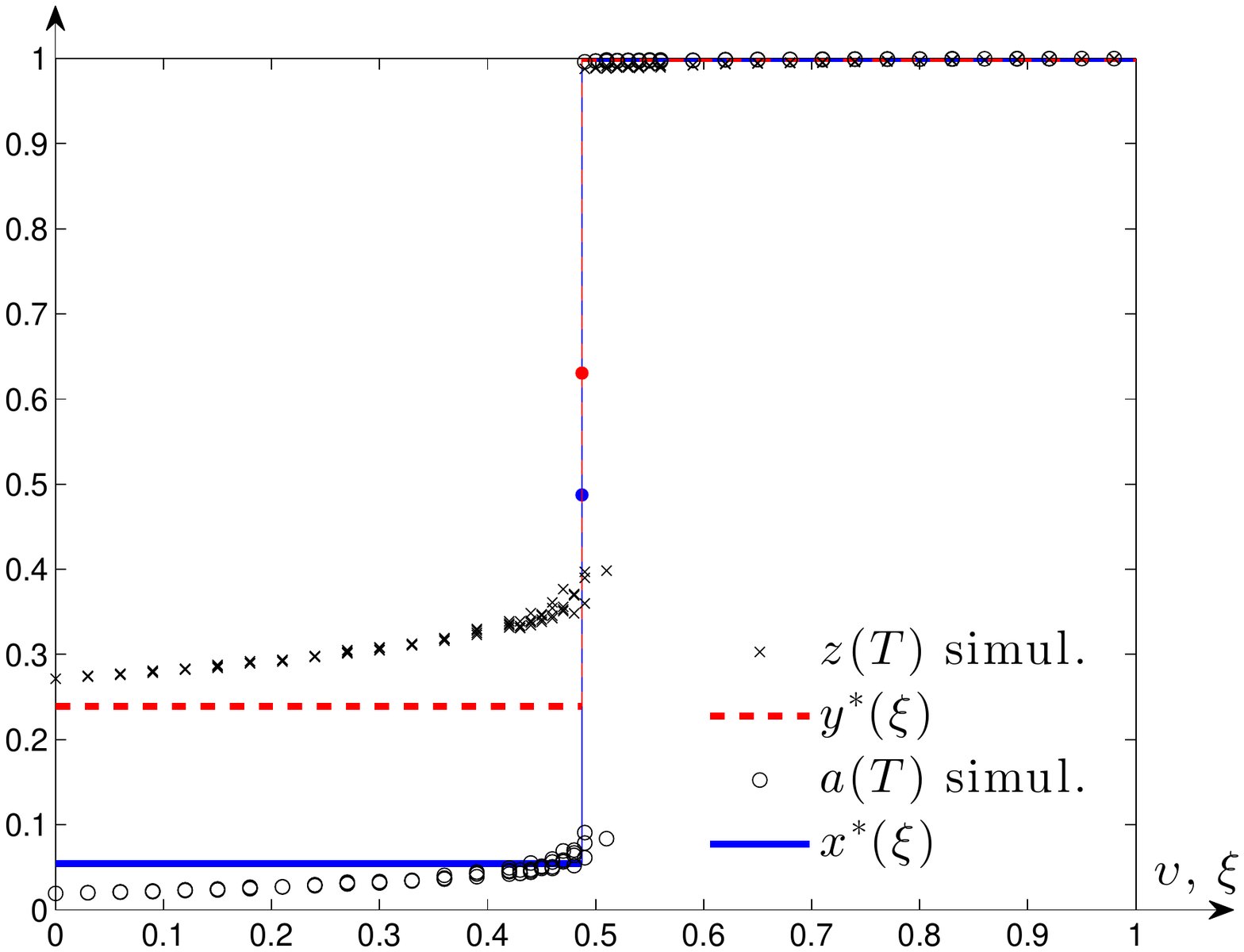}%
	\caption{\label{fig:epinions-sigle-phi-asy} The left plot reports the functions $\phi(x)$ (solid blue) and $\psi(x)$ (dashed red), corresponding to the \url{Epinions.com} network where each agent $i$ is endowed by  the thresholds $\rho_i = \lceil \frac12 \kappa_i \rceil$.
	The right plot compares the values reached by the simulations at the time horizon $T=100$, for various value of the fraction $\ups$ of initially active nodes, with the asymptotic activation predicted by the recursion initialized with $\xi = \ups$. The black crosses represent $z(T)$, i.e. the fraction of state-$1$ adopters, to be compared with the recursion limits $y^*(\xi)$ in dashed red. The black circles represent $a(T)$, i.e. the fraction of links pointing to state-$1$ adopters,  to be compared with the recursion limits $x^*(\xi)$ in dashed red. Near the discontinuity, predicted in $\xi^* \approx 0.487$ and well matched by the simulation, the starting values of $\ups$ are more dense.}
\end{figure}

%

In the first group of simulation we assume every agent $i$ in the network shares the same common normalized threshold $\Theta_i = 0.500$ and hence node $i$'s threshold is $\rho_i = \lceil \frac12 \kappa_i \rceil$. 
This assumption corresponds to the distribution function $F(\theta) = h(\theta - \frac12)$. Having set a common normalized threshold and given $\ups \in [0,1]$, each simulation consists in choosing a random initial state assignment, such that  exactly a fraction $\ups$ of the nodes has $\sigma_i = 1$, and in computing the LTM dynamic until a prearranged time horizon $T$. 
Given $\ups$ we typically repeat the simulation a few times and compare them with the dynamic predicted with the recursion, initialized with $\xi = \ups$. 
Figure~\ref{fig:epinions-sigle-dyn} reports an example of the simulations with $\ups = 0.475$: the left plot contains the simulated dynamics $a(t)$ to be compared with the recursion's state dynamic $x(t)$; the right plot contains the corresponding simulated fraction of active nodes, $z(t)$, to be compared with the recursion's output dynamic $y(t)$. The recursion captures the qualitative behavior of the simulations. 
The left plot of Figure~\ref{fig:epinions-sigle-phi-asy} represents the recursion's functions $\phi(x)$ and $\psi(x)$ corresponding to this group of simulations. The right plot of the same figure compares the asymptotic activation predicted by the recursion with a few actual simulations, obtained for  various $\ups$ and computed up to a time horizon $T = 100$.  
The fractions of state-$1$ adopters $z(T)$ shall be compared with the recursion's output asymptotic value $y^*(\xi)$, while the corresponding fraction of links pointing at state-$1$ adopters, $a(T)$, shall be compared with the recursion's state asymptotic value $x^*(\xi)$. The discontinuity, predicted in $\xi^* \approx 0.487$ by the recursion, is well matched by the simulation. Before the discontinuity, the simulated values of $z(T)$ are higher that the limit $y^*(\xi)$, showing an increasing trend. The same trend is present in the corresponding values of $a(T)$, that are however closer to the limit $x^*(\xi)$. After the discontinuity, simulations and limits agree to  value one.
\end{example}

\begin{example} 
In the second group of simulation we allow the normalized thresholds to take two different values: to 40\% of the nodes we assign $\frac14$ as normalized threshold; the remaining 60\% of nodes gets $\frac34$. The choice corresponds to the cumulative distribution of the normalized threshold  
 $F(\theta) = \frac{4}{10}h(\theta - \frac14) + \frac{6}{10}h(\theta - \frac34)$. Figure~\ref{fig:epinions-2-phi-asy} contains the results of these simulations. 
The left plot represents the functions $\phi(x)$ and $\psi(x)$ corresponding to the thresholds chosen: the recursion predicts the presence of two discontinuities in the asymptotic activation for the LTM, for the seed values $\xi^*_1 \approx   0.241$ and $\xi^*_2 \approx  0.7482$, corresponding to the unstable equilibria of $\phi(x)$. 
The right plot compares the predicted asymptotic activation with the  simulations, obtained for various $\ups$ and computed up to time $T = 100$.
The fractions of state-$1$ adopters $z(T)$ shall be compared with the recursion's output asymptotic value $y^*(\xi)$, while the corresponding fraction of links pointing at state-$1$ adopters, $a(T)$, is nearly superimposed to recursion's state asymptotic value $x^*(\xi)$.
The plot shows a good agreement between $a(T)$ and $x^*(\xi)$, while $z(T)$ seems a bit underestimated by $y^*(\xi)$. 
The values $z(T)$ and $a(T)$ of one simulation with $\ups = 0.310$ settled to a smaller limit, compatible with those obtained for $\ups < 0.270$. Apart from this simulation, the discontinuities are matched well.
Also in this group of simulations, the values of $z(T)$ (and less markedly also those of $a(T)$) show an increasing trend with respect to the fraction of initially active nodes $\ups$, a feature not expected by the comparison with the recursion limits.



\begin{figure}[t]
	\centering
	\includegraphics[trim={\figtrimla} {\figtrimba} {\figtrimra} {\figtrimta},clip, width={\figwidthduo},
		keepaspectratio=true]{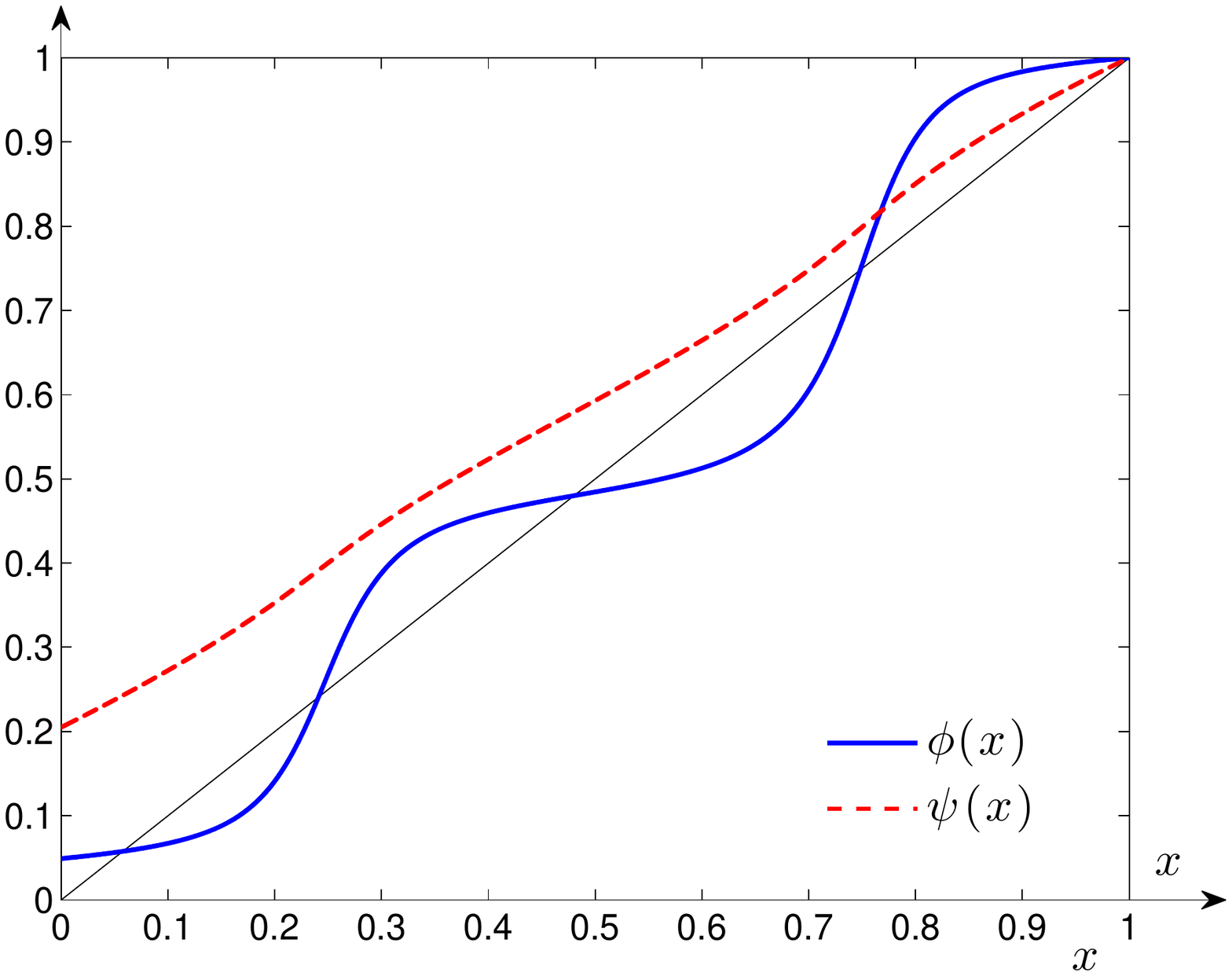}%
	\hspace{1cm}%
	\includegraphics[trim={\figtrimla} {\figtrimba} {\figtrimra} {\figtrimta},clip, width={\figwidthduo},
		keepaspectratio=true]{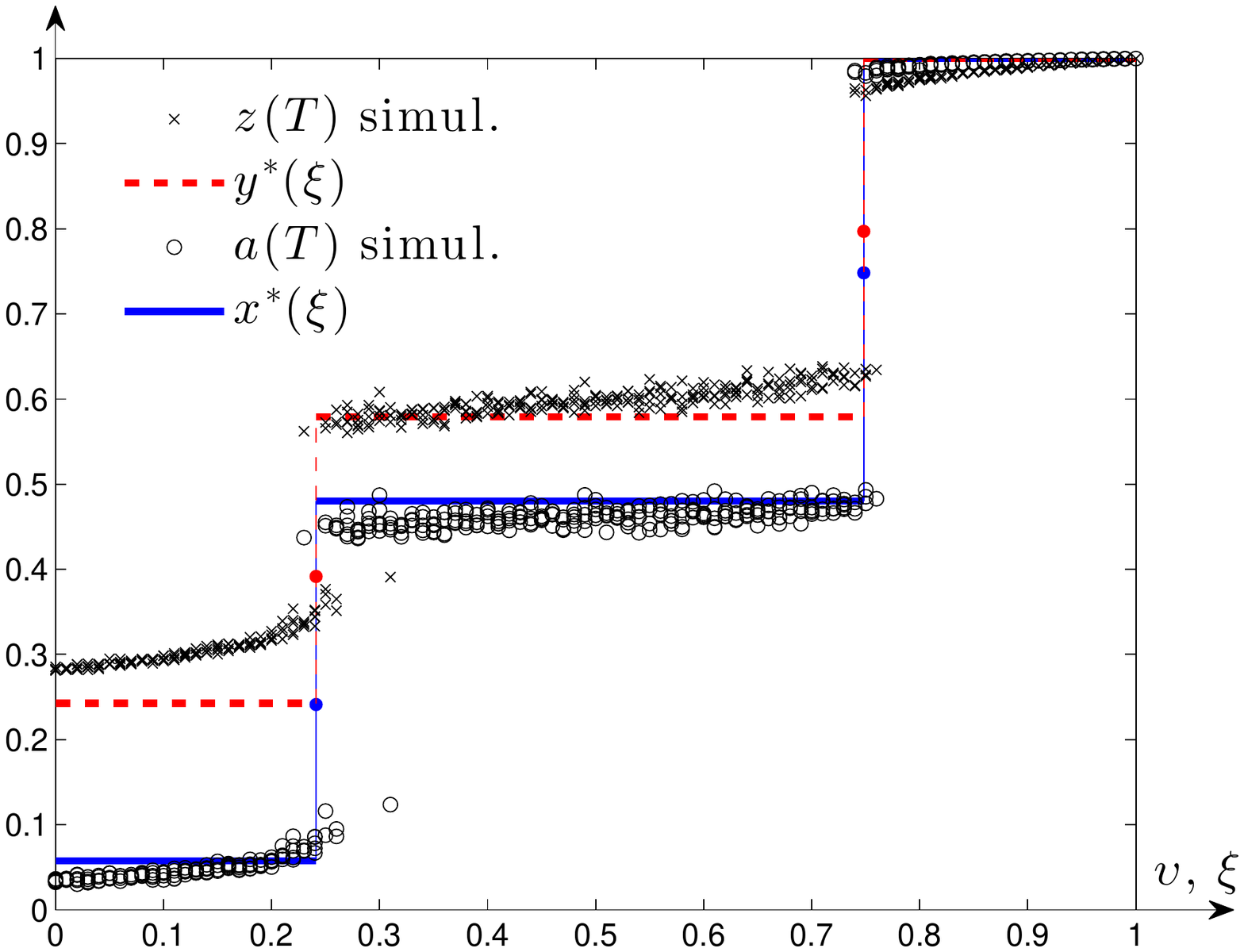}%
	\caption{\label{fig:epinions-2-phi-asy}
	The left plot contains the functions $\phi(x)$ (solid blue) and $\psi(x)$ (dashed red), corresponding to the \url{Epinions.com} network where  40\% of the nodes is endowed by the normalized threshold $\frac14$ and the remaining 60\% by $\frac34$.
	The right plot compares the values reached by the simulations at the time horizon $T=100$, for various value of the fraction $\ups$ of initially active nodes, with the	asymptotic activation predicted by the recursion initialized with $\xi = \ups$. 	
	The black crosses represent the fraction of state-$1$ adopters $z(T)$, to be compared with the recursion limits $y^*(\xi)$ in dashed red. 
	The black circles represent the fraction of links pointing to state-$1$ adopters $a(T)$, to be compared with the recursion limits $x^*(\xi)$ in dashed red. 
	The predicted limits $y^*(\xi)$ and $x^*(\xi)$ are discontinuous for  $\xi^*_1 \approx   0.241$ and $\xi^*_2 \approx  0.7482$, which are the two unstable equilibria of $\phi(x)$ (cf. left plot). The discontinuities are well matched by the simulation, except for one point obtained with $\ups = 0.310$.
	Apart from the matching the discontinuities, the simulated values show a slowly increasing trend, unexpected from the recursion limits. 	}
\end{figure}
\end{example}

\begin{example}
Finally, we present a group of simulations where we allow the normalized thresholds to take three different values: 30\% of the nodes are endowed by the normalized threshold $\frac15$, 30\%  by  $\frac12$ and the remaining 40\% by  $\frac45$.
The corresponding cumulative distribution is 
 $F(\theta) = \frac{3}{10}h(\theta - \frac15) + \frac{3}{10}h(\theta - \frac12) +  \frac{4}{10}h(\theta - \frac45)$. 
The left plot of Figure~\ref{fig:epinions-3-phi-dyn} represents the functions $\phi(x)$ and $\psi(x)$: the function $\phi(x)$ has seven fixed points, while the convexities of $\psi(x)$ are minimal. 
The right plot of the same figure contains the dynamic of the fraction of state-1 node $z(t)$, starting from a fraction  $\ups = 0.700$ of initial adopters.  
The simulations are compared with the output $y(t)$ of the recursion: the majority of the simulations tend to a limit just above the recursion, while showing a ripple with period two; three simulations tend to a smaller value. 
With this choice of normalized thresholds, the recursion predicts the presence of three discontinuities in the asymptotic activation for the LTM, in $\xi^*_1 \approx   0.201$, $\xi^*_2 \approx  0.509$ and $\xi^*_3 \approx 0.789$. The comparison between recursion and simulation is available in Figure~\ref{fig:epinions-3-asy}. The left plot represent the simulated values of $z(T)$ at time $T=100$, for various $\ups$, compared with the limit $y^*(\xi)$ obtained assuming $\xi = \ups$ as initial condition for the recursion.
The right plot represents the corresponding simulated values of $a(T)$, at $T=100$, to be compared with the recursion's limit $x^*(\xi)$.
Some of the simulations in Figure~\ref{fig:epinions-3-asy} settle to values smaller than the those of the points having similar $\ups$, values  that might be expected from a smaller  initial condition. 


\begin{figure}
	\centering
	\includegraphics[trim={\figtrimla} {\figtrimba} {\figtrimra} {\figtrimta},clip, width={\figwidthduo},
		keepaspectratio=true]{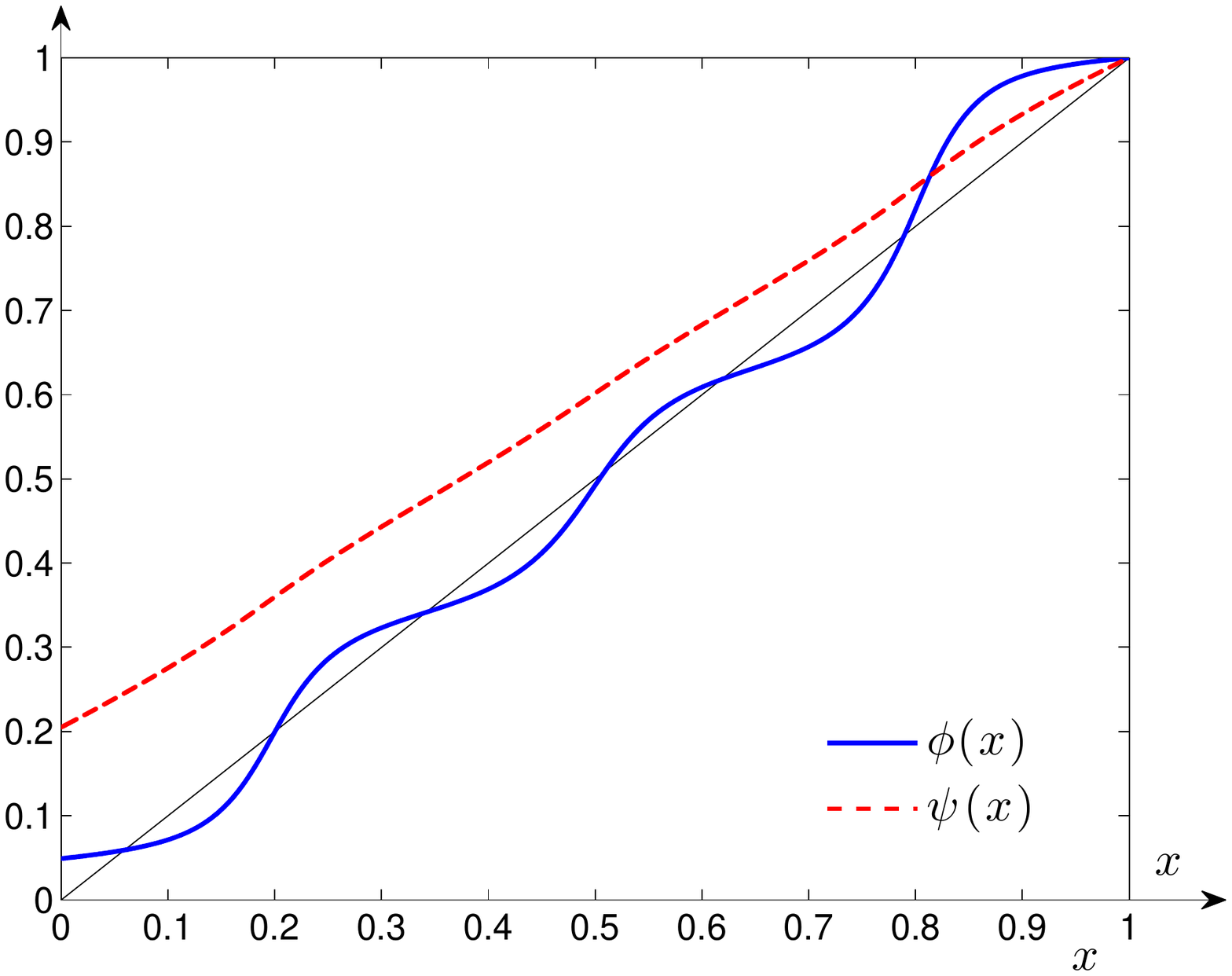}%
	\hspace{1cm}%
	\includegraphics[trim={\figtrimla} {\figtrimba} {\figtrimra} {\figtrimta},clip, width={\figwidthduo},
		keepaspectratio=true]{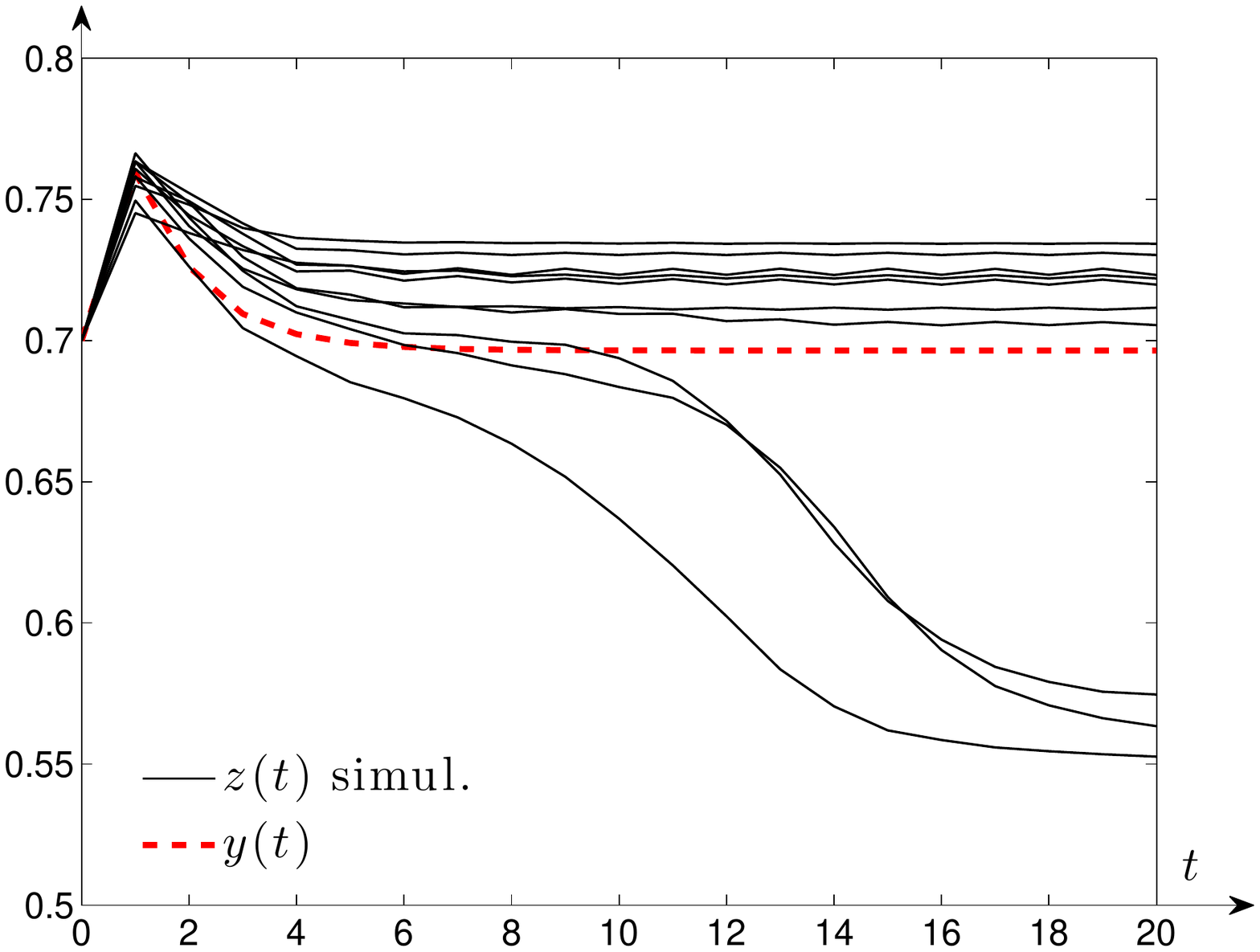}%
	\caption{\label{fig:epinions-3-phi-dyn}
	The left plot contains the functions $\phi(x)$ (solid blue) and $\psi(x)$ (dashed red), corresponding to the \url{Epinions.com} network where  30\% of the nodes is endowed by the normalized threshold $\frac15$, 30\%  by  $\frac12$ and the remaining 40\% by  $\frac45$.
	The right plot contains a few simulations (thin black lines) of the dynamic of the fraction of state-$1$ adopters, $z(t)$, starting from a fraction  $\ups = 0.700$ of nodes with state one. The majority of the simulations tend to a limit just above the recursion, while showing a ripple with period two; three simulations tend to a smaller value.
The simulations are compared with the output $y(t)$ of the recursion (dashed red line).  Note that the vertical axis has been zoomed to the interval $[0.5, 0.8]$.
	}
\end{figure}

\begin{figure}
	\centering
	\includegraphics[trim={\figtrimla} {\figtrimba} {\figtrimra} {\figtrimta},clip, width={\figwidthduo},
		keepaspectratio=true]{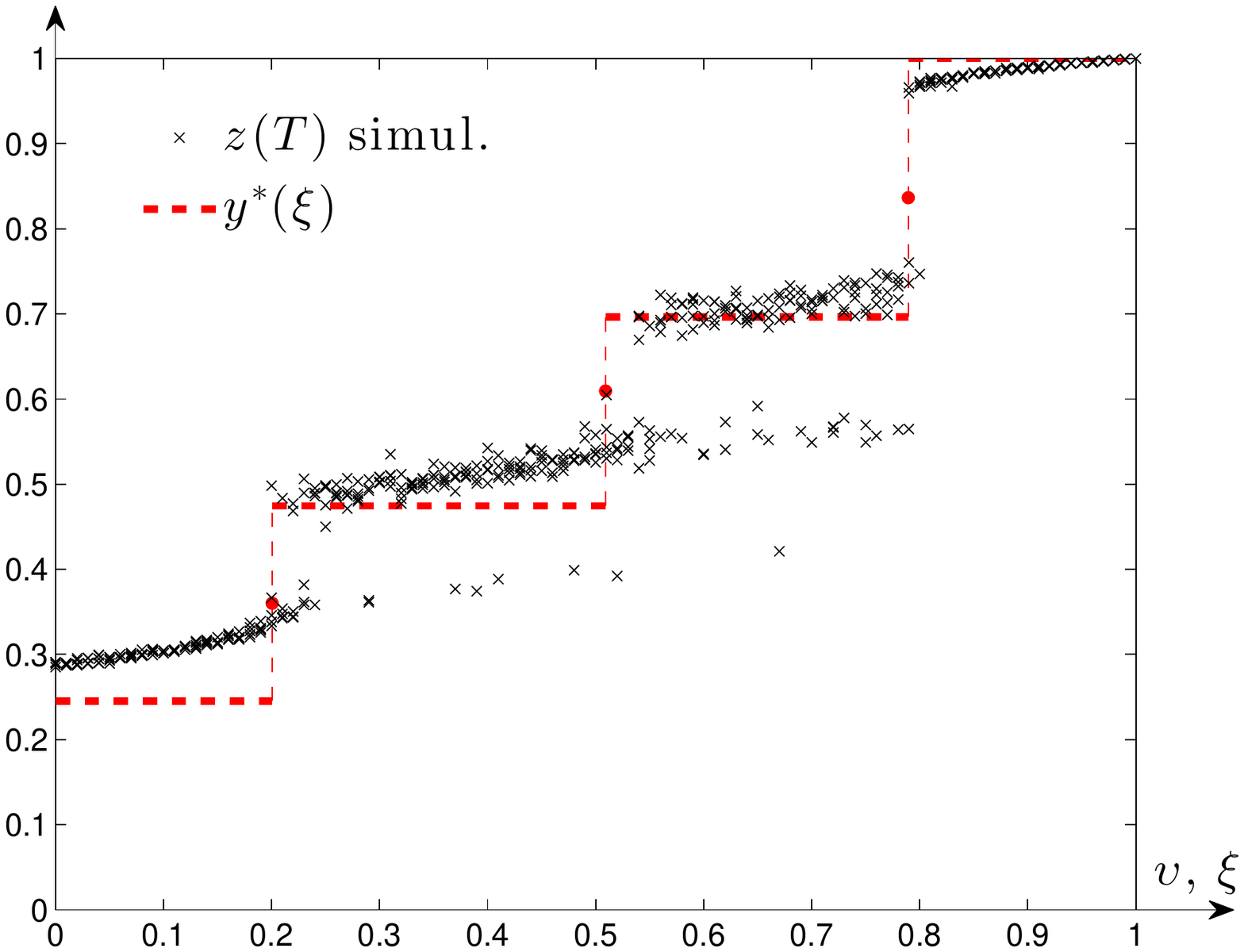}%
	\hspace{1cm}%
	\includegraphics[trim={\figtrimla} {\figtrimba} {\figtrimra} {\figtrimta},clip, width={\figwidthduo},
		keepaspectratio=true]{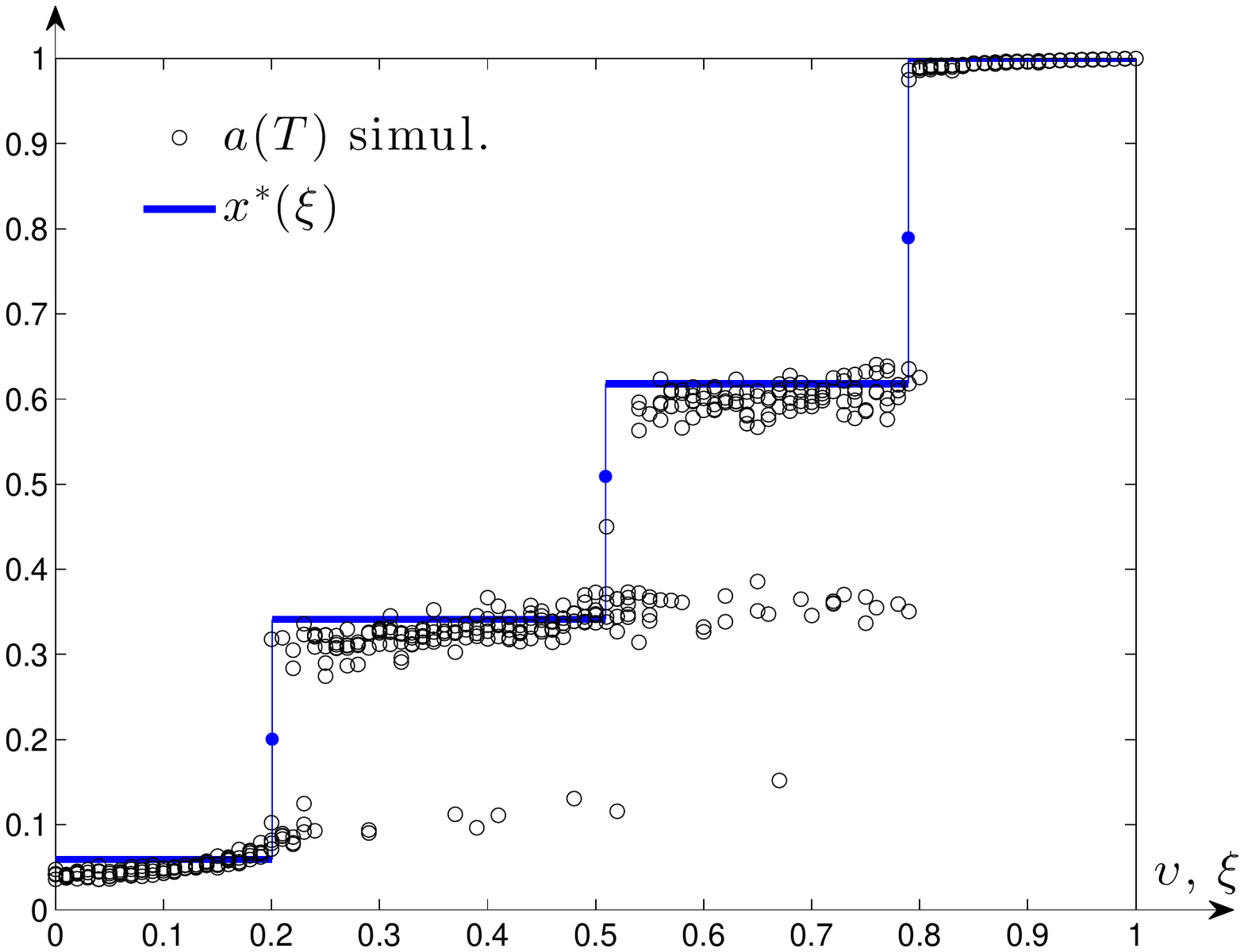}
	\caption{\label{fig:epinions-3-asy} 
	Comparison between the predicted asymptotic activation and the actual simulations, on the \url{Epinions.com} graph where the nodes where  30\% of the nodes is endowed by the normalized threshold $\frac15$, 30\%  by  $\frac12$ and the remaining 40\% by  $\frac45$.
	 The left plot contains the simulated values of the fraction of state-$1$ adopters $z(T)$ at time $T=100$ (black crosses), for various $\ups$, compared with the limit $y^*(\xi)$ (red dashed line)  of the recursion output, obtained assuming $\xi = \ups$. 
	 The right plot represents the values of the fraction of links pointing to state-$1$ adopters, $a(t)$, for the corresponding simulations, to be compared with the asymptotic value of the recursion's state $x^*(\xi)$. 
We observe that some simulation settle to values that are smaller than  those of the points having similar $\ups$.
With this choice of normalized thresholds, the limits $y^*(\xi)$ and  $x^*(\xi)$ have three discontinuities, in $\xi^*_1 \approx   0.201$, $\xi^*_2 \approx  0.509$ and $\xi^*_3 \approx 0.789$.	}
\end{figure}
\end{example}

\subsection{Comments on the results}

The simulations of the LTM using the topology of the social network \url{Epinions.com} give some interesting insights. Overall, the prediction obtained with the recursion are in good agreement with the simulations. 

A few differences between the simulations and the predictions remain. 
In several simulations we observed that the dynamics of $z(t)$ and $a(t)$ presents a periodic variation, with period two, superimposed to the settling value. 
In particular during the last example, the supposed settling value of a few simulations, evaluated with  $z(T)$ and $a(T)$ at time $T=100$, seemed to smaller that what expected from similar simulations. 
Finally, for increasing initial condition $\ups$, the values  $z(T)$ and $a(T)$ seem to have an increasing trend besides the expected jumps, and the value $z(T)$ seem to be a little but consistently underestimated by the recursion.

There are few possible explanations for these behaviors . 
The social community used in this simulations is based on an online network. Even though it does not have a ``geographical'' origin, it is not a completely random network. The recursion does not take into account any possible community structure of the network, which may play a role in the periodic behavior observed as well as in the increasing trend of the settling values. Furthermore, the presence of a few nodes with extremely high in and out-degree, is able to influence the single simulations, depending on the initial state and threshold assigned to that node. This may contribute to the explanation of the presence of points with smaller-than-expected settling value. 

These hypothesis require further work on the \url{Epinions.com} topology to be verified. 
The simulations however show a good predicting ability by the recursion: the discontinuities in the settling values of the simulations match well with the jumps in the recursion's limits

%
%
%
%
%
%
%
%

\section{Conclusion}
\label{sec:conclusion}
In this paper, we have studied the Linear Threshold Model (LTM) of cascades in large-scale networks. We have shown that, for all but an asymptotically vanishing fraction of networks with given degree and threshold statistics, the fraction of state-$1$ adopters in the LTM can be approximated by the output of a one-dimensional nonlinear recursion. We have also analyzed the asymptotic behavior of this recursion both for homogeneous and heterogeneous networks. Our results apply both to the original LTM and to the Progressive LTM on the configuration model ensemble of directed networks and for the Progressive LTM (but not to the original LTM) on the configuration model ensemble of undirected networks. Numerical simulations run on the actual topology of the social network \url{Epinions.com} confirm the validity of our theoretical result in predicting the behavior of the LTM in actual large-scale networks.  Ongoing work is concerned with the use of the obtained one-dimensional recursion for the design of feedback control policies for the LTM -- see \cite[ch. 4]{wsR:2015:phd-thesis} and \cite{RCF:2016:mtns} for preliminary results. 

\section*{Acknowledgments}
The authors wish to acknowledge Prof. Julien Hendrickx of Universit\'e catholique de Louvain for many valuable comments on the second author's PhD thesis \cite{wsR:2015:phd-thesis}.



\providecommand{\bysame}{\leavevmode\hbox to3em{\hrulefill}\thinspace}
\providecommand{\MR}{\relax\ifhmode\unskip\space\fi MR }
\providecommand{\MRhref}[2]{%
  \href{http://www.ams.org/mathscinet-getitem?mr=#1}{#2}
}
\providecommand{\href}[2]{#2}

\end{document}